\documentclass[11pt,a4paper]{article}
\usepackage{mathpazo}
\usepackage{hyperref}
\usepackage{xcolor}
\usepackage{setspace}
\usepackage[utf8]{inputenc}
\usepackage[T1]{fontenc}

\usepackage{complexity}
\usepackage{amsmath,amsthm,amssymb}
\usepackage{fullpage}
\hypersetup{colorlinks=true,linkcolor=black,citecolor=blue}

\usepackage{xspace}
\usepackage{nicefrac}
\usepackage{enumerate}
\usepackage{combelow}
\usepackage{tikz}
\usetikzlibrary{arrows.meta}
\usepackage{multirow}
\makeatletter
\usepackage{bbm}

\newtheorem{theorem}{Theorem}[section]

\newtheorem{lemma}[theorem]{Lemma}
\newtheorem{definition}[theorem]{Definition}%[section]
\newtheorem{claim}[theorem]{Claim}%[section]
\newtheorem{proposition}[theorem]{Proposition}
\newtheorem{question}[theorem]{Open Question}

\makeatletter
\newcommand\slarge{\@setfontsize\slarge{14}{14}}
\makeatother

  \theoremstyle{definition}
\usepackage{framed}
\usepackage{dsfont}
 \usepackage{color,bm}
\usepackage{colortbl}
\usetikzlibrary{decorations.pathreplacing}
\usetikzlibrary{patterns}\tikzset{
    ncbar angle/.initial=90,
    ncbar/.style={
        to path=(\tikztostart)
        -- ($(\tikztostart)!#1!\pgfkeysvalueof{/tikz/ncbar angle}:(\tikztotarget)$)
        -- ($(\tikztotarget)!($(\tikztostart)!#1!\pgfkeysvalueof{/tikz/ncbar angle}:(\tikztotarget)$)!\pgfkeysvalueof{/tikz/ncbar angle}:(\tikztostart)$)
        -- (\tikztotarget)
    },
    ncbar/.default=0.5cm,
}

\tikzset{square left brace/.style={ncbar=0.5cm}}
\tikzset{square right brace/.style={ncbar=-0.5cm}}
\tikzset{small square left brace/.style={ncbar=0.2cm}}
\tikzset{small square right brace/.style={ncbar=-0.2cm}}
\usetikzlibrary{matrix}

\date{}

\renewcommand{\Mod}[1]{\ (\mathrm{mod}\ #1)}
\renewcommand{\W}{\textsf{W}}
\renewcommand{\XP}{\textsf{XP}}
\renewcommand{\NP}{\textsf{NP}}
\renewcommand{\P}{\textsf{P}}
\renewcommand{\RP}{\textsf{RP}}
\renewcommand{\FPT}{\textsf{FPT}}
\newcommand{\MDP}{{\textsf{MDP}}}
\newcommand{\SVP}{{\textsf{SVP}}}
\newcommand{\NCP}{{\textsf{NCP}}}
\renewcommand{\CVP}{{\textsf{NVP}}}
\newcommand{\col}{{\textsf{col}}}
\newcommand{\bA}{{\mathbf{A}}}
\newcommand{\bB}{{\mathbf{B}}}
\newcommand{\bx}{{\mathbf{x}}}
\newcommand{\by}{{\mathbf{y}}}
\newcommand{\bw}{{\mathbf{w}}}
\newcommand{\bu}{{\mathbf{u}}}
\newcommand{\vone}{\mathbf{1}}
\newcommand{\vzero}{\mathbf{0}}
\newcommand{\etal}{\text{et al.}}
\newcommand{\geqs}{\geqslant}
\newcommand{\leqs}{\leqslant}
\renewcommand{\le}{\leqs}
\renewcommand{\leq}{\leqs}
\renewcommand{\ge}{\geqs}
\renewcommand{\geq}{\geqs}
\newcommand{\Z}{\mathbb{Z}}
\newcommand{\F}{\mathbb{F}_2}
\newcommand{\N}{\mathbb{N}}
\renewcommand{\poly}{\text{poly}}
\newcommand{\bL}{{\mathbf{L}}}
\newcommand{\bT}{{\mathbf{T}}}
\newcommand{\bs}{{\mathbf{s}}}
\newcommand{\bz}{{\mathbf{z}}}

\newcommand{\bP}{{\mathbf{P}}}
\newcommand{\bv}{{\mathbf{v}}}
\newcommand{\br}{{\mathbf{r}}}
\newcommand{\tbx}{{\tilde{\bx}}}
\newcommand{\bzero}{{\mathbf{0}}}

\newcommand{\Id}{{\rm Id}}
\newcommand{\cB}{\mathcal{B}}
\newcommand{\gapvec}{\textsc{GapMLD}}

\newcommand{\mdp}{{\textsc{GapMDP}}}
\newcommand{\sncp}{{\textsc{GapSNC}}}
\newcommand{\snvp}{\textsc{GapNVP}}
\newcommand{\svp}{\textsc{GapSVP}}

\newcommand{\LSDC}{\textsf{LSDC}}
\newcommand{\LDC}{\textsf{LDC}}
\newcommand{\cA}{\mathcal{A}}

\newcommand{\kvec}{\textsc{MLD}}

\newcommand{\cC}{\mathcal{C}}

\newcommand{\cF}{\mathcal{F}}

\newcommand{\cW}{\mathcal{W}}

\newcommand{\be}{{\mathbf{e}}}

\newcommand{\LDS}{{\textsf{LDS}}}
\newcommand{\gapLDS}{{\textsc{GapLDS}}}
\newcommand{\gapbsmd}{{\textsc{GapBSMD}}}
\newcommand{\biclique}{\textsc{Biclique}}

\renewcommand{\polylog}{\text{polylog}}

\newcommand{\odds}{\textsc{OddSet}\xspace}
\newcommand{\gapodds}{\textsc{GapOddSet}\xspace}
\newcommand{\cS}{\mathcal{S}}

\newcommand{\nc}{\newcommand}
\nc{\rnc}{\renewcommand} \nc{\nev}{\newenvironment}

\newlength{\probwidth}
\nc{\prob}[3][9]{
\begin{center}
  \normalfont\fbox{
   \begin{tabular}[t]{
     rp{#1cm}}\textit{Instance:}&#2. \\
     \textit{Problem:}&#3
   \end{tabular}}
\end{center}}

\nc{\pprob}[4][9]{
\begin{center}
   \normalfont\fbox{
    \begin{tabular}[t]{
     rp{#1cm}}\textit{Instance:}&#2. \\
     \textit{Parameter:}&#3. \\
     \textit{Problem:}&#4
   \end{tabular}}
\end{center}}

\nc{\nprob}[4][9]{
\begin{center}
  \normalfont\fbox{

\addtolength{\probwidth}{#1cm}\parbox{\probwidth}{\textsc{#2}\\\hspace*{1.5em}
     \begin{tabular}[t]{
      rp{#1cm}}\textit{Instance:}&#3. \\
      \textit{Problem:}&#4
     \end{tabular}}}
\end{center}}

\nc{\npprob}[5][9]{
\begin{center}
  \normalfont\fbox{

\addtolength{\probwidth}{#1cm}\parbox{\probwidth}{\textsc{#2}\\\hspace*{1.5em}
    \begin{tabular}[t]{
     rp{#1cm}}\textit{Instance:}&#3. \\
     \textit{Parameter:}&#4. \\
     \textit{Problem:}&#5
    \end{tabular}}}
\end{center}}

\nc{\nppxrob}[5][9]{ \normalfont\fbox{

\addtolength{\probwidth}{#1cm}\parbox{\probwidth}{\textsc{#2}\\\hspace*{1.5em}
   \begin{tabular}[t]{
    rp{#1cm}}\textit{Instance:}&#3. \\
    \textit{Parameter:}&#4. \\
    \textit{Problem:}&#5
   \end{tabular}}}}

\nc{\nppprob}[5][4]{
\begin{center}
  \normalfont\fbox{

\addtolength{\probwidth}{#1cm}\parbox{\probwidth}{\textsc{#2}\\\hspace*{1.5em}
    \begin{tabular}[t]{
     rp{#1cm}}\textit{Instance:}&#3. \\
     \textit{Parameter:}&#4. \\
     \textit{Problem:}&#5
    \end{tabular}}}
\end{center}}

\nc{\dotcup}{\;\dot\cup\;}

\makeatother

\providecommand{\definitionname}{Definition}

\title{Parameterized Intractability of Even Set and \\Shortest Vector Problem\vspace{0.5cm}}

\author{Arnab Bhattacharyya\footnote{National University of Singapore. Email: \texttt{arnabb@nus.edu.sg}}\and 
\'Edouard Bonnet\footnote{CNRS, \'Ecole Normale Sup\'erieure de Lyon, Universit\'e Claude Bernard Lyon 1, LIP UMR5668. Email: \texttt{edouard.bonnet@ens-lyon.fr}}\and 
L\'{a}szl\'{o} Egri\footnote{Indiana State University, Terre Haute. Email: \texttt{laszlo.egri@mail.mcgill.ca}}\and 
Suprovat Ghoshal\footnote{Indian Institute of Science. Email:  \texttt{suprovat@iisc.ac.in}} \vspace{0.3cm} \and 
Karthik C.\ S.\footnote{Weizmann Institute of Science. Email: \texttt{karthik.srikanta@weizmann.ac.il}} \and 
Bingkai Lin\footnote{Nanjing University. Email:  \texttt{lin@nju.edu.cn}}\and 
Pasin Manurangsi\footnote{University of California, Berkeley. Email: \texttt{pasin@berkeley.edu}}\and 
D\'{a}niel Marx\footnote{Institute for Computer Science and Control, Hungarian Academy of Sciences. Email: \texttt{dmarx@cs.bme.hu}}}

\begin{document}

	\maketitle
\begin{abstract}
  The \emph{$k$-Even Set} problem is a parameterized variant of the \emph{Minimum Distance Problem} of linear codes over $\mathbb F_2$, which can be stated as follows: given a generator matrix $\mathbf A$ and an integer $k$, determine whether the code generated by $\mathbf A$ has distance at most $k$, or in other words, whether there is a nonzero vector $\bx$ such that $\mathbf A \bx$ has at most $k$ nonzero coordinates.  The question of whether $k$-Even Set is fixed parameter tractable (FPT) parameterized by the distance $k$ has been repeatedly raised in literature; in fact, it is one of the few remaining open questions from the seminal book of Downey and Fellows (1999). In this work, we show that $k$-Even Set is \W[1]-hard under randomized reductions.\vspace{0.1cm}

We also consider the parameterized  \emph{$k$-Shortest Vector Problem (SVP)}, in which we are given a lattice whose basis vectors are integral and an integer $k$, and the goal is to determine whether the norm of the shortest vector (in the $\ell_p$ norm for some fixed $p$) is at most $k$. Similar to $k$-Even Set, understanding the complexity of this problem is also a long-standing open question in the field of Parameterized Complexity. We show that, for any $p > 1$, $k$-SVP is \W[1]-hard to approximate (under randomized reductions) to some constant factor. %Furthermore, for the case of $p = 2$, the inapproximability factor can be amplified to any constant. 
\end{abstract}

%%% Local Variables:
%%% mode: latex
%%% TeX-master: "JACM submission"
%%% End:

\clearpage	

\section{Introduction} \label{sec:intro}

The study of error-correcting codes gives rise to many interesting computational problems. One of the most fundamental among these is the problem of computing the distance of a linear code. In this problem, which is commonly referred to as the \emph{Minimum Distance Problem (\MDP)}, we are given as input a generator matrix $\bA \in \mathbb{F}_2^{n \times m}$ of a binary\footnote{Note that \MDP\ can be defined over larger fields as well; we discuss more about this in Section~\ref{sec:open}.} linear code and an integer $k$. The goal is to determine whether the code has distance at most $k$. Recall that the distance of a linear code is $\underset{\vzero \ne \bx \in \mathbb{F}_2^m}{\min}\ \|\bA\bx\|_0$ where $\| \cdot \|_0$ denote the 0-norm (aka the Hamming norm).

To see the fundamental nature of \MDP, let us discuss two other natural ways of arriving at (equivalent formulations of) this problem. \MDP\ has the following well-known dual formulation: the minimum distance of the code generated by  $\bA \in \mathbb{F}_2^{n \times m}$  can be also expressed as 
$\underset{\vzero \ne \by \in \mathbb{F}_2^m, \bA^{\perp}\by=0}{\min}\ \|\by\|_0$, where $\bA^{\perp}$ is the orthogonal complement of $\bA$. In other words, finding the minimum distance of the code is equivalent to finding the minimum set of linearly dependent vectors among the column vectors of $\bA^{\perp}$. Thus, \MDP\ is equivalent to the the \textsf{Linear Dependent Set} problem on vectors over $\mathbb{F}_2$ or,  using the language of matroid theory, solving the \textsf{Shortest Circuit} problem on a represented binary matroid.

One can arrive at a more combinatorial formulation of the problem as a variant of the \textsf{Hitting Set} problem. Given a set system $\cS$ over a universe $U$ and an integer $k$, the \textsf{Hitting Set} problem asks for a $k$-element subset $X$ of $U$ such that $|S\cap X|\neq 0$ for every $S\in \cS$. \textsf{Hitting Set} is a basic combinatorial optimization that is well studied (often under the dual formulation \textsf{Set Cover}) in the approximation algorithms and the parameterized complexity literature. More restrictive versions of the problem are the \textsf{Exact Hitting Set} problem, where we require $|S\cap X|=1$, and the \textsf{Odd Set} problem, where we require $|S\cap X|$ to be odd. By analogy, we can define the \textsf{Even Set} problem, where we require $|S\cap X|$ to be even, but in this case we need to add the requirement $X\neq \emptyset$ to avoid the trivial solution. While \textsf{Hitting Set}, \textsf{Exact Hitting Set}, and \textsf{Odd Set} are known to be \W[1]-hard parameterized by $k$, \textsf{Even Set} can be easily seen to be equivalent to \MDP\ (in the dual formulation of \MDP, the rows of $\bA^{\perp}$ play the same role as the sets in $\cS$).

The study of this problem dates back to at least 1978 when Berlekamp \etal~\cite{BMT78} conjectured that it is \NP-hard. This conjecture remained open for almost two decades until it was positively resolved by Vardy~\cite{Var97a,Var97b}. Later, Dumer \etal~\cite{DMS03} strengthened this intractability result by showing that even \emph{approximately} computing the minimum distance of the code is hard. Specifically, they showed that, unless $\NP = \RP$, no polynomial time algorithm can distinguish between a code with distance at most $k$ and one whose distance is greater than $\gamma \cdot k$ for any constant $\gamma \geqs 1$. Furthermore, under stronger assumptions, the ratio can be improved to superconstants and even almost polynomial. Dumer \etal's result has been subsequently derandomized by Cheng and Wan~\cite{CW12} and further simplified by Austrin and Khot~\cite{AK14} and Micciancio~\cite{Mic14}.

While the aforementioned intractability results rule out not only efficient algorithms but also efficient approximation algorithms for \MDP, there is another popular technique in coping with \NP-hardness of problems which is not yet ruled out by the known results: \emph{parameterization}.

In parameterized problems, part of the input is an integer that is designated as the parameter of the problem, and the goal is now not to find a polynomial time algorithm but a \emph{fixed parameter tractable (FPT)} algorithm. This is an algorithm whose running time can be upper bounded by some (computable) function of the parameter in addition to some polynomial in the input length. Specifically, for \MDP, its parameterized variant\footnote{Throughout Sections~\ref{sec:intro}  and \ref{sec:overview}, for a computational problem $\Pi$, we denote its parameterized variant by $k$-$\Pi$, where $k$ is the parameter of the problem.} $k$-\MDP\ has $k$ as the parameter and the question is whether there exists an algorithm that can decide if the code generated by $\bA$ has distance at most $k$ in time $f(k) \cdot \poly(mn)$ where $f$ can be any computable function that depends only on $k$.

Note that $k$-\MDP\ can be solved in $n^{O(k)}$ time. This can be easily seen in the dual formulation, as we can enumerate through all vectors $\by$ with Hamming norm at most $k$ and check whether $\bA^\perp \by = \bzero$. In Parameterized Complexity language, this means that $k$-\MDP\ belongs to the class \XP.

The parameterized complexity of $k$-\MDP\ was first questioned by Downey \etal~\cite{DFVW99}, who showed that parameterized variants of several other coding-theoretic problems, including the Nearest Codeword Problem and the Nearest Vector Problem\footnote{The Nearest Vector Problem is also referred to in the literature as the Closest Vector Problem.} which we will discuss in more details in Section~\ref{sec:ncp-cvp},  are $\W[1]$-hard. Thereby, assuming the widely believed $\W[1] \ne \FPT$ hypothesis, these problems are rendered intractable from the parameterized perspective. Unfortunately, Downey \etal~fell short of proving such hardness for $k$-\MDP\ and left it as an open problem: %More broadly speaking, we would be happy to know the answer to the following question:

\begin{question}\label{openq1}
Is $k$-\MDP\ fixed parameter tractable? %Or, can we show its parameterized intractability under any reasonable complexity theoretic assumption?
\end{question}

Although almost two decades have passed, the above question remains unresolved to this day, despite receiving significant attention from the community. In particular, the problem was listed as an open question in the seminal 1999 book of Downey and Fellows~\cite{DF99} and has been reiterated numerous times over the years~\cite{DGMS07,FGMS12,GKS12,DBLP:conf/birthday/FominM12,DF13,CFJKLMPS14,CFKLMPPS15,BGGS16,CFHW17,Maj17}. This problem is one of the few questions that remained open from the original list of Downey and Fellows~\cite{DF99}. In fact, in their second book~\cite{DF13}, Downey and Fellows even include this problem as one of the six\footnote{So far, two of the six problems have been resolved: that of parameterized complexity of $k$-Biclique~\cite{Lin15} and that of parameterized approximability of $k$-Dominating Set~\cite{
KLM18}.} ``most infamous'' open questions in the area of Parameterized Complexity.

Another question posted in Downey \etal's work~\cite{DFVW99} that remains open is the parameterized \emph{Shortest Vector Problem ($k$-\SVP)} in lattices. The input of $k$-\SVP\ (in the $\ell_p$ norm) is an integer $k \in \N$ and a matrix $\bA \in \Z^{n \times m}$ representing the basis of a lattice, and we want to determine whether the shortest (non-zero) vector in the lattice has length at most $k$, i.e., whether $\underset{\vzero \ne \bx \in \Z^{m}}{\min}\ \|\bA\bx\|_p \leqs k$. Again, $k$ is the parameter of the problem. It should also be noted here that, similar to~\cite{DFVW99}, we require the basis of the lattice to be integer valued, which is sometimes not enforced in literature (e.g.~\cite{VEB,Ajt98}). This is because, if $\bA$ is allowed to be any matrix in $\mathbb{R}^{n \times m}$, then parameterization is meaningless because we can simply scale $\bA$ down by a large multiplicative factor.

The (non-parameterized) Shortest Vector Problem (\SVP) has been intensively studied, motivated partly due to the fact that both algorithms and hardness results for the problem have numerous applications. Specifically, the celebrated LLL algorithm for \SVP~\cite{LLL} can be used to factor rational polynomials, and to solve integer programming (parameterized by the number of unknowns)~\cite{Len83} and many other computational number-theoretic problems (see e.g.~\cite{NV10}). Furthermore, the hardness of (approximating) \SVP\ has been used as the basis of several cryptographic constructions~\cite{Ajt98,AD97,Reg03,Reg05}. Since these topics are out of scope of our paper, we refer the interested readers to the following surveys for more details:~\cite{Reg06,MR-survey,NV10,Reg10}.

On the computational hardness side of the problem, van Emde-Boas~\cite{VEB} was the first to show that \SVP\ is \NP-hard for the $\ell_{\infty}$ norm, but left open the question of whether \SVP\ on the $\ell_p$ norm for $1 \leqs p < \infty$ is \NP-hard. It was not until a decade and a half later that Ajtai~\cite{Ajt96} showed, under a randomized reduction, that \SVP\ for the $\ell_2$ norm is also \NP-hard; in fact, Ajtai's hardness result holds not only for exact algorithms but also for $(1 + o(1))$-approximation algorithms as well. The $o(1)$ term in the inapproximability ratio was then improved in a subsequent work of Cai and Nerurkar~\cite{CN99}. Finally, Micciancio~\cite{Mic00} managed to achieve a factor that is bounded away from one. Specifically, Micciancio~\cite{Mic00} showed (again under randomized reductions) that \SVP\ on the $\ell_p$ norm is \NP-hard to approximate to within a factor of $\sqrt[p]{2}$ for every $1 \leqs p < \infty$. Khot~\cite{Khot05} later improved the ratio to any constant, and even to $2^{\log^{1/2 - \varepsilon}(nm)}$ under a stronger assumption. Haviv and Regev~\cite{HR07} subsequently simplified the gap amplification step of Khot and, in the process, improved the ratio to almost polynomial. We note that both Khot's and Haviv-Regev reductions are also randomized and it is still open to find a deterministic \NP-hardness reduction for \SVP\ in the $\ell_p$ norms for $1 \leqs p < \infty$ (see~\cite{Mic12}); we emphasize here that such a reduction is not known even for the exact (not approximate) version of the problem. For the $\ell_\infty$ norm, the following stronger result due to Dinur~\cite{Din02} is known: \SVP\ in the $\ell_\infty$ norm is \NP-hard to approximate to within $n^{\Omega(1/\log \log n)}$ factor (under a \emph{deterministic} reduction).

Very recently, fine-grained studies of \SVP\ have been initiated~\cite{BGS17,AD17}. The authors of~\cite{BGS17,AD17} showed that \SVP\ for any $\ell_p$ norm cannot be solved (or even approximated to some constant strictly greater than one) in subexponential time assuming the existence of a certain family of lattices\footnote{This additional assumption is only needed for $1 \leqs p \leqs 2$. For $p > 2$, their hardness is conditional only on Gap-ETH.} and the (randomized) \emph{Gap Exponential Time Hypothesis (Gap-ETH)}~\cite{D16,MR16}, which states that no randomized subexponential time algorithm can distinguish between a satisfiable 3-CNF formula and one which is only 0.99-satisfiable. 

As with \MDP, Downey \etal~\cite{DFVW99} were the first to question the parameterized tractability of $k$-\SVP\ (for the $\ell_2$ norm). Once again, Downey and Fellows included $k$-\SVP\ as one of the open problems in both of their books~\cite{DF99,DF13}. As with Open Question~\ref{openq1}, this question remains unresolved to this day:

\begin{question}\label{openq2}
Is $k$-\SVP\ fixed parameter tractable?
\end{question}

We remark here that, similar to $k$-\MDP, $k$-\SVP\ also belongs to \XP, as we can enumerate over all vectors with norm at most $k$ and check whether it belongs to the given lattice. There are only $(mn)^{O(k^p)}$ such vectors, and the lattice membership of a given vector can be decided in polynomial time (e.g., see page 18 \cite{Micc12}). Hence, this is an $(nm)^{O(k^p)}$-time algorithm for $k$-\SVP. 

\subsection{Our Results}\label{sec:results}

The main result of this paper is a resolution to the previously mentioned Open~Questions~\ref{openq1}~and~\ref{openq2}: more specifically, we prove that $k$-\MDP\ and $k$-\SVP\  (on $\ell_p$ norm for any $p > 1$) are $\W[1]$-hard under randomized reductions. In fact, our result is stronger than stated here as we rule out not only exact FPT algorithms but also FPT \emph{approximation} algorithms as well. In particular, all of our results use the $\W[1]$-hardness of approximating the $k$-\biclique\ problem recently proved by Lin \cite{Lin15} as a starting point.

With this in mind, we can state our results starting with the parameterized intractability of $k$-\MDP, more concretely (but still informally), as follows:

\begin{theorem}[Informal; see Theorem~\ref{thm:MDPmain}]
For any $\gamma \geqs 1$, given input $(\bA, k) \in \F^{n \times m} \times \N$, it is \W[1]-hard (under randomized reductions) to distinguish between
\begin{itemize}
\item the distance of the code generated by $\bA$ is at most $k$ , and,
\item the distance of the code generated by $\bA$ is more than $\gamma \cdot k$.
\end{itemize} 
\end{theorem}

Notice that our above result rules out FPT approximation algorithms with \emph{any} constant approximation ratio for $k$-\MDP. In contrast, we can only prove FPT inapproximability with \emph{some} constant ratio for $k$-\SVP\ in $\ell_p$ norm for $p > 1$. These are stated more precisely below.

\begin{theorem}[Informal; see Theorem~\ref{thm:SVPmain}]
For any $p > 1$, there exists a constant $\gamma_p > 1$ such that given input $(\bA, k) \in \Z^{n \times m} \times \N$, it is \W[1]-hard (under randomized reductions) to distinguish between
\begin{itemize}
\item the $\ell_p$ norm of the shortest vector of the lattice generated by $\bA$ is $\leqs k$, and,
\item the $\ell_p$ norm of the shortest vector of the lattice generated by $\bA$ is $> \gamma_p \cdot k$.
\end{itemize} 
\end{theorem}

We remark that our results do not yield hardness for \SVP\ in the $\ell_1$ norm and this remains an interesting open question. Section~\ref{sec:open} contains discussion on this problem. We also note that, for Theorem~\ref{thm:SVPmain} and onwards, we are only concerned with $p \ne \infty$; this is because, for $p = \infty$, the problem is \NP-hard to approximate even when $k = 1$~\cite{VEB}!

\subsubsection{Nearest Codeword Problem and Nearest Vector Problem}
\label{sec:ncp-cvp}

As we shall see in Section~\ref{sec:overview}, our proof proceeds by first showing FPT hardness of approximation of the non-homogeneous variants of $k$-\MDP\ and $k$-\SVP\ called the \emph{$k$-Nearest Codeword Problem} ($k$-\NCP) and the \emph{$k$-Nearest Vector Problem } ($k$-\CVP) respectively. For both $k$-\NCP\ and $k$-\CVP, we are given a target vector $\by$ (in $\F^n$ and $\Z^n$, respectively) in addition to $(\bA, k)$, and the goal is to find whether there is any $\bx$ (in $\F^m$ and $\Z^m$, respectively) such that the (Hamming and $\ell_p$, respectively) norm of $\bA\bx - \by$ is at most $k$. Note that their homogeneous counterparts, namely $k$-\MDP\ and $k$-\SVP, explicitly require the coefficient vector $\bx$ to be non-zero, and hence they cannot be interpreted as special cases of $k$-\NCP and $k$-\CVP~respectively.

As an intermediate step of our proof, we show that the $k$-\NCP\ and $k$-\CVP\ problems are hard to approximate\footnote{While our $k$-\MDP\ result only applies for $\F$, it is not hard to see that our intermediate reduction for $k$-\NCP\ actually applies for every finite field $\mathbb{F}_q$ too.}
%\bnote{Comment on this footnote: The inapproximability of $k$-\LDS\ implies inapproximability of $k$-\MDP\ over large fields. Can we show inapproximability of $k$-\MDP\ for any finite field?} 
(see Theorem~\ref{thm:MLDmain} and Theorem~\ref{thm:CVPmain} respectively). This should be compared to~Downey et al.~\cite{DFVW99}, in which the authors show that both problems are $\W[1]$-hard to solve exactly. Therefore our inapproximability result significantly improves on their work to rule out any $\polylog(k)$ factor FPT-approximation algorithm (assuming $\W[1]\neq \FPT$) and are also the first inapproximability results for these problems.

We end this section by remarking that the computational complexity of both (non-parameterized) \NCP\ and \CVP\ are also thoroughly studied (see e.g.~\cite{Mic01,DKRS03,Ste93,ABSS97,GMSS99} in addition to the references for \MDP\ and \SVP), and indeed the inapproximability results of these two problems form the basis of hardness of approximation for \MDP\ and \SVP. We would like to emphasize that while \W[1]-hardness results were known for $k$-\NCP\ and $k$-\CVP, it does not seem easy to transfer them to \W[1]-hardness results for $k$-\MDP\ and $k$-\SVP; we really need parameterized \emph{inapproximability} results for $k$-\NCP\ and $k$-\CVP\ to be able to transfer them to (slightly weaker) inapproximability results for $k$-\MDP\ and $k$-\SVP. There are other parameterized problems that resisted all efforts at proving hardness so far, and we believe that it may be the case for these problems as well that building a chain of inapproximability results is more feasible than building a chain of W[1]-hardness results.

\subsection{Organization of the paper}

In the next section, we give an overview of our reductions and proofs. After that, in Section~\ref{sec:prelim}, we define additional notation and preliminaries needed to fully formalize our proofs. 
In Section~\ref{sec:lin-dep}, we show the inapproximability of {\em $k$-Linear Dependent Set} ($k$-\LDS), a problem naturally arising from linear algebra, that would be used as the base step for all future inapproximability results in this paper. 
In Section~\ref{sec:csp-to-gapvec}  we show the inapproximability of $k$-\NCP. Next, in Section~\ref{sec:main-reduction}, we establish the constant inapproximability of $k$-\MDP. Section~\ref{sec:svp} provides the  inapproximability of $k$-\CVP\ and $k$-\SVP. Finally, in Section~\ref{sec:open}, we conclude with a few open questions and research directions.
%%% Local Variables:
%%% mode: latex
%%% TeX-master: "JACM submission"
%%% End:

\section{Proof Overview}
\label{sec:overview}

In the non-parameterized setting, all the aforementioned inapproximability results for both \MDP\ and \SVP\ are shown in two steps: first, one proves the inapproximability of their inhomogeneous counterparts (i.e. \NCP\ and \CVP), and then reduces them to \MDP\ and \SVP. We follow this general outline. That is, we first show, that both $k$-\NCP\ and $k$-\CVP\ are \W[1]-hard to approximate. Then, we reduce $k$-\NCP\ and $k$-\CVP\ to $k$-\MDP\ and $k$-\SVP\ respectively. In this second step, we employ an adaptation of Dumer \etal's reduction~\cite{DMS03} for $k$-\MDP\ and Khot's reduction~\cite{Khot05} for $k$-\SVP. While the latter reduction works almost immediately in the parameterized regime, there are several technical challenges in adapting Dumer \etal's reduction to our setting. The remainder of this section is devoted to presenting all of our reductions and to highlight such technical challenges and changes in comparison with the non-parameterized setting.

As mentioned before, the starting point of all the hardness results in this paper is the $\W[1]$-hardness of approximating the $k$-\biclique\ problem. In Subsection~\ref{sec:LDSoverview}, we show a gap-retaining reduction from the gap $k$-\biclique\ problem to gap $k$-Linear Dependent Set (referred to hereafter as $k$-\LDS), an intermediate problem that we introduce which might be of independent interest.
We show a gap-retaining reduction from gap $k$-\LDS\ to gap $k$-\NCP\ in  Subsection~\ref{sec:NCPoverview}, and then a randomized reduction from gap $k$-\NCP\ to $k$-\MDP\ in  Subsection~\ref{sec:MDPoverview}. Finally, in Subsection~\ref{sec:SVPoverview}, we show a gap-retaining reduction from gap $k$-\LDS\ to gap $k$-\CVP, and then a randomized reduction from gap $k$-\CVP\ to $k$-\SVP. 

In the next subsection, we first give an overview of Dumer \etal's reduction~\cite{DMS03} and highlight the challenges in extending their reduction to the parameterized setting, following which we give a sketch of the various steps involved in the actual reduction to $k$-\MDP. %We end this section by also giving a brief description of our reduction to $k$-\SVP.

\subsection{The Dumer-Micciancio-Sudan reduction}

We start this subsection by describing the Dumer \etal's (henceforth DMS) reduction~\cite{DMS03}. The starting point of the DMS reduction is the \NP-hardness of approximating \NCP\ to any constant factor \cite{ABSS97}. Let us recall that in \NCP\ we are given a matrix $\bA\in\F^{n\times m}$, an integer $k$, and a target vector $\by\in\F^n$, and the goal is to determine whether there is any $\bx\in\F^m$ such that $\|\bA\bx - \by\|_0$ is at most $k$.  Arora et al. \cite{ABSS97} shows that for any constant $\gamma\geqs 1$, it is \NP-hard to distinguish the case when there exists $\bx$ such that $\|\bA\bx - \by\|_0\leqs k$ from the case when for all $\bx$ we have that $\|\bA\bx - \by\|_0> \gamma k$.

Dumer et al.\ introduce the notion of ``locally dense codes'' to enable a gadget reduction from $\NCP$ to $\MDP$.
Informally, a locally dense code is a linear code $\bL$ with minimum distance $d$
admitting a ball $\cB(\bs, r)$ centered at $\bs$ of radius\footnote{Note that for the ball to contain more than a single codeword, we must have $r\geqs \nicefrac{d}{2}$.} $r < d$ and containing a large (exponential in the dimension) number of codewords. Moreover, for the gadget reduction to \MDP\ to go through, we require not only the knowledge of the code, but also the center $\bs$ and  a linear transformation $\bT$ used to index the codewords in $\cB(\bs, r)$, i.e., $\bT$ maps $\cB(\bs, r)\cap \bL$ onto a smaller subspace. Given an instance $(\bA,\by,k)$ of \NCP, and a  locally dense code $(\bL,\bT,\bs)$ whose parameters (such as dimension and distance) we will fix later, Dumer et al.\ build the following matrix:
\begin{align}
	\bB =
	\begin{bmatrix}
		\bA \bT \bL & - \by \\
		\vdots & \vdots \\
		\bA \bT \bL & - \by \\
		\bL & - \bs \\
		\vdots & \vdots \\
		\bL & - \bs
	\end{bmatrix}
	\begin{tikzpicture}[baseline=8.5ex]
	\draw [decorate,decoration={brace,amplitude=10pt}]
	(-1.5,1.3) -- (-1.5,0) node [black,midway,xshift=35pt]{$b$ copies};
	\draw [decorate,decoration={brace,amplitude=10pt}]
	(-1.5,2.85) -- (-1.5,1.55) node [black,midway,xshift=35pt]{$a$ copies};
	\end{tikzpicture}
	,
	\label{eq:LDCGadget}
\end{align}

where $a,b$ are some appropriately chosen positive integers. If there exists $\bx$ such that $\|\bA\bx-\by\|_0\leqs k$ then consider $\bz'$ such that $\bT\bL \bz' = \bx$ (we choose the parameters of $(\bL,\bT,\bs)$, in particular the dimensions of $\bL$ and $\bT$ such that all these computations are valid). Let $\bz=\bz'\circ 1$, and note that $\|\bB\bz\|_0=a\|\bA\bx-\by\|_0+b\|\bL\bz-\bs\|_0\leqs ak+br$. In other words, if $(\bA,\by,k)$ is a YES instance of \NCP\ then $(\bB,ak+br)$ is a YES instance of \MDP. On the other hand if we had that for all $\bx$, the norm of $\|\bA\bx-\by\|_0$ is more than $\gamma k$ for some constant\footnote{Note that in the described reduction, we need the inapproximability  of \NCP\ to a factor greater than two, even to just reduce to the \emph{exact} version of \MDP.} $\gamma>2$, then it is possible to show that for all $\bz$ we have that $\|\bB\bz\|_0>\gamma' (ak+br)$ for any $\gamma'<\frac{2\gamma}{2+\gamma}$. The proof is based on a case analysis of the last coordinate of $\bz$. If that coordinate is 0, then, since $\bL$ is a code of distance $d$, we have $\|\bB\bz\|_0 \geqs bd > \gamma'(ak+br)$; if that coordinate is 1, then the assumption that $(\bA,\by,k)$ is a NO instance of \NCP\ implies that $\|\bB\bz\|_0 > a\gamma k >\gamma'(ak+br)$. Note that this gives an inapproximability for \MDP\ of ratio $\gamma' < 2$; this gap is then further amplified by a simple tensoring procedure.

We note that Dumer et al.\ were not able to find a deterministic construction of locally dense code with all of the above described properties. Specifically, they gave an efficient deterministic construction of a code $\bL$, but only gave a randomized algorithm that finds a linear transformation $\bT$ and a center $\bs$ w.h.p. Therefore, their hardness result relies on the assumption that $\NP\neq\RP$, instead of the more standard $\NP\neq \P$ assumption. Later, Cheng and Wan \cite{CW12} and Micciancio \cite{Mic14} provided constructions for such (families of) locally dense codes with an explicit center, and thus showed the constant ratio inapproximability of \MDP\ under the assumption of $\NP\neq\P$.

Trying to follow the DMS reduction in order to show the parameterized intractability of $k$-\MDP, we face the following three immediate obstacles. First, there is no inapproximability result known for  $k$-\NCP, for any constant factor greater than 1. Note that to use the DMS reduction, we need the parameterized inapproximability of $k$-\NCP, for an approximation factor which is greater than two. Second, the construction of locally dense codes of Dumer et al. only works when the distance is linear in the block length (which is a function of the size of the input). However, we need codes whose distance are bounded above by a function of the parameter of the problem (and not dependent on the input size). This is because the DMS reduction converts an instance $(\bA,\by,k)$ of $k$-\NCP\ to an instance $(\bB,ak+br)$ of $(ak+br)$-\MDP, and for this reduction to be an FPT reduction, we need $ak+br$ to be a function only depending on $k$, i.e., $d$, the distance of the code $\bL$ (which is at most $2r$), must be a function only of $k$. Third, recall that the DMS reduction needs to identify the vectors in the ball $\cB(\bs, r) \cap \bL$ with all the potential solutions of $k$-\NCP. Notice that the number of vectors in the ball is at most $(nm)^{O(r)}$ but the number of potential solutions of $k$-\NCP\ is exponential in $m$ (i.e. all $\bx \in \F^m$). However, this is impossible since $r \leqs d$ is bounded above by a function of $k$!

We overcome the first obstacle by proving the inapproximability of $k$-\NCP\ upto poly-logarithmic factors under $\W[1]\neq \FPT$ (see Subsection~\ref{sec:LDSoverview}). Note that in order to follow the DMS reduction, it suffices to just show the inapproximability of $k$-\NCP\ for some constant factor greater than 2; nonetheless the hardness of approximating $k$-\NCP\ up to poly-logarithmic factors is of independent interest. 

We overcome the third obstacle by introducing an intermediate problem in the DMS reduction, which we call the \emph{sparse nearest codeword problem}. The sparse nearest codeword problem is a promise problem which differs from $k$-\NCP\ in the following way: the objective here considers the distance of the target vector $\by$ to the nearest codeword $\bA\bx$ as well as the Hamming weight of the coefficient vector $\bx$ which realizes the nearest codeword. We show the inapproximability of the sparse nearest codeword problem (See Subsection~\ref{sec:NCPoverview}).
%%%old: 
%%%We show that $k$-\NCP\ can be reduced to the sparse nearest codeword problem (See Subsection~\ref{sec:NCPoverview}).

Finally, we overcome the second obstacle by introducing a variant of locally dense codes, which we call \emph{locally suffix dense codes}. Roughly speaking, we show that any systematic code which nears the sphere-packing bound (aka Hamming bound) in the high rate regime is a locally suffix dense code. Then we follow the DMS reduction with the new ingredient of locally suffix dense codes (replacing locally dense codes) to reduce the sparse nearest codeword problem to $k$-\MDP.

The full reduction goes through several intermediate steps, which we will describe in more detail in the coming subsections. The high-level summary of these steps is also provided in Figure~\ref{fig:overviewMDC}. Throughout this section, for any gap problem, if we do not specify the gap in the subscript, then it implies that the gap can be any arbitrary constant (or even super constant).
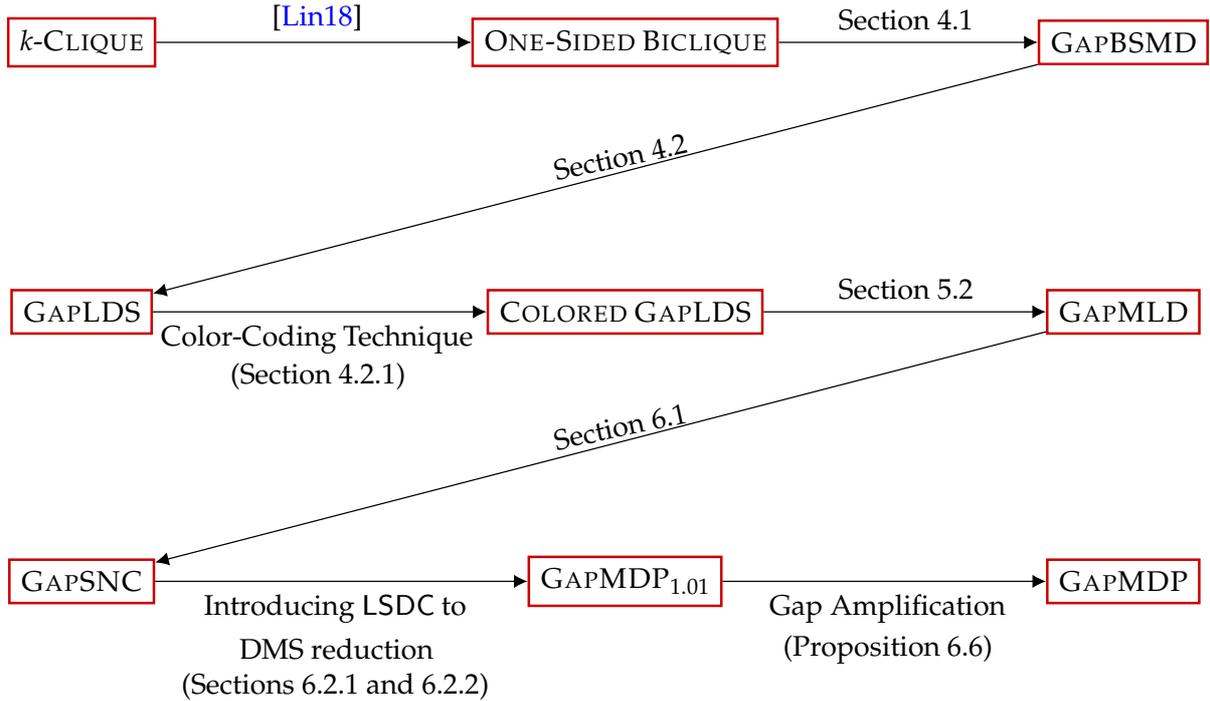
\begin{figure}[h!]
	\centering
	\resizebox{\textwidth}{!}{\begin{tikzpicture}

\node (gapLDS) [draw=red!80!black,thick] at (1.5, 0) {\footnotesize\gapLDS};

\node (colLDS) [draw=red!80!black,thick] at (7.5, 0) {\footnotesize\textsc{Colored} \gapLDS};
\node (gapmld) [draw=red!80!black,thick] at (13, 0) {\footnotesize\gapvec};

\node (kclique) [draw=red!80!black,thick] at (1.5, 3) {\footnotesize$k$-\textsc{Clique}};

\node (biclique) [draw=red!80!black,thick] at (7.5, 3) {\footnotesize\textsc{One-Sided Biclique}};
\node (gapBSMD) [draw=red!80!black,thick] at (13, 3) {\footnotesize\gapbsmd};

\node (gapsnc) [draw=red!80!black,thick] at (1.5, -3) {\footnotesize\sncp};

\node (gapmdc) [draw=red!80!black,thick] at (7.5, -3) {\footnotesize\mdp$_{1.01}$};
\node (ampgapmdc) [draw=red!80!black,thick] at (13, -3) {\footnotesize\mdp};

\draw [-{Latex[length=1.5mm, width=1.5mm]}] (kclique) -- (biclique);
\draw [-{Latex[length=1.5mm, width=1.5mm]}] (biclique) -- (gapBSMD);

\draw [-{Latex[length=1.5mm, width=1.5mm]}] (gapBSMD) -- (gapLDS);
\draw [-{Latex[length=1.5mm, width=1.5mm]}] (gapLDS) -- (colLDS);
\draw [-{Latex[length=1.5mm, width=1.5mm]}] (colLDS) -- (gapmld);

\draw [-{Latex[length=1.5mm, width=1.5mm]}] (gapmld) -- (gapsnc);
\draw [-{Latex[length=1.5mm, width=1.5mm]}] (gapsnc) -- (gapmdc);
\draw [-{Latex[length=1.5mm, width=1.5mm]}] (gapmdc) -- (ampgapmdc);

\node [above, align=center] at (4.1, 3) {\footnotesize  \cite{Lin15}};
\node [above, align=center] at (10.6, 3) {\footnotesize Section~\ref{sec:biclique-to-bsmd}};

\node [above, align=center,rotate=12] at (7.5, 1.48) {\footnotesize Section~\ref{sec:bsmd-to-LDS}};

\node [above, align=center] at (4.1, -0.6) {\footnotesize Color-Coding Technique};
\node [above, align=center] at (4.11, -1) {\footnotesize (Section~\ref{sec:colorLDS})};

\node [above, align=center] at (10.6, 0) {\footnotesize Section~\ref{sec:full-result}};

\node [above, align=center,rotate=12] at (7.5, -1.52) {\footnotesize Section~\ref{sec:mld-to-sncp}};

\node [above, align=center] at (4.3, -3.6) {\footnotesize Introducing \LSDC\ to};
\node [above, align=center] at (4.31, -4) {\footnotesize DMS reduction};
\node [above, align=center] at (4.3, -4.45) {\footnotesize (Sections~\ref{sec:dense-code} and \ref{sec:main-red})};

\node [above, align=center] at (10.4,-3.6) {\footnotesize Gap Amplification};
\node [above, align=center] at (10.4, -4.05) {\footnotesize (Proposition~\ref{prop:gap-amplification})};

\end{tikzpicture}}
	\caption{The figure provides an overview of our reduction from the canonical \W[1]-complete  $k$-Clique problem to the parameterized Minimum Distance problem. Our starting point is the gap one-sided biclique problem which is now known to be \W[1]-hard from Lin's work \cite{Lin15}. Based on the hardness of approximating the one-sided biclique problem, we  obtain the constant inapproximability of a different graph problem, namely the bipartite subgraph with minimum degree problem (\gapbsmd); see Section~\ref{sec:biclique-to-bsmd} for details.
		Next, we reduce \gapbsmd\ to the gap linear dependent set problem ($\gapLDS$) in Section~\ref{sec:bsmd-to-LDS}, and then use standard color-coding techniques in Section~\ref{sec:colorLDS} to obtain the constant inapproximability of a colored version of \gapLDS~over fields of non-constant size. 
In Section~\ref{sec:full-result}, we reduce the aforementioned colored version of \gapLDS\ to the \gapvec~problem over $\F$, and thus rule out constant approximation parameterized algorithms for \NCP.
		Via a simple reduction from \gapvec, in Section~\ref{sec:mld-to-sncp} we obtain the constant parameterized inapproximability of \sncp. In Section~\ref{sec:dense-code}, we formally introduce locally suffix dense codes and show how to efficiently (but probabilistically) construct them. These codes are then used in Section~\ref{sec:main-red} to obtain the parameterized innapproximability of $\mdp_{1.01}$. The final step is a known gap amplification by tensoring (Proposition~\ref{prop:gap-amplification}).
	} \label{fig:overviewMDC}
\end{figure}

\subsection{Parameterized Inapproximability of $k$-\LDS}\label{sec:LDSoverview}

To prove the inapproximability of \MDP\ we first consider its dual problem \LDS. Given a set $\bA$ of $n$ vectors over a finite field $\mathbb{F}_q$ and an integer $k$, the goal of $k$-\LDS\ problem is to decide if there are $k$ vectors in $\bA$ that are linearly dependent. The gap version of this problem (\gapLDS) is to distinguish the case when  there are $k$ vectors in  $\bA$ that are linearly dependent from the case when any $\gamma k$ ($\gamma\ge 1$) vectors in $\bA$ are linearly independent.
As briefly touched upon in the introduction, $k$-\LDS\ is closely related to $k$-\MDP:  one might think of $\bA$ as a matrix  in $\mathbb{F}_q^{n\times m}$ and the goal of $k$-\LDS\  is to find a vector  $\by\in\mathbb{F}_q^m$ with $\| \by \|_0\le k$ and $\bA\by=\mathbf{0}$, then $k$-\LDS\ is a yes-instance if and only if $\underset{\mathbf{0}\neq \bx\in\mathbb{F}_q^{n'}}{\min}\|\bA^{\perp}\bx\|_0\le k$, where $\bA^{\perp}\in\mathbb{F}_q^{m\times n'}$ is a  matrix with maximum number  of linearly independent column vectors such that $\bA\bA^{\perp}=\mathbf{0}$.
Note that the parameterized inapproximability of $k$-\MDP\  follows by the parameterized intractability of \gapLDS\ over the binary field. 

However, we cannot prove the hardness of \gapLDS\ over the binary field directly. Instead, we tackle this problem in three steps. Our first step is to show the parameterized intractability of \gapLDS\ over large fields by giving a reduction from the \textsc{One-Sided Biclique} problem to  \gapLDS. %In the second step, we prove the inapproximability of $k$-\NCP\ using the hardness of of \gapLDS. Note that in this step, we reduce the underlying field size to two. Finally, we present  a  reduction from gap $k$-\NCP\ to gap $k$-\MDP\ and obtain the parameterized inapproximability of $k$-\MDP.

It will be more convenient to view the inapproximability result of \textsc{One-Sided Biclique} from~\cite{Lin15} as a hardness of the following problem which we call \textsc{Bipartite Subgraph with Minimum Degree (BSMD)}: given a bipartite graph $G$ and positive integers $s, h$ with $s \leq h$, find smallest (in terms of edges) non-empty subgraph of $G$ such that every left vertex of the subgraph has degree at least $h$ and every right vertex has degree at least $s$. Here the parameter is $s + h$. The gap version of \textsc{BSMD}, called $\gapbsmd_{\gamma}$, is to distinguish between (i) the YES case in which $G$ contains a complete bipartite graph with $s$ vertices on the left and $h$ on the right (which satisfies the property with $hs$ edges) and (ii) the NO case in which every desired subgraph must have at least $\gamma \cdot hs$ edges.

It is not hard to see that Lin's reduction, with appropriate parameter setting, gives \W[1]-hardness of $\gapbsmd_\gamma$ for any constant $\gamma$. In what follows we sketch the reduction from \gapbsmd \ to \gapLDS.
Given an instance $(G=(L,R,E), s, h)$ of \gapbsmd, we choose a large finite field $\mathbb{F}_q$ so that the vertices of $G$ can be treated as elements of $\mathbb{F}_q$. Then we construct a function $\iota : L\cup R\to \mathbb{F}_q^{h-1}$ such that:
\begin{itemize}
\item[(L1)]the images of any $s-1$ vertices in $L$ under $\iota$ are linearly independent;
\item[(L2)] the images of any $s$ vertices in $L$ under $\iota$ are linearly dependent.
\end{itemize}
Similarly,
\begin{itemize} 
\item[(R1)] the images of any $h-1$ vertices in $R$ under $\iota$ are linearly independent;
\item[(R2)] while the   images of any $h$ vertices in $R$ under $\iota$ are linearly dependent.
\end{itemize}
We point out that one can construct functions satisfying the above properties by mapping the vertices (now identified with field elements) to the columns of a Vandermonde matrix of appropriate dimensions padded with zeros. Finally  we  construct a vector  $\mathbf{w}_e\in\mathbb{F}_q^{q(h-1)}$ for every edge in $e\in E$  and then let $\{\mathbf{w}_e : e\in E\}$ be our target instance of \gapLDS. To define $w_e$, firstly, we partition each vector in $\mathbb{F}_q^{q(h-1)}$ into $q$ blocks. Each vertex in $G$ has its unique corresponding block. Each block has $h-1$ elements. Suppose $e=\{u,v\}$ where $u\in L$ and $v\in R$. We set the $u$-th block of the vector $\mathbf{w}_{e}$ equal to $\iota(v)$,  the the $v$-th block of the vector $\mathbf{w}_{e}$ equal to $\iota(u)$ and all the other blocks of $\mathbf{w}_{e}$ equal to $\mathbf{0}$. Note that the $u$-th block is equal to $\iota(v)$ (not to $\iota(u)$!) and the $v$-th block is equal to $\iota(u)$.

Suppose that  $u_1,\ldots,u_s\in L$ and $v_1,\ldots,v_h\in R$ form a complete bipartite subgraph in $G$. We will show that the $sh$-sized set $W=\{\mathbf{w}_{u_i,v_j} : i\in[s],j\in[h]\}$ is linearly dependent. It is not hard to see that for all $i\in[s]$, the restriction of $W$ to the $u_i$-th block is a set of $h$ vectors $\{\iota(v_1),\ldots,\iota(v_h)\}$. By the property (R2) of $\iota$, these vectors are linearly dependent, i.e., there are $b_1,\ldots,b_h\in\mathbb{F}_q$ such that $\sum_{j\in[h]}b_j\iota(v_j)=\mathbf{0}$.
Similarly, we can see that for all $j\in [h]$, the restriction of $W$ to the $v_j$-th block is a set of $s$ linearly dependent vectors $\{\iota(u_1),\ldots,\iota(u_s)\}$ and $\sum_{i\in[s]}a_i\iota(u_i)=\mathbf{0}$ for some $a_1,\ldots,a_s\in\mathbb{F}_q$.
It is easy to check that $\sum_{i\in[s],j\in[h]}a_ib_jw_{u_i,v_j}=\mathbf{0}$ and $a_ib_j$ ($i\in[s]$,$j\in[h]$) are not all zero.

On the other hand, if $G$ is a NO instance of $\gapbsmd_\gamma$, we will show that any linearly dependent set must have at least $\gamma \cdot hs$ vectors. Observe that every vector in the \gapLDS\ instance is corresponding to an edge in the graph $G$. Suppose $W$ is a set of linearly dependent vectors. We consider the graph $H_W$ in $G$ induced by the edges corresponding to vectors in $W$. We can argue that every vertex on the left side of $H_W$ must have at least $h$ neighbors and every vertex on the right side of $H_W$ must have at least $s$ neighbors, using properties (R1) and (L1) respectively. From the definition of the NO instance of $\gapbsmd_\gamma$, we can immediately conclude that $|W| \geqs \gamma \cdot hs$.

\subsection{Parameterized Inapproximability  of $k$-\NCP}\label{sec:NCPoverview}

In the second step, we prove the inapproximability of $k$-\NCP\ using the hardness of \gapLDS. Note that this is the step in which we reduce the field size to two, i.e., the hardness for \gapLDS\ described above is for a large field ($\mathbb{F}_q$ where $q = \Theta(n)$) but the $k$-\NCP\ problem is for $\F$.

The reduction is simpler to state if we use the dual (equivalent) formulation of $k$-\NCP\ called Maximum Likelihood Decoding ($k$-\kvec). The gap version of the problem, denoted by $\gapvec_\gamma$, can be formulated as follows: Given a matrix $\bA\in\F^{n\times m}$, a vector $\by\in\F^n$ and a positive integer $k\in\mathbb{N}$, the goal of $\gapvec_\gamma$ problem is to distinguish the case when there exists a nonzero vector $\bx\in\F^m$ with Hamming weight at most $k$ such that $\bA\bx=\by$ from the case when for all  $\bx\in\F^m$ with Hamming weight at most $\gamma k$, $\bA\bx\neq\by$.

\bigskip

\noindent\textbf{Reducing \gapLDS\ to $\gapvec$.}
We  present a   reduction from    \gapLDS\ to \gapvec. For ease of presentation, we think of the input of \gapvec\ as a set $\mathcal W$ of vectors (i.e. column vectors of $\bA$) in $\F^n$, the goal is to distinguish the case when there exist $k$  vectors  whose sum is $\by$ from the case when the sum of any  nonempty subset of  vectors in $\mathcal W$ of size   at most $\gamma k$ is not equal to $\by$.

%Note that the underlying field of \gapvec\ is binary and the hardness result of \gapLDS\ is for $\mathbb{F}_{\Theta(n)}$.
We start with the hardness of $\gapLDS_\gamma$ where the input vectors are $\mathbb{F}_{2^d}$-vectors, for $d = \Theta(\log n)$. To reduce the field size, we transform vectors from $\mathbb{F}_{2^d}^m$ into $\mathbb{F}_{2}^{dm}$ using a linear bijection $f$ between  $\mathbb{F}_{2^d}^m$ and $\mathbb{F}_{2}^{dm}$. Observe that, even if $\bw_1,\ldots,\bw_k$ are linearly dependent vectors in $\mathbb{F}_{2^d}^m$, the sum of their images under $f$ is not necessarily zero. This is because we need coefficients $a_1,\ldots,a_k\in\mathbb{F}_{2^d}\setminus\{0\}$ so that $\sum_{i\in[k]}a_i\bw_i=\mathbf{0}$ and hence
$\sum_{i\in[k]}f(a_i\bw_i)=\mathbf{0}$, while $\sum_{i\in[k]}f(\bw_i)=\mathbf{0}$ may not hold.

With these in mind, we will try to construct an instance $\mathcal W'$ of \gapvec\ such that for all $a\in\mathbb{F}_{2^d}$ and $\bw\in\mathcal W$,  $f(a\bw)$ has a corresponding vector in $\mathcal W'$. And if $\mathcal W$ has $k$ linearly dependent vectors $\sum_{i\in[k]}a_i\bw_i=\mathbf{0}$, then the sum of vectors  corresponding to $f(a_1\bw_1),\ldots,f(a_k\bw_k)$ is equal to $\by$.

We need some mechanism to force the solution of \gapvec\ to select vectors corresponding to at least $k$ distinct vectors $f(a_1\bw_1),\ldots,f(a_k\bw_k)$.  To that end, we use the color-coding technique to reduce \gapLDS\ to its colored version (see Section~\ref{sec:colorLDS} for details).  Thus, we can assume that the instance $\mathcal W$ of \gapLDS\ comes with a coloring  $c : \mathcal W\to [k]$ such that if $\mathcal W$ is a YES instance, then there are exactly $k$ vectors  in $\mathcal W$ with distinct colors under $c$ that are linearly dependent.

For $i\in[k]$, let $\be_i \in \F^k$ be the vector whose $i$-th coordinate is $1$  and the other coordinates are equal to $0$. It is natural to construct a reduction as follows: given an instance $\mathcal W$ of \gapLDS\ over $\mathbb{F}_{2^d}$ and a coloring function $c : \mathcal W\to [k]$,  output  
\[
\mathcal W'=\{\be_{c(\bw)} \circ f(a\bw) : \bw\in\mathcal W,a\in\mathbb{F}_{2^d}\setminus\{0\}\} 
\text{ and }
\by = \vone_{k} \circ \bzero_{md}
\] 
as the target instance of \gapvec, where $\circ$ stands for the concatenation of vectors.

It is easy to see that if $\mathcal W$ contains $k$ linearly dependent vectors $\sum_{i\in[k]}a_i\bw_i=\mathbf{0}$, then the sum of the vectors $\be_{c(\bw_1)} \circ f(a_1\bw_1),\dots,\be_{c(\bw_k)} \circ f(a_k\bw_k)$ is equal to $\mathbf{1}_{dm}\circ \mathbf{0}_{k}$.

On the other hand, if any $3k$ vectors of $\mathcal W$ are linearly independent, we will show that for any  $W\subseteq \mathcal W'$ such that $\sum_{\bx\in W} \bx = \by$, we have $|W| \ge 3k$. Let the elements of $W$ be $\be_{c(\bw_1)} \circ f(a_1\bw_1), \dots, \be_{c(\bw_{k'})} \circ f(a_{k'}\bw_{k'})$, and suppose for the sake of contradiction that $k' < 3k$. By restricting the equation $\sum_{\bx\in W} \bx = \by$ onto the last $md$ coordinates, it follows that
\[
\sum_{i\in[k']}f(a_i\bw_i)=\mathbf{0},
\]
which  implies 
\[
\sum_{i\in[k']}a_i\bw_i=\mathbf{0}.
\]
At this moment, we cannot yet say that the set $\{\bw_1,\ldots,\bw_{k'}\}$ is linearly dependent (and therefore contradicts $k' < 3k$), because $\bw_1,\ldots,\bw_{k'}$ may contain duplicated elements. For example, it is possible that $k'=3$, $a_1+a_2+a_3=0$ and $\bw_1=\bw_2=\bw_3$ could be any nonzero vector. To get a contradiction by this way, we need to show that there is a vector $\bw$ which appears exactly once in $\bw_1,\ldots,\bw_{k'}$.

To see this, first observe that, since $k' < 3k$, there must be a color $j \in [k]$ that corresponds to at most two vectors from $\bw_1, \dots, \bw_{k'}$ (duplicated counted). However, if we restrict the equation $\sum_{\bx\in W} \bx = \by$ to only the $j$-th coordinate, we can see that the left hand side equals to the number of occurrences of color $j$ modulo 2, whereas the right hand side is one. This means that there is only a unique vector among $\bw_1, \dots, \bw_{k'}$ that is of the $j$-th color; this immediately implies that this vector occurs only once in $\bw_1, \dots, \bw_{k'}$. This in turns means that $\{\bw_1,\ldots,\bw_{k'}\}$ is linearly dependent and therefore $k' > 3k$, a contradiction.

Note that our argument only gives hardness of approximation with factor $3 - \varepsilon$ for any $\varepsilon > 0$. Nonetheless, this factor suffices for the subsequent steps. We can in fact also prove hardness for every constant factor, using a slight tweak of the above idea. Please see Section~\ref{sec:full-result} for more details

\bigskip

\noindent\textbf{Reducing $\gapvec$ to $\sncp$.}
Now we introduce the \emph{sparse nearest codeword problem} that we will use to prove the parameterized inapproximability of $k$-\MDP. 
We define the gap version of this problem, denoted by $\sncp_{\gamma}$ (for some constant $\gamma\geqs 1$) as follows: on input $(\bA',\by',k)$, distinguish between the YES case where there exists $\bx\in \F^m$ such that $\|\bA'\bx-\by'\|_0 + \|\bx\|_0 \leqs k$, and the NO case where for all $\bx$ (in the entire space), we have $\|\bA'\bx-\by'\|_0 + \|\bx\|_0 >\gamma k$. We highlight that the difference between $k$-\NCP\ and $\sncp_{\gamma}$ is that the objective also depends on the Hamming weight of the coefficient vector $\bx$. We sketch below the reduction from an instance $(\bA,\by,k)$ of $\gapvec_{\gamma}$ to an instance $(\bA',\by',k)$ of $\sncp_{\gamma}$. Given $\bA,\by$, let
$$
\bA'=\begin{bmatrix} 
    \bA  \\
    \vdots \\
    \bA\\
    \Id 
    \end{bmatrix}
    \begin{tikzpicture}[baseline=2.2ex]
\draw [decorate,decoration={brace,amplitude=10pt}]
(-1.5,1.4) -- (-1.5,0) node [black,midway,xshift=45pt] { $\gamma k + 1$ copies};
\end{tikzpicture},\ \ \\
\by'=\begin{bmatrix} 
    \by  \\
    \vdots \\
    \by\\
    \bzero 
    \end{bmatrix}
    \begin{tikzpicture}[baseline=2.2ex]
\draw [decorate,decoration={brace,amplitude=10pt}]
(-1.5,1.4) -- (-1.5,0) node [black,midway,xshift=45pt] { $\gamma k + 1$ copies};
\end{tikzpicture}
.$$

Notice that for any $\bx$ (in the entire space), we have
$$\|\bA'\bx - \by'\|_0 = (\gamma k + 1)\|\bA\bx-\by\|_0+\|\bx\|_0,$$ 
and thus both the completeness and soundness of the reduction easily follow.

\subsection{Parameterized Inapproximability of $k$-\MDP}\label{sec:MDPoverview}

Let us recall that in the \NCP\ we are given a matrix $\bA \in \F^{n\times q}$, an integer $k$, and a target vector $\by\in\F^n$, and the goal is to determine whether there is exists a vector $\bx\in\F^m$ such that $\|\bA\bx - \by\|_0 $ is at most $k$. A natural first idea for reducing an \NCP\ instance $(\bA\in \F^{n\times m},\by \in\F^{n})$ to \MDP\ would be to introduce the $n\times (m+1)$ matrix
\begin{align}
	\bB =
	\begin{bmatrix}
		\bA  & - \by \\
	\end{bmatrix};
\end{align}
then any solution $\bx\in \F^m$ of the \NCP\ instance with $\|\bA\bx-\by\|_0\le k$ would give a solution $\bx'=\bx\circ 1\in \F^{m+1}$ of the \MDP\ instance with $|\bB \bx'|\le k$. However, the problem is that if the \MDP\ instance has a solution $\bx'=\bx \circ 0$ (i.e., the last coordinate is zero), then $\|\bB \bx'\|_0\le k$ implies only $\|\bA\bx\|_0\le k$, but does not imply $\|\bA\bx-\by\|_0\le k$. Thus we need a way to force the last coordinate to 1 in the solution of the \MDP\ instance. We can try to use error correcting codes for this purpose. Let $\bL\in \mathbb{F}_2^{h\times m}$ be the generator matrix of an error correcting code with minimum distance $d$. Let us consider now the matrix
\begin{align}
	\bB =
	\begin{bmatrix}
		\bA  & - \by \\
		\bL & - \bs\\
	\end{bmatrix};
\end{align}
for some arbitrarily chosen vector $\bs\in \mathbb{F}_2^{h}$. Now for any nonzero $\bx'=\bx\circ 0$, we have $\|\bB\bx'\|_0=\|\bA\bx\|_0+\|\bL\bx\|_0\ge \|\bA\bx\|_0+d$, since $\bx$ is a nonzero vector and the code generated by $\bL$ has minimum distance $d$. Thus the second term gives a penalty of $d$ if the last coordinate of $\bx'$ is 0. However, the problem now is that if $\bx$  is a solution of the \NCP\ instance with $\|\bA\bx-\by\|_0\le k$, then defining $\bx'=\bx\circ 1$ gives $\|\bB\bx'\|_0=\|\bA\bx-\by\|_0+\|\bL\bx-\bs\|_0=k+\|\bL\bx-\bs\|_0$. We would need to argue that this second term $\|\bL\bx-\bs\|_0$ is small, much smaller than the penalty $d$ in the previous case. While in general, there is no reason why the chosen vector $\bs$ would be close to $\bL\bx$ for the hypothetical solution $\bx$. 
However, we can hope to increase the chances of finding such an $\bs$, if we could somehow enforce that there are many distinct choices of $\bx$ for which we would have $\|\bA\bx-\by\|_0\le k$. This can indeed by achieved by padding the matrix $\bA$ with additional dummy zero columns, and padding the corresponding solution $\bx$ with additional dummy coordinates. In particular, this ensures that for even a random choice of $\bs$ (sampled from an appropriate distribution) is close to $\bL\bx$ for at least one of the choices of $\bx$ with non-negligible probability. We formalize this intuition in the form of {\em Locally Suffix Dense Codes} described below.
\bigskip

\noindent\textbf{Locally Suffix Dense Codes.}
A locally suffix dense code (LSDC) is a linear code $\bL \in \F^{h \times m}$ of block length $h$ with minimum distance $d$ such that the following holds. For any choice of prefix $\bx \in \F^q$ and a randomly drawn suffix vector $\bs \overset{{\rm u.a.r}}\sim \F^{h-q}$ the vector $\bx\circ\bs$ is $r$ close to the code $\bL$ with non-negligible probability. In other words, for every choice of prefix vector $\bx$, the restriction of the code $\bL$ to the affine subspace $V_{\bx}:=\{\bx \} \times \F^{h-q}$ is {\em dense}. While one can think of the suffix vector $\bs$ as being analogous to the center in LDC, note that $\bs$ is merely a suffix which is used to extend the vector $\bx$. Therefore, due to systematicity of the code, the distance of the vector $\bx\circ\bs$ to the code $\bL$ depends only on the choice of $\bs$, which allows us to ensure that the parameters $r$ and $d$ can be chosen to functions of $k$, without explicitly depending on the block length $h$.

As in the case of Dumer et al.\, we too cannot find an explicit suffix $\bs$ for the LSDCs that we construct, but instead provide an efficiently samplable distribution such that, for any $\mathbf{x} \in \F^q$, the probability (over $\bs$ sampled from the distribution) that $  \mathcal{B}(\mathbf{x}\circ \bs,r) \cap \bL \neq \emptyset$ is non-negligible. This is what makes our reduction from \sncp$_{2.5}$ to \mdp$_{1.01}$ randomized. We will not elaborate more on this issue here, but focus on the (probabilistic) construction of such codes. For convenience, we will assume throughout this overview that $k$ is much smaller than $d$, i.e., $k = 0.001d$.

Recall that the sphere-packing bound (aka Hamming bound) states that a binary code of block length $h$ and distance $d$ can have at most $2^h/|\cB(\vzero, \lceil \frac{d - 1}{2} \rceil)|$ codewords; this is simply because the balls of radius $\lceil \frac{d - 1}{2} \rceil$ at the codewords do not intersect. Our main theorem regarding the existence of locally dense suffix code is that any systematic code that is ``near'' the sphere-packing bound is a locally dense suffix code with $r = \lceil \frac{d - 1}{2} \rceil$. Here ``near'' means that the number of codewords must be at least $2^h/|\cB(\vzero, \lceil \frac{d - 1}{2} \rceil)|$ divided by $f(d) \cdot \poly(h)$ for some function $f$ that depends only on $d$. (Equivalently, this means that the message length must be at least $h - (d/2 + O(1))\log h$.) The BCH code over binary alphabet is an example of a code satisfying such a condition.

While we will not sketch the proof of the existence theorem here, we note that the general idea is as follows. We choose $\bL$ in such a way that for every choice of $\bx \in \F^q$, the restriction of $\bL$ to the affine subspace $V_{\bx}$ is near the sphere packing bound. Then from the above discussion, it follows that for $\bs$ sampled uniformly from $\F^{h-q}$, the probability that $\cB(\bx\circ\bs, r) \cap \bL  \neq \emptyset$ is at least the probability that a random point in $\F^{h-q}$ is within distance $r = \lceil \frac{d - 1}{2} \rceil$ of some codeword from $V_\bx \cap \bL$. The latter is non-negligible from our choice of $\bL$ which ensures that the restriction of the code to any affine subspace $V_{\bx}$ nears the sphere-packing bound. 

Finally, we remark that our proof here is completely different from the DMS proof of existence of locally dense codes. Specifically, DMS uses a group-theoretic argument to show that, when a code exceeds the Gilbert-Varshamov bound, there must be a center $\bs$ such that $\cB(\bs, r)$ contains many codewords. Then, they pick a random linear map $\bT$ and show that w.h.p. $\bT(\cB(\bs, r) \cap \bL)$ is the entire space. Note that this second step does not use any structure of $\cB(\bs, r) \cap \bL$; their argument is simply that, for any sufficiently large subset $Y$, a random linear map $\bT$ maps $Y$ to an entire space w.h.p. However, such an argument fails for us, due to the fact that, in LSDC, we want to ensure that $\bL$ is dense (up to Hamming distance $r = O(k)$) in all the affine subspaces $\{V_\bx : \bx \in \F^{q}\}$, instead of exactly covering the whole space $\F^{h}$. Now if we insist on exactly covering all the affine subspaces using a linear map $\bT$, as in the DMS construction, we will then have $\bT(\cB(\bs, r)) \supseteq \F^h$. This would instead require $r$ to depend on $h$, whereas in our setting we want $r$ to depend only on the parameter $k$.

\bigskip

\noindent\textbf{Reducing $\sncp_{2.5}$ to  \mdp$_{1.01}$.} Equipped with the notion of locally suffix dense codes defined above, we now prove the hardness of $\mdp_{1.01}$.

We begin with an instance $(\bA,\by,k)$ with $\bA\in\mathbb{F}_2^{n\times q}$ of $\sncp_{2.5}$. Let $\bL \in \F^{h \times m}$ be a locally suffix dense code with distance  $d\approx 2.5 k$, where we can choose $h,m \leq {\rm poly}(q,d)$. We also choose a vector $\bs\in\mathbb{F}_2^{h}$ uniformly at random with the first $q$ coordinates equal to zero and construct the matrix
\begin{align*}
\bB =
\begin{bmatrix}
  \bA  & \bzero_{n \times (m-q)} & - \by \\
  \bL & 						 &- \bs\\
\end{bmatrix}.
\end{align*}

We shall show that with probability at least $p = p(k)$ {\footnote{Here the probability $p = p(k)$ depends only on the parameter $k$}}, we have that $(\bB,k+d/2)$ is an instance of $\mdp_{1.01}$. 
% and in particular see \eqref{eq:LDCGadget}).

If $(\bA,\by,k)$ is a YES instance of $\sncp_{2.5}$, then there exists $\bx\in \cB(\mathbf 0,k)$ such that $\|\bA\bx-\by\|_0\leqs k$. Furthermore, from the guarantees of the locally suffix dense codes, with probability at least $p$ (over the choice of the vector $s$), we have $\|\bL\bx-\bs\|_0\le (d-1)/2$. Therefore, setting $\bz=\bx'\circ 1$, we get that
\[
\|\bB\bz\|_0 = \|\bA'\bx'-\by\|_0+ \|\bL\bx'-\bs\|_0\le k+ (d-1)/2.
\]
In other words,  if $(\bA,\by,k)$ is a YES instance of \NCP, then $(\bB,k+d/2)$ is a YES instance of $\MDP_{1.01}$. 

\bigskip

\noindent On the other hand, if we had that $\|\bA\bx-\by\|_0 + \|\bx\|_0> 2.5 k$ for all $\bx$, then for all non-zero $\bz\in\mathbb{F}_2^m$,
\[
\|\bB(\bz\circ 0)\|_0=\|\bA'\bz\|_0+\|\bL\bz\|_0\ge d,
\]
and
\[
\|\bB(\bz\circ 1)\|_0=\|\bA'\bz-\by\|_0+\|\bL\bz-\bs\|_0\ge 2.5 k.
\]
Since from our choice of parameters, we have $d\approx 2.5 k\ge 1.01(k+d/2)$, which implies that $(\bB,k+d/2)$ is a NO instance of $\MDP_{1.01}$.
%$\|\bB\bz\|_0>\gamma' (ak+br)$ for all $\bz$. 

\bigskip

\noindent\textbf{Gap Amplification for $\mdp_{1.01}$.} It is  well known that the distance of the tensor product of two linear codes is the product of the distances of the individual codes (see Proposition~\ref{prop:gap-amplification} for a formal statement). We can use this proposition to reduce $\mdp_{\gamma}$ to $\mdp_{\gamma^2}$ for any $\gamma\geqs 1$. In particular, we can obtain, for any constant $\gamma$, the intractability of $\mdp_\gamma$ starting from $\mdp_{1.01}$ by just recursively tensoring the input code $\lceil \log_{1.01} \gamma\rceil $ times.

\subsection{Parameterized Intractability of $k$-\SVP}\label{sec:SVPoverview}

We begin this subsection by briefly describing Khot's reduction. The starting point of Khot's reduction is the \NP-hardness of approximating \CVP\ in every $\ell_p$ norm to any constant factor \cite{ABSS97}. Let us recall that in \CVP\ in the $\ell_p$ norm, we are given a matrix $\bA \in \Z^{n\times m}$, an integer $k$, and a target vector $\by\in\Z^n$, and the goal is to determine whether there is any $\bx\in\Z^m$ such that\footnote{Previously, we use $\|\bA\bx - \by\|_p$ instead of $\|\bA\bx - \by\|_p^p$. However, from the fixed parameter perspective, these two versions are equivalent since the parameter $k$ is only raised to the $p$-th power, and $p$ is a constant in our setting.} $\|\bA\bx - \by\|_p^p$ is at most $k$.  The result of Arora et al. \cite{ABSS97} states that for any constant $\gamma\geqs 1$, it is \NP-hard to distinguish the case when there exists $\bx$ such that $\|\bA\bx - \by\|_p^p\leqs k$ from the case when for all (integral) $\bx$ we have that $\|\bA\bx - \by\|_p^p> \gamma k$. 
Khot's reduction proceeds in four steps. First, he constructs a gadget lattice called the  ``BCH Lattice'' using BCH Codes. Next, he reduces \CVP\ in the $\ell_p$ norm (where $p\in(1,\infty)$) to an instance of \SVP\ on an intermediate lattice by using the BCH Lattice.
This intermediate lattice has the following property. For any YES instance of \CVP\ the intermediate lattice contains multiple copies of the witness of the YES instance; For any NO instance of \CVP\ there are also many ``annoying vectors'' (but far less than the total number of YES instance witnesses)  which look like witnesses of a YES instance. However, since the annoying vectors are outnumbered, Khot reduces this intermediate lattice to a proper \SVP\ instance, by randomly picking a sub-lattice via a random homogeneous linear constraint on the coordinates of the lattice vectors (this annihilates all the annoying vectors while retaining at least one witness for the YES instance). Thus he obtains some constant factor hardness for \SVP. Finally, the gap is amplified via ``Augmented Tensor Product''. It is important to note that Khot's reduction is randomized, and thus his result of inapproximability of \SVP\ is based on $\NP\neq \RP$.

Trying to follow Khot's reduction, in order to show the parameterized intractability of $k$-\SVP, we face only one obstacle: there is no known parameterized inapproximability of $k$-\CVP\  for any constant factor greater than 1. Let us denote by $\snvp_{p,\eta}$ for any constant $\eta\geqs 1$ the gap version of $k$-\CVP\ in the $\ell_p$ norm. Recall that in $\snvp_{p,\eta}$ we are given a matrix $\bA \in \Z^{n\times m}$, a target vector $\by\in\Z^n$, and a parameter $k$, and we would like to distinguish the case when there exists $\bx \in \Z^m$ such that $\|\bA\bx - \by\|_p^p\leqs k$ from the case when for all $\bx \in \Z^m$ we have that $\|\bA\bx - \by\|_p^p> \eta k$.
As it turns out, our reduction from $k$-\LDS\ to $\gapvec$, can be translated to show the inapproximability of $\gapvec$ over any larger (but still constant) field in a straightforward manner. We then provide a simple for $\gapvec$ over large field to $\snvp_p$ that establishes $\W[1]$-hardness of the latter.

Once we have established the constant parameterized inapproximability of $\snvp_p$, we follow Khot's reduction, and everything goes through as it is to establish the inapproximability for some factor of the gap version of $k$-\SVP\ in the $\ell_p$ norm (where $p\in (1,\infty)$). We denote by $\svp_{p,\gamma}$ for some constant $\gamma(p)\geqs 1$ the the gap version of $k$-\SVP\ (in the $\ell_p$ norm) where we are given a matrix $\bB \in \Z^{n\times m}$ and a parameter $k \in \N$, and we would like to distinguish the case when there exists a non-zero $\bx \in \Z^m$ such that $\|\bB\bx\|_p^p\leqs k$ from the case when for all $\bx \in \Z^m \setminus \{\bzero\}$ we have that $\|\bB\bx\|_p^p> \gamma k$. 
Let $\gamma^*:=\frac{2^p}{2^{p-1}+1}$. Following Khot's reduction, we obtain the inapproximability of $\svp_{p,\gamma^*}$.

Summarizing, in Figure~\ref{fig:overviewSV}, we provide the proof outline of our \W[1]-hardness  result of $\svp_p$ to some constant approximation factor, for every $p\in (1,\infty)$.

\begin{figure}[h!]
    \centering
    \resizebox{\textwidth}{!}{\begin{tikzpicture}

\node (gapLDS) [draw=red!80!black,thick] at (1.5, 0) {\footnotesize \gapLDS};

\node (colLDS) [draw=red!80!black,thick] at (7.5, 0) {\footnotesize \textsc{Colored} \gapLDS};
\node (gapmld) [draw=red!80!black,thick] at (13, 0) {\footnotesize \gapvec};

\node (kclique) [draw=red!80!black,thick] at (1.5, 3) {\footnotesize $k$-\textsc{Clique}};

\node (biclique) [draw=red!80!black,thick] at (7.5, 3) {\footnotesize \textsc{One-Sided Biclique}};
\node (gapBSMD) [draw=red!80!black,thick] at (13, 3) {\footnotesize \gapbsmd};

\node (gapsnc) [draw=red!80!black,thick] at (4.25, -3) {\footnotesize \snvp$_p$};

\node (gapmdc) [draw=red!80!black,thick] at (10.55, -3) {\footnotesize \svp$_{p,\frac{2^p}{2^{p-1}+1} }$};

\node [above, align=center] at (7.2, -3.5) {\footnotesize Khot's Reduction};
\node [above, align=center] at (7.25, -3.9) {\footnotesize Lemma~\ref{lem:snvp-to-svp}};

\draw [-{Latex[length=1.5mm, width=1.5mm]}] (kclique) -- (biclique);
\draw [-{Latex[length=1.5mm, width=1.5mm]}] (biclique) -- (gapBSMD);

\draw [-{Latex[length=1.5mm, width=1.5mm]}] (gapBSMD) -- (gapLDS);
\draw [-{Latex[length=1.5mm, width=1.5mm]}] (gapLDS) -- (colLDS);
\draw [-{Latex[length=1.5mm, width=1.5mm]}] (colLDS) -- (gapmld);

\draw [-{Latex[length=1.5mm, width=1.5mm]}] (gapmld) -- (gapsnc);
\draw [-{Latex[length=1.5mm, width=1.5mm]}] (gapsnc) -- (gapmdc);

\node [above, align=center] at (4.1, 3) {\footnotesize \cite{Lin15}};
\node [above, align=center] at (10.6, 3) {\footnotesize Section~\ref{sec:biclique-to-bsmd}};

\node [above, align=center,rotate=12] at (7.5, 1.48) {\footnotesize Section~\ref{sec:bsmd-to-LDS}};

\node [above, align=center] at (4.1, -0.6) {\footnotesize Color-Coding Technique};
\node [above, align=center] at (4.11, -1) {\footnotesize (Section~\ref{sec:colorLDS})};

\node [above, align=center] at (10.6, 0) {\footnotesize Section~\ref{sec:full-result}};

\node [above, align=center,rotate=18] at (7.85, -1.85) {\footnotesize Section~\ref{sec:csp-to-snvp}};

\end{tikzpicture}}
    \caption{
The figure provides an overview of our reduction from the canonical \W[1]-complete  $k$-Clique problem to the parameterized Shortest Vector problem in the $\ell_p$ norm, where $p\in(1,\infty)$. The proof outline of the reduction from $k$-Clique to \gapvec\ (to rule out constant approximation parameterized algorithms for \NCP) is reiterated in the above figure.
		In Section~\ref{sec:csp-to-snvp}, we reduce \gapvec\ to \snvp\ and obtain the constant inapproximability of \CVP. Then, applying
	        Lemma~\ref{lem:snvp-to-svp} (i.e., Khot's reduction)  implies  the parameterized inapproximability of $\svp_{p,\frac{2^p}{2^{p-1}+1}}$. } \label{fig:overviewSV}
\end{figure}
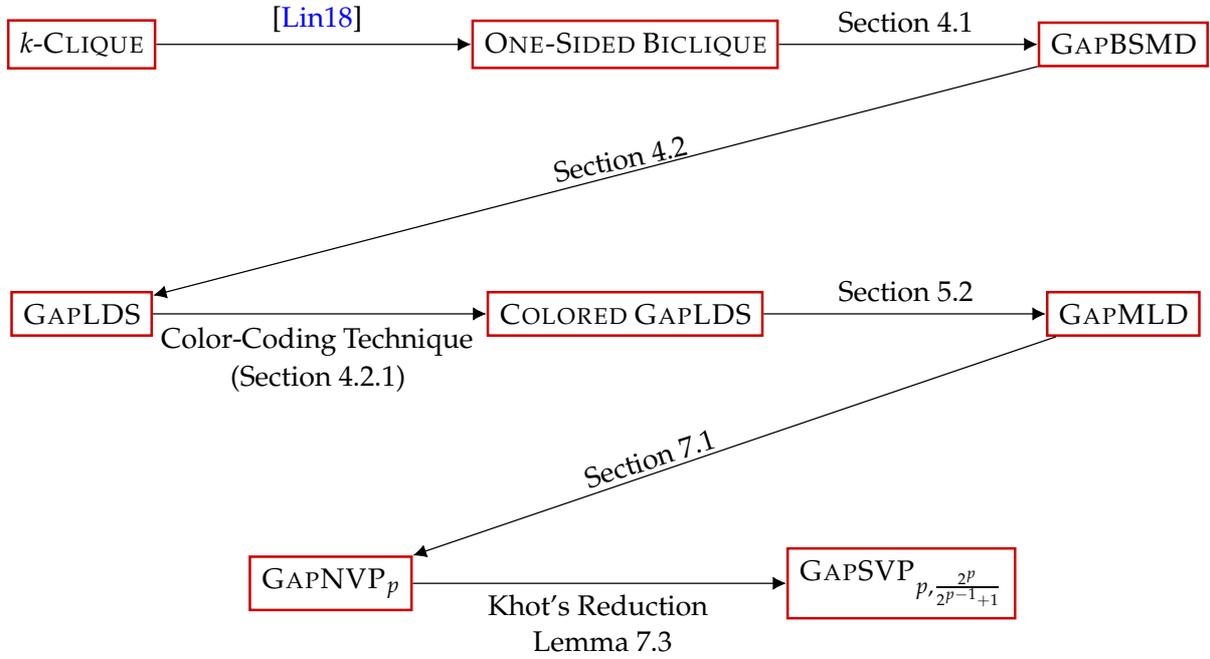

%%% Local Variables:
%%% mode: latex
%%% TeX-master: "JACM submission"
%%% End:

\section{Preliminaries} \label{sec:prelim}

We use the following notation throughout the paper.

\noindent\textbf{Notation.} We use boldface (e.g. $\bx, \bA$ or $\bzero$) to stress that the objects are vectors or matrices. When we refer to a vector $\bx$, we assume that it is a column vector. Moreover, since subscripts will often be used for other purposes, we instead use the notation $\bx[i]$ for $i \in \mathbb{N}$ to denote the value of the $i$-th coordinate of the vector. For matrices, we use $\bA[i]$ to denote its $i$-th column vector.

For $p \in \mathbb{N}$, we use $\vone_p$ (respectively, $\vzero_p$) to denote the all ones (respectively, all zeros) vector of length $p$.
We sometimes drop the subscript if the dimension is clear from the context. 
For $p,q \in \mathbb{N}$, we use $\vzero_{p\times q}$ to denote the all zeroes matrix of $p$ rows and $q$ columns. We use Id$_q$ to denote the identity matrix of $q$ rows and $q$ columns.

For any vector $\bx\in\mathbb{R}^d$, the $\ell_p$ norm of $\bx$
is defined as
$\ell_p(\bx) = \|\bx\|_p = \left(\sum_{i=1}^d|\bx[i]|^p\right)^{1/p}$.
Thus, $\ell_{\infty}(\bx) = \|\bx\|_{\infty} = \max_{i\in[d]}\{|\bx_i|\}$.
The $\ell_0$ norm of $\bx$ is defined as
$\ell_0(\bx) = \|\bx\|_0=|\{\bx[i]\neq 0: i\in [d]\}|$, i.e.,
the number of non-zero entries of $\bx$. We note that the $\ell_0$ norm is also referred to as the Hamming norm. For $a\in\mathbb N$, $t\in\mathbb{N}\cup \{0\}$, and $\bs\in\{0,1\}^a$, we use $\mathcal{B}_a(\bs,t)$ to denote the Hamming ball of radius $t$ centered at $\bs$, i.e., $\mathcal{B}_a(\bs,t)=\{\bx\in\{0,1\}^a\mid \|\bs-\bx\|_0\leqs t\}$. Finally, given two vectors $\bx$ and $\by$, we use $\bx \circ \by$ to denote the concatenation of vectors $\bx$ and $\by$. 

We sometimes use $\dotcup$ to emphasize that the sets are disjoint; for instance, we may write $G = (L \dotcup R, E)$ for bipartite graphs to indicate that $L, R$ are disjoint.

\subsection{Parameterized Promise Problems and (Randomized) FPT Reductions}
In this subsection, we briefly describe the various kinds of fixed-parameter reductions that are used in this paper. We start by defining the notion of promise problems in the fixed-parameter world, which is naturally analogues to promise problems in the NP world (see e.g.~\cite{Gol06}).

\begin{definition}
A parameterized promise problem $\Pi$ is a pair of parameterized languages $(\Pi_{YES}, \Pi_{NO})$ such that $\Pi_{YES} \cap \Pi_{NO} = \emptyset$.
\end{definition}

Next, we formalize the notion of algorithms for these parameterized promise problems:

\begin{definition}
A deterministic algorithm $\cA$ is said to be an \emph{FPT algorithm for $\Pi$} if the following holds:
\begin{itemize}
\item On any input $(x, k)$, $\cA$ runs in time $f(k)|x|^c$ for some computable function $f$ and constant $c$.
\item (YES) For all $(x, k) \in \Pi_{YES}$, $\cA(x, k) = 1$.
\item (NO) For all $(x, k) \in \Pi_{NO}$, $\cA(x, k) = 0$.
\end{itemize}
\end{definition}

\begin{definition}
A Monte Carlo algorithm $\cA$ is said to be a \emph{randomized FPT algorithm for $\Pi$} if the following holds:
\begin{itemize}
\item $\cA$ runs in time $f(k)|x|^c$ for some computable function $f$ and constant $c$ (on every randomness).
\item (YES) For all $(x, k) \in \Pi_{YES}$, $\Pr[\cA(x, k) = 1] \geqs 2/3$.
\item (NO) For all $(x, k) \in \Pi_{NO}$, $\Pr[\cA(x, k) = 0] \geqs 2/3$.
\end{itemize}
\end{definition}

Finally, we define deterministic and randomized reductions between these problems.

\begin{definition}
A (deterministic) FPT reduction from a parameterized promise problem $\Pi$ to a parameterized promise problem $\Pi'$ is a (deterministic) procedure that transforms $(x, k)$ to $(x', k')$ that satisfies the following:
\begin{itemize}
\item The procedure runs in $f(k) |x|^c$ for some computable function $f$ and constant $c$.
\item There exists a computable function $g$ such that $k' \leqs g(k)$ for every input $(x, k)$.
\item For all $(x, k) \in \Pi_{YES}$, $(x', k') \in \Pi_{YES}'$.
\item For all $(x, k) \in \Pi_{NO}$, $(x', k') \in \Pi_{NO}'$.
\end{itemize}
\end{definition}

\begin{definition}					\label{def:rand-fpt-red}
A randomized (one sided error) FPT reduction from a parameterized promise problem $\Pi$ to a parameterized promise problem $\Pi'$ is a randomized procedure that transforms $(x, k)$ to $(x', k')$ that satisfies the following:
\begin{itemize}
\item The procedure runs in $f(k) |x|^c$ for some computable function $f$ and constant $c$ (on every randomness).
\item There exists a computable function $g$ such that $k' \leqs g(k)$ for every input $(x, k)$.
\item For all $(x, k) \in \Pi_{YES}$, $\Pr[(x', k') \in \Pi_{YES}'] \geqs 1/(f'(k)|x|^{c'})$ for some computable function $f'$ and constant $c'$.
\item For all $(x, k) \in \Pi_{NO}$, $\Pr[(x', k') \in \Pi_{NO}'] = 1$.
\end{itemize}
\end{definition}

Note that the above definition corresponds to the notion of \emph{Reverse Unfaithful Random (RUR) reductions} in the classical world~\cite{J90}. The only difference (besides the allowed FPT running time) is that the above definition allows the probability that the YES case gets map to the YES case to be as small as $1/(f'(k)\poly(|x|))$, whereas in the RUR reductions this can only be $1/\poly(|x|)$. The reason is that, as we will see in Lemma~\ref{lem:red-intract} below, FPT algorithms can afford to repeat the reduction $f'(k)\poly(|x|)$ times, whereas polynomial time algorithms can only repeat $\poly(|x|)$ times.

We also consider randomized two-sided error FPT reductions, which are defined as follows. 

\begin{definition}
	A randomized \emph{two sided error} FPT reduction from a parameterized promise problem $\Pi$ to a parameterized promise problem $\Pi'$ is a randomized procedure that transforms $(x, k)$ to $(x', k')$ that satisfies the following:
	\begin{itemize}
		\item The procedure runs in $f(k) |x|^c$ for some computable function $f$ and constant $c$ (on every randomness).
		\item There exists a computable function $g$ such that $k' \leqs g(k)$ for every input $(x, k)$.
		\item For all $(x, k) \in \Pi_{YES}$, $\Pr[(x', k') \in \Pi_{YES}'] \geqs 2/3$.
		\item For all $(x, k) \in \Pi_{NO}$, $\Pr[(x', k') \in \Pi_{NO}']  \geqs 2/3$.
	\end{itemize}
\end{definition}

Note that this is not a generalization of the standard randomized FPT reduction (as defined in Definition \ref{def:rand-fpt-red}), since the definition requires the success probabilities for the YES and NO cases to be constants independent of the parameter. In both cases, using standard techniques randomized FPT reductions, can be used to transform randomized FPT algorithms for $\Pi'$ to randomized FPT algorithm for $\Pi$, as stated by the following lemma:

\begin{lemma}
	Suppose there exists a randomized (one sided/ two sided) error FPT reduction from a parameterized promise problem $\Pi$ to a parameterized promise problem $\Pi'$. If there exists a randomized FPT algorithm $\mathcal{A}$ for $\Pi'$, there there also exists a randomized FPT algorithm for $\Pi$.
	\label{lem:red-intract}
\end{lemma}
\begin{proof}
	We prove this for one sided error reductions, the other case follows using similar arguments. Suppose there exists a randomized one sided error reduction from $\Pi$ to $\Pi'$. 
	Let $f'(\cdot),c'$ be as in Definition \ref{def:rand-fpt-red}. We consider the following subroutine. Given instance $(x,k)$ of promise problem $\Pi$, we apply the randomized reduction on $(x,k)$ to get instance $(x',k')$ of promise problem $\Pi'$. 
	We run $\mathcal{A}$ on $(x',k')$ repeatedly $100\log (f'(k)|x|^c)$ times, and output the majority of the outcomes. 
	
	If $(x,k)$ is a YES instance, then with probability at least $1/(f'(k)|x|^{c'})$, $(x',k')$ is also a YES instance for $\Pi'$. Using Chernoff bound, conditioned on $(x',k')$ being a YES instance, the majority of the outcomes is YES with probability at least $1  - e^{-10\log(f'(k)|x|^{c'})}$. Therefore using union bound, the output of the above algorithm is YES with probability at least $1/(f'(k)|x|^{c'}) - e^{-10\log(f'(k)|x|^{c'})} \geqs 1/2(f'(k)|x|^{c'})$. Similarly, if $(x,k)$ is a NO instance, then the subroutine outputs YES with probability at most $ e^{-10\log(f'(k)|x|^{c'})}$.
	
	Equipped with the above subroutine, our algorithm is simply the following: given $(x,k)$, it runs the subroutine $10f'(k)|x|^{c'}$ times. If at least one of the outcomes is YES, then the algorithm outputs YES, otherwise it outputs NO. Again we can analyze this using elementary probability. If $(x,k)$ is a YES instance, then the algorithm outputs NO only if outcomes of all the trials is NO. Therefore, the algorithm outputs YES with probability at least $1 - ( 1 - 1/2(f'(k)|x|^{c'}))^{10f'(k)|x|^{c'}} \geqs 0.9$. Conversely, if $(x,k)$ is a NO instance, then by union bound, the algorithm outputs NO with probability at least $1 - 10f'(k)|x|^{c'}  e^{-10\log(f'(k)|x|^{c'})} \geqs 0.9$. Finally, if $\mathcal{A}$ is FPT, then the running time of the proposed algorithm is also FPT. Hence the claim follows{\footnote{For the case of $2$-sided error, we change the final step of the algorithm as follows; we invoke the subroutine $O(\log 1/\delta)$-times (where $\delta$ is a constant) and again output the majority of the outcomes. The guarantees again follow by a Chernoff bound argument. }}.
\end{proof}	

Since the conclusion of the above proposition holds for both types of randomized reductions, we will not be distinguishing between the two types in the rest of the paper. 

\subsection{Bipartite Subgraph with Minimum Degrees}

As stated in the proof overview, it will be convenient to view Lin's hardness of \textsc{Biclique} in terms of hardness of approximating Bipartite Subgraph with Minimum Degree, where the goal, given a bipartite graph $G$, is to find a non-empty subgraph $H$ of $G$ such that every left vertex in $H$ has degree at least $h$ and every right vertex of $H$ has degree at least $s$. The parameter here is $s + h$.

The gap version that we will use is to distinguish between the YES case where there is such a subgraph with $hs$ edges, i.e., a complete bipartite subgraph with $s$ left vertices and $h$ right vertices, and the NO case where every such subgraph $H$ must contains more than $\gamma \cdot hs$ edges (for $\gamma \geqs 1$). This is defined more precisely below.

\begin{framed}
$\gamma$-Gap Bipartite Subgraph with Minimum Degree Problem ($\gapbsmd_{\gamma}$)

{\bf Input: } A bipartite graph $G= (L\dotcup R, E)$ with $n$ vertices, $s,h\in\mathbb{N}$

{\bf Parameter: } $s + h$

{\bf Question: } Distinguish between the following two cases:
\begin{itemize}
\item (YES) There is a complete bipartite subgraph of $G$ with $s$ vertices in $L$ and $h$ vertices in $R$.
\item (NO) For any non-empty subgraph $H$ of $G$ such that every left vertex of $H$ has degree at least $h$ and every right vertex of $H$ has degree at least $s$, $H$ contains at least $\gamma \cdot (sh)$ edges.
\end{itemize}
\end{framed}

\subsection{Linear Dependent Set Problems}
We next introduce the parameterized Linear Dependent Problem. In this problem, we are given $\mathbb{F}_q$-vectors $\bw_1, \dots, \bw_n$ and the goal is to find a smallest number of vectors that are linearly dependent. It should be stressed here that the field $\mathbb{F}_q$ is part of the input (i.e. $q$ will be of the order of $n$ in our proofs); this is indeed the main difference between this problem and the Minimum Distance Problem which is in fact equivalent to the Linear Dependent Problem for a fixed $q = 2$.

\begin{framed}
$\gamma$-Gap Linear Dependent Set Problem ($\gapLDS_{\gamma}$)

{\bf Input: } A field $\mathbb{F}_q$, a set $\mathcal{W} \subseteq \mathbb{F}_q^m$ and a positive integer $k\in\mathbb{N}$.

{\bf Parameter: } $k$

{\bf Question: } Distinguish between the following two cases:
\begin{itemize}
\item (YES) there exist $k$ distinct vectors
$\bw_1, \dots, \bw_k \in \cW$ and $a_1, \dots, a_k \in \mathbb{F}_q \setminus \{0\}$ such that $\sum_{i \in [k]} a_i \bw_i = \bzero$ (which implies that $\bw_1, \dots, \bw_k$ are linearly dependent)
\item (NO) there are no $\gamma \cdot k$ vectors in $\mathcal W$ that are linearly dependent
\end{itemize}
\end{framed}

Notice here that the guarantee in the YES case is slightly stronger than ``there exist $k$ vectors that are linearly dependent'', as we also require the coefficients to be non-zero. (This would be automatically true if, for instance, any $k - 1$ vectors are linearly dependent.) We remark that this does not significantly change the complexity of the problem, as our hardness applies to both versions; however, it will be more convenient in subsequent steps to have such an additional guarantee.

It will also be convenient to work with a colored version of $\gapLDS$ which we introduce below. 

\begin{framed}
$\gamma$-Gap Colored Linear Dependent Set Problem ($\gapLDS_{\gamma}^{\textsf{col}}$)

{\bf Input: } A field $\mathbb{F}_q$, a set $\mathcal{W} \subseteq \mathbb{F}_q^m$, a positive integer $k\in\mathbb{N}$ and a coloring $c: \cW \to [k]$

{\bf Parameter: } $k$

{\bf Question: } Distinguish between the following two cases:
\begin{itemize}
\item (YES) there exist $k$ vectors
$\bw_1, \dots, \bw_k \in \cW$ of distinct colors (i.e. $c(\{\bw_1, \dots, \bw_k\}) = [k]$) and $a_1, \dots, a_k \in \mathbb{F}_q \setminus \{0\}$ such that $\sum_{i \in [k]} a_i \bw_i = \bzero$
\item (NO) there are no $\gamma \cdot k$ vectors in $\mathcal W$ that are linearly dependent
\end{itemize}
\end{framed}
We point out that in we require the vectors to have distinct colors only in the YES case; in the NO case, we assume that there are no $\gamma \cdot k$ linearly dependent vectors of arbitrary colors.

\subsection{Minimum Distance Problem}
In this subsection, we define the fixed-parameter variant of the minimum distance problem and other relevant parameterized problems. We actually define them as gap problems -- as later in the paper, we show the constant inapproximability of these problems. 

For every $\gamma\geqs 1$, we define the $\gamma$-gap minimum distance problem\footnote{In the parameterized complexity literature, this problem is referred to as the $k$-Even set problem \cite{DFVW99} and the input to the problem is (equivalently) given through the parity-check matrix, instead of the generator matrix as described in this paper.} as follows:

\begin{framed}
$\gamma$-Gap Minimum Distance Problem ($\mdp_{\gamma}$)

{\bf Input: } A matrix $\bA \in \mathbb{F}_2^{n \times m}$ and a positive integer $k \in \mathbb{N}$

{\bf Parameter: } $k$

{\bf Question: } Distinguish between the following two cases:
\begin{itemize}
\item (YES) there exists $\bx \in \mathbb{F}_2^m \setminus \{\bzero\}$ such that $\|\bA\bx\|_0 \leqs k$
\item (NO) for all $\bx \in \mathbb{F}_2^m \setminus \{\bzero\}$, $\|\bA\bx\|_0 > \gamma \cdot k$
\end{itemize}
\end{framed}

Next, for every $\gamma\geqs 1$, we define the $\gamma$-gap maximum likelihood decoding problem\footnote{The maximum likelihood decoding problem is also equivalently known in the literature as the nearest codeword problem.} as follows:

\begin{framed}
$\gamma$-Gap Maximum Likelihood Decoding Problem ($\gapvec_{\gamma}$)

{\bf Input: } A matrix $\bA \in \mathbb{F}_2^{n \times m}$, a vector $\by \in \mathbb{F}_2^n$ and a positive integer $k \in \mathbb{N}$

{\bf Parameter: } $k$

{\bf Question: } Distinguish between the following two cases:
\begin{itemize}
\item (YES) there exists $\bx \in \cB_{m}(\mathbf{0}, k)$ such that $\bA\bx = \by$
\item (NO) for all $\bx \in \cB_{m}(\mathbf{0}, \gamma k)$, $\bA\bx \ne \by$
\end{itemize}
\end{framed}

For brevity, we shall denote the exact version (i.e., $\gapvec_1$) of the problem as $\kvec$.

It should be noted that the Odd Set problem discussed in the introduction is closely related to \gapvec; in particular, the only different is that, in \odds, $\by$ is not part of the input but is always fixed as $\vone$, the all-ones vector. Indeed, it is not hard to see that our parameterized hardness of approximation for \gapvec\ also transfers to that of \gapodds. This is formulated in Appendix~\ref{app:oddset}.

We also define the \gapvec\ problem over larger (constant) field $\mathbb{F}_p$ below; this version of the problem will be used in proving hardness of Nearest Vector Problem. In this version, we have an additional requirement that, in the YES case, the solution $\bx$ must be a $\{0, 1\}$-vector. (Note that this is automatically the case for \gapvec\ over $\F$.)

\begin{framed}
$\gamma$-Gap Maximum Likelihood Decoding Problem over $\mathbb{F}_p$ ($\gapvec_{\gamma, p}$)

{\bf Input: } A matrix $\bA \in \mathbb{F}_p^{n \times m}$, a vector $\by \in \mathbb{F}_p^n$ and a positive integer $k \in \mathbb{N}$

{\bf Parameter: } $k$

{\bf Question: } Distinguish between the following two cases:
\begin{itemize}
\item (YES) there exists $\bx \in \{0, 1\}^m$ with $\|\bx\|_0 \leqs k$ such that $\bA\bx = \by$
\item (NO) for all $\bx \in \mathbb{F}_q^m$ such that $\|\bx\|_0 \leqs \gamma k$, $\bA\bx \ne \by$
\end{itemize}
\end{framed}

Finally, we introduce a ``sparse'' version of the \gapvec\ problem called the sparse nearest codeword problem, and later in the paper, we show a reduction from \gapvec\ to this problem, followed by a reduction from this problem to \mdp. As its name suggest, the sparse nearest codeword problem priorities not only the Hamming distance of the codeword $\bA\bx$ to the target vector $\by$ but also the ``sparsity'' (i.e. Hamming weight) of $\bx$. Formally, for every $\gamma\geqs 1$, we define the $\gamma$-gap sparsest nearest codeword problem as follows:

\begin{framed}
$\gamma$-Gap Sparse Nearest Codeword Problem ($\sncp_{\gamma}$)

{\bf Input: } A matrix $\bA \in \mathbb{F}_2^{n \times m}$, a vector $\by \in \mathbb{F}_2^n$ and a positive integer $k \in \mathbb{N}$

{\bf Parameter: } $k$

{\bf Question: } Distinguish between the following two cases:
\begin{itemize}
\item (YES) there exists $\bx \in \mathbb{F}_2^m$ such that $\|\bA\bx - \by\|_0 + \|\bx\|_0 \leqs k$
\item (NO) for all $\bx \in \mathbb{F}_2^m$, $\|\bA\bx - \by\|_0 + \|\bx\|_0 > \gamma \cdot k$
\end{itemize}
\end{framed}

\subsection{Shortest Vector Problem and Nearest Vector Problem}

In this subsection, we define the fixed-parameter variants of the shortest vector and nearest vector problems. As in the previous subsection, we define them as gap problems, for the same reason that later in the paper, we show the constant inapproximability of these two problems. 

Fix $p\in\mathbb{R}_{\geqs 1}$. For every $\gamma\geqs 1$, we define the $\gamma$-gap shortest vector problem in the $\ell_p$-norm{\footnote{Note that we define $\snvp$ and $\svp$ problems in terms of $\ell^p_p$, whereas traditionally, it is defined in terms of $\ell_p$. However, it is sufficient for us to work with the $\ell^p_p$ variant, since an $\alpha$-factor inapproximability in $\ell^p_p$ translates to an $\alpha^{1/p}$-factor inapproximabillity in the $\ell_p$ norm, for any $\alpha \geqs 1$}} as follows:
\begin{framed}
	$\gamma$-Gap Shortest Vector Problem ($\svp_{p,\gamma}$)
	
	{\bf Input: } A matrix $\bA \in \mathbb{Z}^{n \times m}$ and a positive integer $k \in \mathbb{N}$
	
	{\bf Parameter: } $k$
	
	{\bf Question: } Distinguish between the following two cases:
	\begin{itemize}
		\item (YES) there exists $\bx \in \mathbb{Z}^{m} \setminus \{\bzero\}$ such that $\|\bA\bx\|^p_p \leqs k$
		\item (NO) for all $\bx \in \mathbb{Z}^m \setminus \{\bzero\}$, $\|\bA\bx\|^p_p > \gamma \cdot k$
	\end{itemize}
\end{framed}

For every $\gamma\geqs 1$, we define the $\gamma$-gap nearest vector problem in the $\ell_p$-norm as follows:

\begin{framed}
	$\gamma$-Gap Nearest Vector Problem ($\snvp_{p,\gamma}$)
	
	{\bf Input: } A matrix $\bA \in \mathbb{Z}^{n \times m}$, vector $\by \in \mathbb{Z}^n$ and a positive integer $k \in \mathbb{N}$
	
	{\bf Parameter: } $k$
	
	{\bf Question: } Distinguish between the following two cases:
	\begin{itemize}
		\item (YES) there exists $\bx \in \Z^{m}$ such that $\|\bA\bx - \by\|^p_p \leqs k$
		\item (NO) for all $\bx \in \Z^m$, $\|\bA\bx - \by\|^p_p > \gamma \cdot k$
	\end{itemize}
\end{framed}

\subsection{Error-Correcting Codes}
An error correcting code $C$ over alphabet $\Sigma$ is a function $C: \Sigma^m \to \Sigma^h$ where $m$ and $h$ are positive integers which are referred to as the {\em message length} (aka \emph{dimension}) and {\em block length} of $C$ respectively. Intuitively, $C$ encodes an original message of length $m$ to an encoded message of length $h$.
The {\em distance} of a code, denoted by $d(C)$, is defined as $\underset{x \ne y \in \Sigma^m}{\min} \|C(x)- C(y)\|_0$, i.e., the number of coordinates on which $C(x)$ and $C(y)$ disagree.
We also define the systematicity of a code as follows: Given $s \in \mathbb N$, a code $C:\Sigma^m\to \Sigma^{h}$ is {\em $s$-systematic} if there exists a size-$s$ subset of $[h]$, which for convenience we identify with $[s]$, such that for every $x \in \Sigma^{s}$ there exists $w \in \Sigma^m$ in which $x = C(w)\mid_{[s]}$.
We use the shorthand $[h,m,d]_{|\Sigma|}$ to denote a code of message length $m$, block length $h$, and distance $d$. 

Additionally, we will need the following existence and efficient construction of BCH codes for every message length and distance parameter.

\begin{theorem}[BCH Code~\cite{H59,BR60}] \label{thm:bch}
For any choice of $h, d \in \mathbb{N}$ such that $h + 1$ is a power of two and that $d \leqs h$, there exists a linear code over $\F$ with block length $h$, message length $h - \left\lceil\frac{d-1}{2}\right\rceil\cdot\log (h + 1)$ and distance $d$. Moreover, the generator matrix of this code can be computed in $\poly(h)$ time. 
\end{theorem}

Finally, we define the tensor product of codes which will be used later in the paper. Consider two linear codes $C_1 \subseteq \F^m$ (generated by $\mathbf{G}_1 \in \F^{m \times m'}$) and ${C}_2 \subseteq \F^n$  (generated by $\mathbf{G}_2 \in \F^{n \times n'}$). Then the tensor product of the two codes ${C}_1 \otimes {C}_2 \subseteq \F^{m \times n}$ is defined as 
\begin{equation*}
{C}_1 \otimes {C}_2 = \{\mathbf{G}_1\mathbf{X}\mathbf{G}^\top_2  | \mathbf{X} \in \F^{m' \times n '}\}.
\end{equation*}
We will only need two properties of tensor product codes. First, the generator matrix of the tensor products of two linear codes $C_1, C_2$ with generator matrices $\mathbf{G}_1, \mathbf{G}_2$ can be computed in polynomial time in the size of $\mathbf{G}_1, \mathbf{G}_2$. Second, the distance of ${C}_1 \otimes {C}_2$ is exactly the product of the distances of the two codes, i.e., $$d({C}_1 \otimes {C}_2) = d(C_1) d(C_2).$$

%%% Local Variables:
%%% mode: latex
%%% TeX-master: "JACM submission"
%%% End:

\section{Parameterized Inapproximability of Linear Dependent Set}
\label{sec:lin-dep}
In this section, we show that the Linear Dependent Set problem has no constant factor FPT approximation algorithm unless $\W[1]=\FPT$. More formally, we prove the following:

\begin{theorem}\label{thm:ldsmain}
For every $\gamma \geqs 1$, $\gapLDS_{\gamma}$ and $\gapLDS_\gamma^\col$ are \W[1]-hard.
\end{theorem}

The proof consists of two steps. First, we will reformulate Lin's reduction for the \text{Biclique} problem in terms of hardness of \gapbsmd. Then, we reduce \gapbsmd to our target problem \gapLDS.

\subsection{Translating \textsc{One-Sided Biclique} to \gapbsmd}\label{sec:biclique-to-bsmd}

In the first step of our proof, we will show that $\gapbsmd$ is \W[1]-hard to approximate to within any constant factor, as stated more precisely below.

\begin{theorem} \label{thm:bsmd}
For every $\gamma \geqs 1$, $\gapbsmd_\gamma$ is \W[1]-hard.
\end{theorem}

Our result relies crucially on the recent \W[1]-hardness of approximation result for the \textsc{One-Sided Biclique} problem by Lin~\cite{Lin15}. Recall that, in \textsc{One-Sided Biclique}, we are given a bipartite graph $G$ and an integer $s$ and the goal is to find $s$ left vertices with maximum number of common neighbors. The following theorem is the main result of Lin~\cite{Lin15} for \textsc{One-Sided Biclique}.

\begin{theorem}[{\cite[Theorem~1.3]{Lin15}}]\label{thm:gapbiclique}
There is a polynomial time algorithm $\mathbb A$ such that, given a graph
$G$ with $n$ vertices and $k\in \mathbb N$ with $\lceil n^{\frac{6}{k+6}}\rceil> (k+6)!$ and $6\mid k+1$, it outputs a bipartite graph $G' = (A \dotcup B, E)$ and $s = \binom{k}{2}$ satisfying:
\begin{enumerate}
\item (YES) If $G$ contains a $k$-clique, then there are $s$ vertices in $A$ with at least $\lceil {n^{\frac{6}{k+1}}}\rceil$ common neighbors in $B$;
\item (NO) If\ $G$ does not contain a $k$-clique, any $s$ vertices in $A$ have at most
    $(k+1)!$ common neighbors in $B$.
\end{enumerate}
\end{theorem}

Another ingredient of our reduction is a simple observation regarding the size of bipartite graphs with prescribed minimum degrees, conditioned on the fact that any small subset of left vertices have small number of neighbors. This is stated below.

\begin{claim}\label{cl:bi}
For any $s,\ell,h\in\mathbb{N}$, let $(X\cup Y,E_W)$ be a non-empty bipartite graph such that
\begin{enumerate}[(i)]
\item every vertex in $X$ has at least $h$ neighbors,
\item every vertex in $Y$ has at least $s$ neighbors, and,
\item every $s$-vertex set of $X$ has at most $\ell$ common neighbors.
\end{enumerate}
Furthermore, the parameters $h,\ell$ and $s$ satisfy $h/\ell \ge \gamma^s s^s$. Then, $|E_W|\ge(h/\ell)^{1/s} \ge \gamma \cdot hs$.
\end{claim}
\begin{proof}[Proof of Claim~\ref{cl:bi}]
Consider any vertex $u\in X$. By (i), $u$ has at least $h$ neighbors in $Y$, so $|Y|\ge h$.
By (ii), for every $v\in Y$, $v$ has at least $s$ neighbors in $X$. If $\binom{|X|}{s}\ell<|Y|$, then there must exist a $s$-vertex set in $X$ which has more than $\ell$ common neighbors in $Y$. Thus, we must have
\[
|X|^s\ge \binom{|X|}{s}\ge \frac{|Y|}{\ell}\ge \frac{h}{\ell}.
\]
By (i) and our choice of parameters $h,\ell,s$, we can conclude that $|E_W|\ge h|X|\ge(h/\ell)^{1/s} \cdot h \ge \gamma \cdot hs$, as desired.
\end{proof}

With Theorem~\ref{thm:gapbiclique} and Claim~\ref{cl:bi} in place, we can prove Theorem~\ref{thm:bsmd} simply by using the reduction from Theorem~\ref{thm:gapbiclique} and choosing an appropriate value of $h$; the guarantee in the NO case would then follow from Claim~\ref{cl:bi}.

\begin{proof}[Proof of Theorem~\ref{thm:bsmd}]
We reduce from the \textsc{$k$-Clique} problem which is well-known to be \W[1]-complete. Let $(G, k)$ be an instance of \textsc{$k$-Clique} and $n$ be the number of vertices in $G$. Without loss of generality, we can assume that $6\mid k+1$ and $\lceil n^{\frac{6}{k+6}}\rceil > (k+6)! \cdot (\gamma \cdot k^2)^{k^2}$. Using the reduction in Theorem~\ref{thm:gapbiclique}, we can produce $(G', s = \binom{k}{2})$ in polynomial time with the guarantees as in the theorem. We then set $h = (k+6)! \cdot (\gamma \cdot k^2)^{k^2}$ and let $(H, s, h)$ be our instance of $\gapbsmd_\gamma$. We will next show that this is indeed a valid reduction from {$k$-Clique} to $\gapbsmd_\gamma$.

{\bf (YES Case)} Suppose that $G$ contains a $k$-clique. Then, Theorem~\ref{thm:gapbiclique} guarantees that $G'$ contains a complete bipartite subgraph with $s$ left vertices and $h$ right vertices as desired.

{\bf (NO Case)} Suppose that $G$ does not contain a $k$-clique. Now, consider any non-empty subgraph $H$ of $G'$ such that every left vertex of $H$ has at least $h$ neighbors and every right vertex of $H$ contains at least $s$ neighbors, i.e., $H$ satisfies condition (i) and (ii) in Claim~\ref{cl:bi}. Furthermore, since $G$ does not contain a $k$-clique,  guarantees that every $s$ vertices in $A$ contains at most $\ell = (k + 1)!$ common neighbors. It can be easily verified that our setting of parameters $h,\ell$ and $s$ satisfies the inequality $h/\ell \ge \gamma^ss^s$. Hence, by applying Claim~\ref{cl:bi} on $H$, the number of edges in $H$ must be at least
$\gamma \cdot (hs)$. This means that $(H, s, h)$ is a NO instance of $\gapbsmd_\gamma$ as desired.
\end{proof}

\subsection{Reducing \gapbsmd\ to \gapLDS}\label{sec:bsmd-to-LDS}

We now move on to the next step of our proof, which is the reduction from \gapbsmd\ to \gapLDS.

Since the reduction itself will be used in the subsequent proofs (with different parameter selections), we also state it separately below. We remark that the reduction as stated below goes from $\gapbsmd_\gamma$ to the uncolored version of the problem ($\gapLDS_\gamma$); we will state how to go from here to the colored version later on.

\begin{theorem} \label{thm:lds-uncolored-red}
Let $\gamma \geqs 1$ be any constant. There is a polynomial time algorithm that, given an instance $(G, s, h)$ of $\gapbsmd_\gamma$ where $G$ contains $n$ vertices and any prime power $q > n$, produces an instance $(\cW \subseteq \mathbb{F}_q^m, k = hs)$ of $\gapLDS_\gamma$ such that
\begin{itemize}
\item (YES) If $(G, s, h)$ is a YES instance of $\gapbsmd_\gamma$, then $(\cW, k)$ is a YES instance of $\gapLDS_\gamma$.
\item (NO) If $(G, s, h)$ is a NO instance of $\gapbsmd_\gamma$, then $(\cW, k)$ is a NO instance of $\gapLDS_\gamma$.
\end{itemize}
\end{theorem}

\begin{proof}
Assume that an instance  $(G = (L \dot\cup R, E), s, h)$ of $\gapbsmd_\gamma$ and a prime power $q > |L| + |R|$ are given. Before we construct $\cW$, let us first define additional notation. We identify vertices in $L \dot\cup R$ with distinct elements of $\mathbb{F}_q$. Let $B:= s + h$ and let $\iota : L\cup R\to \mathbb{F}_q^B$ be defined as follows.
\begin{itemize}
\item for each $v\in R$, $\iota(v):=(1,v,\ldots ,v^{h-2}) \circ \textbf{0}_{B-h+1}$,
\item for each $u\in L$, $\iota(u):=(1,u,\ldots ,u^{s-2}) \circ \textbf{0}_{B-s+1}$.
\end{itemize}

By a well-known property of Vandermonde matrices, any $h-1$ vectors in $\iota(R)$ are linearly independent and any $h$ vectors from $\iota(R)$ are linearly dependent.
To summarize, we have
\begin{itemize}
\item[(R1)] For all $I\in\binom{R}{h}$,  the vectors $\{\iota(v) : v\in I\}$ are linearly dependent.
\item[(R2)] For all  $I\in\binom{R}{h-1}$, the vectors $\{\iota(v) : v\in I\}$ are linearly independent.
\end{itemize}

Similarly, we also have
\begin{itemize}
\item[(L1)] For all $I\in\binom{L}{s}$,  the vectors $\{\iota(u) : u\in I\}$ are linearly dependent.
\item[(L2)] For all  $I\in\binom{L}{s-1}$, the vectors $\{\iota(u) : u\in I\}$ are linearly independent.
\end{itemize}

Let $m=qB$ and consider vectors from $\mathbb{F}_q^m=\mathbb{F}_q^{qB}$, which can be seen as the concatenation of $q$ blocks, each of $B$ coordinates.
For $x\in \mathbb{F}_q^m$, we use the notation $\bx^{(i)}$ to refer to the $i$-block, i.e. the $B$-dimensional vector given by coordinates $(i-1)B+1, (i-1)B+2, \dots, iB$.

\bigskip

\noindent\textbf{Construction of $(\cW,k)$.} First, we let $k = hs$. Then, for each $(u, v) = e \in E$ (where $u \in L, v \in R$), we introduce a vector $\bw_e\in \mathbb{F}_q^{qB}$ such that

\begin{itemize}
\item[(W1)] for all $i\in [q]\setminus\{v,u\}$, $\bw_e^{(i)}=\textbf{0}_B$,
\item[(W2)] $\bw_e^{(v)}=\iota(u)$,
\item[(W3)] $\bw_e^{(u)}=\iota(v)$.
\end{itemize}

That is, we can imagine $w_e$ as being partitioned $q$ blocks of $B$ coordinates, with the representation of $u$ appearing in the $v$-th block and the representation of $v$ appearing in the $u$-th block. Note the use of $u$ and $v$ in the definition: the $v$-th block on its own describes both $v$ (by its position) and $u$ (by its content), and similarly the $u$-th block also describes both endpoints of $e$. We then let
\[
\cW:=\{w_e : e\in E\}.
\]

Obviously, $(\cW, k)$ can be computed in polynomial time. We next argue its correctness.

\textbf{(YES case)} Suppose $(G, s, h)$ is a YES instance of $\gapbsmd_\gamma$. There exist a set $X\in\binom{L}{s}$ and a set $Y\in\binom{R}{h}$ such that for all
$u\in X$ and $v\in Y$, $(u, v) \in E$. By (R1) and (L1), there exists $b_u \in \mathbb{F}_q$ for each $u \in X$ and $b_v \in \mathbb{F}_q$ for each $v \in Y$ such that
\[
\sum_{u \in X} b_u \iota(u)=\textbf{0}_B
\text{ and }
\sum_{v \in Y} b_v \iota(v)=\textbf{0}_B.
\]
By (R2) and (L2), we deduce that, for all $u \in X$ and $v \in Y$, $b_u \neq 0$ and $b_v \neq 0$. We now claim that $\{\bw_{(u, v)}\}_{u \in X, v \in Y}$ is the set of desired vectors, with the coefficient of $\bw_{(u, v)}$ being $b_ub_v \ne 0$. In other words, we are left to show that
\[
\sum_{u \in X, v \in Y} b_u b_v \bw_{(u, v)} = \bzero_m.
\]
To see that this is true, let $\bw = \sum_{i\in[s],j\in[h]} b_u b_v \bw_{\{u_i, v_j\}}$. It is easy to check that
\begin{itemize}
\item by (W1), for every $z\in [q] \setminus (X\cup Y)$, $w^{(z)}=\textbf{0}_B$,
\item by (W2), for every $v \in Y$, $w^{(v)}=\sum_{u \in X} b_u b_v\iota(u) = b_v \sum_{u \in X} b_u \iota(u) = \textbf{0}_B$,
\item by (W3), for every $u \in X$, $w^{(u_i)}=\sum_{j\in [h]}a_ib_j\iota(v_j)=a_i\sum_{j\in [h]}b_j\iota(v_j)=\textbf{0}_B$.
\end{itemize}
Hence, we have completed the proof for the YES case.

\textbf{(NO case)} Suppose $(G, s, h)$ is a NO instance of $\gapbsmd_\gamma$. Let $W\subseteq\mathcal W$ be a set of  vectors that are linearly dependent. We define two vertex sets and their edge set as follows. Let
 \[
 X:=\{u\in L : \text{there exists $v\in R$ such that $w_{(u, v)}\in W$}\},
 \]
 \[
 Y:=\{v\in R : \text{there exists $u\in L$ such that $w_{(u, v)}\in W$}\},
 \]
 and
 \[
E_W:=\{ e \in E:  \text{$w_e \in W$ }\}.
 \]
Note that $X$ and $Y$ are not empty because $W$ is non-empty. By (R2) and (W3), for every $u\in X$, there exist at least $h$ vertices in $Y$ that are adjacent to $u$, i.e. $|N(u)\cap Y|\ge h$. Similarly, by (L2) and (W2), for every $v\in Y$, we have $|N(v)\cap X|\ge s$. Hence, by the guarantee in the NO case of $\gapbsmd_\gamma$, we can conclude that $\gamma \cdot sh \leqs |E_W| = |W|$ as desired.
\end{proof}

\subsubsection{Reducing Uncolored LDS to Colored LDS}\label{sec:colorLDS}

In this section, we show a simple reduction from the uncolored version of LDS to the colored version of LDS. As is usual in such a reduction, we will need the definition of perfect hash families and an efficient construction stated below.

\begin{definition}
An $(n, k)$-perfect hash family is a collection $\cF$ of functions from $[n]$ to $[k]$ such that, for every subset $S \subseteq [n]$ of size $k$, there exists $f \in \cF$ that maps every $S$ to distinct elements in $[k]$, i.e., $f(S) = [k]$.
\end{definition}

\begin{theorem}[\cite{NSS95}] \label{thm:perf-hash}
There exists an algorithm that, for any $n, k \in \mathbb{N}$, constructs an $(n, k)$-perfect hash family in time $2^{O(k)} poly(n)$.
\end{theorem}

If we use perfect hash families to reduce $\gapLDS$ to $\gapLDS^{\col}$ in a straightforwad manner, we will end up with a Turing reduction, i.e., we will produce multiple instances of $\gapLDS^{\col}$. Our observation here is that these instances can be ``merged'' into a single instance, i.e., by shifting the vectors appropriately so that the coordinates of vectors from different instances are not overlap:

\begin{lemma} \label{lem:lds-coloring}
There exists an algorithm reduction that takes in $\cW \subseteq \mathbb{F}^{m}_q$ and an integer $k$, runs in $2^{O(k)} poly(m, |\cW|)$ time, and outputs $\cW' \subseteq \mathbb{F}^{m'}_q$ and a coloring $c: \cW' \to [k]$ such that
\begin{itemize}
\item (YES) if $(\cW, k)$ is a YES instace of $\gapLDS_\gamma$, then $(\cW', k, c)$ is a YES instace of $\gapLDS^{\col}_\gamma$;
\item (NO) if $(\cW, k)$ is a NO instace of $\gapLDS_\gamma$, then $(\cW', k, c)$ is a NO instace of $\gapLDS^{\col}_\gamma$;
\end{itemize}
\end{lemma}

\begin{proof}
Let $(\cW, k)$ be any instance of $\gapLDS_\gamma$, and let $n$ denote $|\cW|$. We use Theorem~\ref{thm:perf-hash} to construct an $(n, k)$-perfect hash family $\cF = \{f_1, \dots, f_R\}$ where $R = 2^{O(k)} poly(n)$. For every $\bw \in \cW$ and $j \in [R]$, we add a vector $\vzero_{m(j - 1)} \circ \bw \circ \vzero_{m(R - j)} \in \F^{mR}$ to $\cW'$ and color this vector by $f_j(\bw)$. Finally, $k$ remains the same as before.

It is obvious that the reduction runs in $2^{O(k)} poly(n)$ time. We now argue its correctness.

\textbf{(YES Case)} Suppose that $(\cW, k)$ is a YES instance of $\gapLDS_\gamma$, i.e., there exist $a_1, \dots, a_k \in \mathbb{F}_q \setminus \{0\}$ such that $a_1 \bw_1 + \cdots + a_k \bw_k = 0$. Since $\cF$ is a perfect hash family, there exists $j \in [R]$ such that $f_j(\{\bw_1, \dots, \bw_k\}) = [k]$. In this case, we have $\sum_{i \in [k]} a_i (\vzero_{m(j - 1)} \circ \bw_i \circ \vzero_{m(R - j)}) = \bzero$ and that the vectors $\vzero_{m(j - 1)} \circ \bw_1 \circ \vzero_{m(R - j)}, \dots, \vzero_{m(j - 1)} \circ \bw_k \circ \vzero_{m(R - j)}$ are of different colors. Hence, $(\cW', k, c)$ is a YES instance of $\gapLDS_\gamma^\col$.

\textbf{(NO Case)} Suppose that $(\cW, k)$ is a NO instance of $\gapLDS_\gamma$. Consider any $W' \subseteq \cW'$ such that the vectors in $W'$ are linearly dependent; we may pick such a set that is minimum, i.e., for every $\bw' \in W'$, there exists a coefficient $a_{\bw'}$ so that $\sum_{\bw' \in W'} a_{\bw'} \bw' = \vzero$.

Consider any element of $W'$; suppose that it is of the form $\vzero_{m(j - 1)} \circ \bw^* \circ \vzero_{m(R - j)}$ for some $j \in [R]$. Let $W$ be $\{\bw \in \cW : \vzero_{m(j - 1)} \circ \bw \circ \vzero_{m(R - j)} \in W'\}$. By restricting the equation $\sum_{\bw' \in W'} a_{\bw'} \bw' = \vzero$ only to the coordinates $m(j - 1) + 1, \dots, mj$, we can conclude that the vectors in $W$ are linearly dependent. Hence, we must have $|W'| \geqs |W| > \gamma k$; that is, $(\cW, k, c)$ is a NO instance of $\gapLDS_\gamma^\col$ as desired.
\end{proof}

Combining Theorem~\ref{thm:lds-uncolored-red} and Lemma~\ref{lem:lds-coloring}, we can get the following theorem, which implies the \W[1]-hardness of $\gapLDS_\gamma^\col$.

\begin{theorem} \label{thm:lds-red}
Let $\gamma \geqs 1$ be any constant. There is a polynomial time algorithm that, given an instance $(G, s, h)$ of $\gapbsmd_\gamma$ where $G$ contains $n$ vertices and any prime power $q > n$, produces an instance $(\cW \subseteq \mathbb{F}_q^m, k = hs, c)$ of $\gapLDS_\gamma$ such that
\begin{itemize}
\item (YES) If $(G, s, h)$ is a YES instance of $\gapbsmd_\gamma$, then $(\cW, k, c)$ is a YES instance of $\gapLDS_\gamma^\col$.
\item (NO) If $(G, s, h)$ is a NO instance of $\gapbsmd_\gamma$, then $(\cW, k, c)$ is a NO instance of $\gapLDS_\gamma^\col$.
\end{itemize}
\end{theorem}

%%% Local Variables:
%%% mode: latex
%%% TeX-master: "JACM submission"
%%% End:

\section{Parameterized Inapproximability of Maximum Likelihood Decoding} \label{sec:csp-to-gapvec}

In this section, we will show the parameterized intractability of $\gapvec$ as stated below.

\begin{theorem}\label{thm:MLDmain}
	For every $\gamma \geqs 1$ and any prime number $p$, $\gapvec_{\gamma, p}$ is \W[1]-hard.
\end{theorem}

We will divide the section into two parts. In the first part, we will give a simpler proof that only yields a hardness of approximation with factor $3 - \varepsilon$ for any $\varepsilon > 0$, and we only focus on the case $p = 2$ for simplicity. We note that this already suffices for proving hardness for Even Set problem. (In fact, any inapproximability result with factor greater than two suffices; see Lemma~\ref{lem:sncp-to-mdp}.)

Next, in the second part, we add an additional step in the proof that allows us to prove hardness of approximation with any constant factor and every prime field. We note here that, while this additional step is not used in proving hardness of Even Set, the technique not only gives the better inapproximability factor for $\gapvec$ but is also crucial in proving hardness of the Shortest Vector Problem (see Appendix~\ref{sec:csp-to-snvp}). 

\subsection{$(3 - \varepsilon)$ Factor Inapproximability of Maximum Likelihood Decoding}\label{sec:simple-analysis}

In this subsection, we will show the inapproximability of \kvec\ over $\F$ for any constant factor less than three. More formally, we show the following:

\begin{theorem}\label{thm:oddweak3}
For any constant $\varepsilon > 0$, $\gapvec_{3 - \varepsilon}$ is \W[1]-hard.
\end{theorem}

\begin{proof}
We will reduce from $\gapbsmd_3$, which is \W[1]-hard due to Theorem~\ref{thm:bsmd}. Let $(G = (L\dot\cup R, E), s, h)$ be an instance of $\gapbsmd_3$. We first run the reduction in Theorem~\ref{thm:lds-red} with $q = 2^{\lceil \log(|L| + |R|) \rceil}$. This gives us an instance $(\cW \subseteq \mathbb{F}_{2^d}^m, k, c)$ of $\gapLDS_3^\col$. We use $n$ to denote $|\cW|$.

We now describe how we construct the instance $(\bA \in \mathbb{F}_2^{m' \times n'}, \by \in \mathbb{F}_2^{m'}, k)$ of $\gapvec_{3 - \varepsilon}$ where $m' = md + k$ and $n' = (2^d - 1)n$. First, the parameter $k$ remains the same from the $\gapbsmd_3^\col$. Second, $\by$ is the $m'$-dimensional vector whose first $k$ coordinates are ones and the remaining coordinates are zeros, i.e., $\by = \vone_k \circ \vzero_{md}$.

To define $\bA$, we need to introduce some notation. First, recall that the elements of the field $\mathbb{F}_{2^d}$ can be viewed as $d$-dimensional $\mathbb{F}_2$-vectors. In other words, there is a map $f: \mathbb{F}_{2^d} \to \mathbb{F}_2^d$ such that $f(x + y) = f(x) + f(y)$ for all $x, y \in \mathbb{F}_{2^d}$, and $f(x) = \bzero_d$ iff $x = 0$. We additionally define $F : \mathbb{F}_{2^d}^m\to \mathbb{F}_2^{md}$ by $F(\bv)=f(\bv[1])\circ \dots \circ  f(\bv[m])$. Again, we have $F(\bu + \bv) = F(\bu) + F(\bv)$ for all $\bu, \bv \in \mathbb{F}_{2^d}^m$, and $F(\bv) = \bzero_{md}$ iff $\bv = \bzero_m$.

Moreover, for every $i \in [k]$, let $\be_i$ be the $k$-dimensional vector with one at the $i$-th coordinate and zero elsewhere. We identify the column indices of $\bA$ by $\cW \times (\mathbb{F}_{2^d} \setminus \{0\})$. Then, we construct $\bA$ by letting its $(\bw, a)$-column be
\begin{align*}
\bA[(\bw, a)] := \be_{c(\bw)} \circ F(a \cdot \bw).
\end{align*}
This completes our reduction description. It is simple to verify that the reduction runs in polynomial time. We now move on to prove the correctness of the reduction.

\textbf{(YES Case)} Suppose that $(G, s, h)$ is a YES instance of $\gapbsmd_3$. From Theorem~\ref{thm:lds-red}, there exist $\bw_1, \dots, \bw_k \in \cW$ all of different colors and non-zero $a_1, \dots, a_k \in \mathbb{F}_{2^d} \setminus \{0\}$ such that $\sum_{i \in [k]} a_i \cdot \bw_i = \bzero$. Let $\bx \in \mathbb{F}_2^{n'}$ such that $\bx[(\bw_i, a_i)] = 1$ for all $i \in [k]$ and all other coordinates of $\bx$ are zero. Clearly, $\|\bx\|_0 = k$ and 
\begin{align*}
\bA\bx = \sum_{i \in [k]} \bA[(\bw_i, a_i)] = \sum_{i \in [k]} \be_{c(\bw)} \circ F\left(a_i \cdot \bw_i\right) = \vone_k \circ F\left(\sum_{i \in [k]} a_i \cdot \bw_i\right) = \vone_k \circ F(\bzero) = \vone_k \circ \bzero_{md} = \by,
\end{align*}
which means that $(\bA, \by, k)$ is indeed a YES instance.

\textbf{(NO Case)} Suppose that $(G, s, h)$ is a NO instance of $\gapbsmd_3$. From Theorem~\ref{thm:lds-red}, $(\cW, k, c)$ is a NO instance of $\gapLDS_3^\col$. Suppose for the sake of contradiction that $(\bA, \by, k)$ is not a NO instance of $\gapvec_{3 - \varepsilon}$. That is, there exists $\bx \in \mathbb{F}_2^{n'}$ such that $\bA \bx = \by$ and $\|\bx\|_0 \leqs (3 - \varepsilon) k < 3k$.

For every $i \in [k]$, let us define $X_i$ as
\begin{align*}
X_i := \{(\bw, a) \in \cW \times (\mathbb{F}_{2^d} \setminus \{0\}) : \bx[(\bw, a)] = 1\}.
\end{align*}
We can write $\bA\bx$ as
\begin{align*}
\bA\bx = \sum_{i \in [k]} \sum_{(\bw, a) \in X_i} \be_i \circ F(a \cdot \bw) = \left(\sum_{i \in [k]} |X_i|\be_i\right) \circ F\left(\sum_{i \in [k]} \sum_{(\bw, a) \in X_i} a \cdot \bw\right). 
\end{align*}

Since $\bA\bx = \by$, we must have $|X_i| \equiv 1 \Mod{2}$ for all $i \in [k]$ and 
\begin{align} \label{eq:sum-3inapprox}
\sum_{i \in [k]} \sum_{(\bw, a) \in X_i} a \cdot \bw = \bzero_m.
\end{align}

Moreover, observe that $\|\bx\|_0 = \sum_{i \in [k]} |X_i|$. Since $\|\bx\|_0 < 3k$ and $|X_i| \equiv 1 \Mod{2}$ for all $i \in [k]$, there must be $i^* \in [k]$ such that $|X_{i^*}| = 1$. Let $(\bw^*, a^*)$ be the unique element of $X_{i^*}$. Notice that $\bw^*$ appears only once in the left hand side of~\eqref{eq:sum-3inapprox} with coefficient $a^* \ne 0$; as a result, this is a non-empty linear combination of less than $3k$ vectors in $\cW$. Hence, there are less than $3k$ vectors in $\cW$ that are linearly dependent, which contradicts the fact that $(\cW, k, c)$ is a NO instance of $\gapLDS_3^\col$.

Thus, $(\cW, k, c)$ must be a NO instance of $\gapvec_{3 - \varepsilon}$ as desired.
\end{proof}

\subsection{Every Constant Factor Inapproximability of Maximum Likelihood Decoding}
\label{sec:full-result}

In this section, we will prove our main result of this section, i.e., Theorem~\ref{thm:MLDmain}. 

To demonstrate the main additional idea, let us recall why the proof in the previous section fails to give us the hardness of factor three. The reason is as follows: when $d > 2$, we can pick three non-zero elements $a, b, c \in \mathbb{F}_{2^d}$ whose sum is zero. We can then select any $\bw_1, \dots, \bw_k$ of different colors, and set $\bx[(\bw_i, a)], \bx[(\bw_i, b)], \bx[(\bw_i, c)]$ to be ones for all $i \in [k]$, and set the rest of coordinates of $\bx$ to be zero. Clearly, $\|\bx\|_0 = 3k$ and this gives
\begin{align*}
\bA\bx = \vone_k \circ F\left(\sum_{i \in [k]} (a + b + c) \bw_i\right) = \vone_k \circ \vzero_{md} = \by.
\end{align*}
That is, the fact that $a + b + c = 0$ allows us to zero out the coefficient of each $\bw_i$. Our fix to overcome this issue is rather straightforward. First, observe that we can write $\mathbb{F}_2^d \setminus \{0\} = C_1 \cup \dots \cup C_d$ such that no such ``problematic'' tuples $(a, b, c)$ appears in $C_i$, where $C_i$ is defined as $\{a \in \mathbb{F}_q \setminus \{0\} : f(a)[i] = 1\}$ (where $f$ is as defined in Theorem~\ref{thm:oddweak3}). In fact, this guarantees not only that any triplet in $C_i$ sums to non-zero, but also that any odd number of elements in $C_i$ sums to non-zero.

Now, the modification is very simple: instead of creating columns for $(\bw, a)$ for all $a \in \mathbb{F}_{2^d} \setminus \{0\}$, we will only create columns for $(\bw, a)$ for $a \in C_{g[c(\bw)]}$ where $g \in [d]^k$, i.e., we restrict the coefficients to only $C_{g[j]}$ for each color $j$. This helps us avoid ``problematic'' coefficients as described above. In particular, we can construct a instance $\bA_g$ for every choice of $c$. As in the reduction to $\gapLDS^\col$, we can merge the various instances corresponding to different choices of $g$ into a single instance using the shifting trick employed in the proof of Lemma \ref{lem:lds-coloring}.

For a general prime $p$, we can write $\mathbb{F}_{p^d}$ similarly as above into a union of subsets, such that each subset does not contain ``problematic'' tuples of elements, as stated below. Note that the definition of ``problematic'' is slightly more complicated for general $p$. Now, the tuple $(a_1, \dots, a_t) \in \mathbb{F}_{p^d}^t$ is ``problematic'' if we can find $b_{a_1}, \dots, b_{a_t} \in \mathbb{F}_p$ such that $b_{a_1} + \cdots + b_{a_t} \ne 0$ (over $\mathbb{F}_p$) but $b_{a_1} \cdot a_1 + \cdots b_{a_t} \cdot a_t = 0$ (over $\mathbb{F}_{p^d}$).  

\begin{definition} \label{def:ci}
For $q = p^d$ where $d \in \mathbb{N}$ and $p$ is a prime, let $f: \mathbb{F}_q \to \mathbb{F}_p^d$ be the isomorphism between $\mathbb{F}_q^{+}$ and the $\mathbb{F}_p$-vector space $\mathbb{F}_p^d$. For every $i \in [d]$ and $\alpha \in \mathbb{F}_p \setminus \{0\}$, we define $C_{(i, \alpha)} := \{a \in \mathbb{F}_q \setminus \{0\} : f(a)[i] = 1\}$. Observe that 
\begin{enumerate}[(i)]
\item $\mathbb{F}_q \setminus \{0\} = \bigcup_{i \in [d], \alpha \in \mathbb{F}_p \setminus \{0\}} C_{(i, \alpha)}$
\item for any $i \in [d], \alpha \in \mathbb{F}_p \setminus \{0\}$ and any $(b_a)_{a \in C_{(i, \alpha)}} \in (\mathbb{F}_p)^{C_{(i, \alpha)}}$ such that $\sum_{a \in C_{(i, \alpha)}} b_a \ne 0$, we have $\sum_{a \in C_{(i, \alpha)}} b_a \cdot a \ne 0$. \label{prop:}
\end{enumerate}
\end{definition}

With this definition, we can easily generalize the (sketched) reduction from $\F$ to $\mathbb{F}_p$. The properties of the reduction are summarized and proved below.

\begin{theorem} \label{thm:mld-strong}
Given an instance $(\cW \subseteq \mathbb{F}^m_{p^d}, k, c)$ of $\gapLDS_\gamma$ where $p$ is a prime, we can create an instance of $\gapvec_{\gamma, p}$ (with the same parameter $k$) in $O((dp)^k \cdot poly(|\cW|, m, p^d))$ time such that
\begin{itemize}
\item (YES) If $(\cW, k, c)$ is a YES instance of $\gapLDS_\gamma$, then the $\gapvec_{\gamma, p}$ is a YES instance.
\item (NO) If $(\cW, k, c)$ is a NO instance of $\gapLDS_\gamma$, then the $\gapvec_{\gamma, p}$ is a NO instance.
\end{itemize}
\end{theorem}

\begin{proof}
Let $(\cW \subseteq \mathbb{F}^m_{p^d}, k, c)$ be an instance of $\gapLDS_\gamma$. Let $n = |\cW|$, and $f: \mathbb{F}_q \to \mathbb{F}_p^d$ be the isomorphism between $\mathbb{F}_q^{+}$ and the $\mathbb{F}_p$-vector space $\mathbb{F}_p^d$. Furthermore, let $F : \mathbb{F}_{p^d}^m\to \mathbb{F}_p^{md}$ be defined by $F(\bv)=f(\bv[1])\circ \dots \circ  f(\bv[m])$. We will also find it convenient to define $\ell = d^k(p-1)^k$, which is the total number of distinct choices of $g$.

For every $g \in [\ell]$, we construct a matrix  $\bA_g \in \mathbb{F}_2^{m' \times md}$ where $m' = k + \ell md$. As before, we index the columns of $\bA_g$ with the set $I_g = \cup_{\bw \in \cW} \{\bw\} \times C_{g[c(\bw)]}$. Here, for any $[\bw,a] \in I_g$, we let the corresponding column be 
\begin{equation}
\bA_g[\bw,a] :=\be_{c^{-1}(\bw)} \circ \left(\bzero_{md(i-1)} \circ F(a\cdot\bw) \circ \bzero_{md(\ell - i)}\right)
\end{equation}
Finally, we define the matrix $\bA = [\bA_g ]_{g \in [\ell]} \in \mathbbm{F}^{m' \times n'}_p$ to be the concatenation of all the $\bA_g$ matrices, where $n' = \ell n$. Note that the above construction ensures that for any distinct pair of $g,g' \in [\ell]$, the column supports of the sub-matrices $\bA_g$ and $\bA_{g'}$ do not intersect in the coordinates $[k+1,n']$. We also define the target vector $\by = \vone_{k} \circ \bzero_{md\ell}$. We set $(\bA,\by,k)$ to be the $\gapvec_{\gamma,p}$ instance output by the reduction. Clearly, the reduction runs in $O((dp)^k \cdot poly(n, m, p^d))$ time. We next argue its correctness.

\textbf{(YES Case)} Suppose there exist $\bw_1, \dots, \bw_k \in \cW$ all of different colors and non-zero $a_1, \dots, a_k \in \mathbb{F}_{p^d} \setminus \{0\}$ such that $\sum_{i \in [k]} a_i \cdot \bw_i = \bzero$. We claim that $(\bA, \by, k)$ is a YES instance of $\gapvec_{\gamma, p}$. To see this, consider $g^*$ where $a_i$ belongs to $C_{g^*[i]}$ for all $i \in [k]$. . Also let $\bx \in \mathbb{F}_p^{n'}$ be such that $\bx_{g^*}[(\bw_i, a_i)] = 1$ for all $i \in [k]$, where $\bx_{g^*}$ is the vector $\bx$ restricted to coordinates corresponding to the sub-matrix $\bA_{g^*}$. We set all other coordinates of $\bx$ to zero. (Note that the column $(\bw_i, a_i)$ exists in $\bA_{g^*}$ because $a_i \in C_{g^*[i]}$.) Clearly, $\bx$ is a $\{0, 1\}$-vector with $\|\bx\|_0 = k$ and  
\begin{align*}
\bA\bx = \bA_{g^*}\bx_{g^*} = \sum_{i \in [k]} \bA_{g^*}[(\bw_i, a_i)] = \sum_{i \in [k]} \be_{c(\bw_i)} \circ F\left(a_i \cdot \bw_i\right) = \vone_k \circ F\left(\sum_{i \in [k]} a_i \cdot \bw_i\right) = \vone_k \circ F(\bzero) = \by,
\end{align*}
which means that $(\bA, \by, k)$ is indeed a YES instance of $\gapvec_{\gamma, p}$.

\textbf{(NO Case)} Suppose that $(\cW, k, c)$ is a NO instance of $\gapLDS_\gamma$. Consider any $\bx \in \mathbb{F}_p^{n'}$ such that $\bA\bx = \by$. Recall that for any $g \in [\ell]$, $\bx_g$ is the sub-vector of $\bx$ which acts on the sub-matrix $\bA_g$. Let us rewrite $\bA\bx$ as follows:
\begin{align*}
\bA\bx = \left(\sum_{i \in [k]} \left(\sum_{g \in [\ell]}\sum_{\bw \in c^{-1}(i), a \in C_{g[i]}} \bx_g[(\bw, a)] \right) \be_i \right) \circ \bv_1 \circ \bv_2 \circ \cdots \circ \bv_\ell
\end{align*}
where for any $g \in [\ell]$, the vector $\bv_g$ is the sub-vector of $\bA_g\bx_g$ which can be formally expressed as     
\begin{align*}
\bv_g = F\left(\sum_{\bw \in \cW} \left(\sum_{a \in C_{g[c(\bw)]}} \bx_g[(\bw, a)] \cdot a\right) \cdot \bw\right)
\end{align*}
In other words, it is the block resulting from $\bA_{g}\bx_{g}$ in the coordinates $k+1,k+2,\ldots,m'$. Since $\bA\bx = \by$, we must have
\begin{align} \label{eq:sum-coeff}
\sum_{g \in [\ell]} \sum_{\bw \in c^{-1}(i), a \in C_{g[i]}} \bx_g[(\bw, a)] = 1 & & \forall i \in [d]
\end{align}
and for every $g \in [\ell]$,
\begin{align} \label{eq:sum-stronginapprox}
\sum_{\bw \in \cW} \left(\sum_{a \in C_{g[c(\bw)]}} \bx_g[(\bw, a)] \cdot a\right) \cdot \bw  = \bzero_{m}.
\end{align}
From~\eqref{eq:sum-coeff} with $i = 1$, there must be $g^* \in [\ell]$ such that 
\begin{align*}
	\sum_{\bw \in c^{-1}(i), a \in C_{g^*[1]}} \bx_{g^*}[(\bw, a)] \neq 0  
\end{align*}
which in turn implies that there exists $\bw^*$ of color 1 such that $\sum_{a \in C_{g^*[1]}} \bx_{g^*}[(\bw^*, a)] \ne 0$.  From this and observation (ii) in Definition~\ref{def:ci}, we have $\sum_{a \in C_{g^*[1]}} \bx_{g^*}[(\bw^*, a)] \cdot a \ne 0$.  
This means that the left hand side of~\eqref{eq:sum-stronginapprox} instantiated with $g^*$ is a non-zero linear combination of at most $\|\bx_{g^*}\|_0$ vectors from $\cW$. Since $(\cW, k, c)$ is a NO instance of $\gapLDS_\gamma$, we can conclude that $\|\bx_{g^*}\|_0$ (and consequently $\|\bx\|_0$) must be larger than $\gamma \cdot k$. Hence, $(\bA, \by, k)$ is a NO instance for $\gapvec_{\gamma, p}$.
\end{proof}

Finally, we note that the above theorem together with Theorems~\ref{thm:bsmd} and~\ref{thm:lds-red} imply the main result of this section (Theorem~\ref{thm:MLDmain}). In particular, by selecting $d = \lceil\log_p(|L| + |R|)\rceil$ in Theorem~\ref{thm:lds-red} and applying Theorem~\ref{thm:mld-strong} afterwards, we get a (Turing) reduction from $\gapbsmd_\gamma$ to $\gapvec_{\gamma, p}$ that runs in time $$O((pd)^k \cdot \poly(|L| + |R|)) \leqs O\left(\left((pd)^{\sqrt{pd}} + (k^2)^{k})\right) \cdot \poly(|L| + |R|)\right) = k^{O(k)} \cdot \poly(|L| + |R|),$$
which is FPT. From this and from \W[1]-hardness of $\gapbsmd_\gamma$ (Theorem~\ref{thm:bsmd}), we arrive at Theorem~\ref{thm:MLDmain}.

%%% Local Variables:
%%% mode: latex
%%% TeX-master: "JACM submission"
%%% End:

\section{Parameterized Intractability of Minimum Distance Problem}			\label{sec:main-reduction}

Next, we will prove our main theorem regarding parameterized intractability of \mdp:

\begin{theorem}\label{thm:MDPmain}
	$\mdp_\gamma$ for any $\gamma \geqs 1$ is \W[1]-hard under randomized reductions.
\end{theorem}

This again proceeds in two steps. First, we give a simple reduction from \gapvec\ to \sncp\ in Section~\ref{sec:mld-to-sncp}. Then, we reduce the latter to \mdp\ in Section~\ref{sec:sncp-to-mdp}.

\subsection{Parameterized Inapproximability of Sparse Nearest Codeword Problem} \label{sec:mld-to-sncp}

We start with a simple approximation-preserving reduction from \gapvec\ to \sncp.

\begin{theorem}\label{thm:SNCPmain}
	$\sncp_\gamma$ for any $\gamma \geqs 1$ is \W[1]-hard under randomized reductions.
\end{theorem}

\begin{proof}
We reduce from $\gapvec_\gamma$, which is \W[1]-hard from Theorem~\ref{thm:MLDmain}. Let $(\bB, \bz, k)$ be the input for $\gapvec_\gamma$ where $\bB \in \F^{n \times m}$, $\by\in\F^n$, and $t$ is the parameter. Let $a = \lceil \gamma k + 1 \rceil$. We produce an instance $(\bA, \by, k)$ for $\sncp_{\gamma}$ by letting
\begin{align*}
\bA =
\begin{bmatrix}
\bB \\
\vdots \\
\bB
\end{bmatrix}
\begin{tikzpicture}[baseline=4ex]
\draw [decorate,decoration={brace,amplitude=10pt}]
(-1.5,1.4) -- (-1.5,0) node [black,midway,xshift=30pt] { $a$ copies};
\end{tikzpicture}
, \by =
\begin{bmatrix}
\bz \\
\vdots \\
\bz
\end{bmatrix}
\begin{tikzpicture}[baseline=4ex]
\draw [decorate,decoration={brace,amplitude=10pt}]
(-1.5,1.4) -- (-1.5,0) node [black,midway,xshift=30pt] { $a$ copies};
\end{tikzpicture}
\end{align*}

The reduction clearly runs in polynomial time, we are only left to argue that it appropriately maps YES and NO cases from $\gapvec_\gamma$ to those in $\sncp_\gamma$.

{\bf (YES Case)} Suppose that $(\bB, \bz, k)$ is a YES instance of $\gapvec_{\gamma}$, i.e., there exists $\bx \in \cB_q(\bzero, k)$ such that $\bB\bx = \bz$. This implies that $\|\bA\bx - \by\|_0 + \|\bx\|_0 = \|\bx\|_0 \leq k$ as desired.

{\bf (NO Case)} Suppose that $(\bB, \bz, k)$ is a NO instance of $\gapvec_{\gamma}$, i.e., for all $\bx \in \cB_m(\bzero, \gamma k)$, we have $\bB\bx \ne \bz$. Now, let us consider two cases, based on whether $\bx \in \cB_m(\bzero, \gamma k)$. First, if $\bx \in \cB(\bzero, \gamma k)$, then we have $\|\bA\bx - \by\|_0 + \|\bx\|_0 \geq a\|\bB\bx - \bz\|_0 \geq a > \gamma k$. On the other hand, if $\bx \notin \cB_m(\bzero, \gamma k)$, then $\|\bA\bx - \by\|_0 + \|\bx\|_0 \geq \|\bx\|_0 > \gamma k$.

Thus, in both cases, $\|\bA\bx - \by\|_0 + \|\bx\|_0 > \gamma k$ and $(\bA, \by, k)$ is a NO instance of $\sncp_{\gamma}$.
\end{proof}

\subsection{Reducing \sncp\ to \mdp} \label{sec:sncp-to-mdp}

In order to reduce \sncp\ to \mdp, we need to formalize the definition of Locally Suffix Dense Codes (\LSDC) and prove their existence; these are done in in Section~\ref{sec:dense-code}. Finally, we show how to use them in the reduction in Section~\ref{sec:main-red}.

\subsubsection{Locally Suffix Dense Codes} \label{sec:dense-code}

Before we formalize the notion of \emph{Locally Suffix Dense Codes} (\LSDC), let us give an intuitive explanation of \LSDC: informally, \LSDC\ is a linear code $\cC \subseteq \F^h$ where, given any short prefix $\bx \in \F^q$ where $q \ll h$ and a random suffix $\bs \in \F^{h - q}$, we can, with non-negligible probability, find a codeword that shares the prefix $\bx$ and has a suffix that is ``close'' in Hamming distance to $\bs$ (i.e. one should think of $r$ below as roughly $d/2$).  More formally, \LSDC\ can be defined as follows.

\begin{definition}\label{def:LEC}
A Locally Suffix Dense Code (\LSDC) over $\F$ with parameters\footnote{We remark that the parameter $h$ is implicit in specifying \LSDC.} $(m,q,d,r,\delta)$ an $m$-dimensional systematic linear code with minimum distance (at least) $d$ given by its generator matrix $\mathbf L\in \F^{h\times m}$ such that for any $\mathbf{x} \in \F^q$, the following holds:
	\begin{equation} \label{eq:dense_prob}
	\Pr_{\mathbf{s} \sim \F^{h - q}} \Bigg[\exists \bz \in \cB_{h - q}(\bs, r): (\bx \circ \bz) \in \bL(\F^m) \Bigg] \geqs \delta.
	\end{equation}
\end{definition}

We note that our notion of Locally Suffix Dense Codes is closely related and inspired by the notion of Locally Dense Codes (\LDC) of Dumer~\etal~\cite{DMS03}. Essentially speaking, the key differences in the two definitions are that (i) Locally Dense Codes are for the case of $q = 0$, i.e., there is no prefix involved, and (ii) $\bs$ in \LDC\ is not chosen at random from $\F^q$ but rather from $\cB_q(\bzero, r)$. Note that, apart from these, there are other subtle additional requirements in Locally Dense Codes that we do not need in our reduction, such as the requirements that the ``center'' $\bs$ is close to not just one but many codewords; however, these are not important and we will not discuss them further.

Unfortunately, the proof of Dumer~\etal\ does not directly give us the desired \LSDC; the main issue is that, when there is no prefix, the set of codewords is a linear subspace, and their proof relies heavily on the linear structure of the set (which is also why $\bs$ is randomly chosen from $\cB_q(\bzero, r)$ instead of $\F^q$). However, the set of our interest is $\Big\{ {\bf z} \in \mathbb{F}^{h - q} \Big| {\bf x} \circ {\bf z} \in \mathbf{L}(\F^m) \Big\}$, which is not a linear subspace but rather an affine subspace; Dumer \etal's argument (specifically Lemma 13 in~\cite{DMS03}) does not apply in the affine subspace case.

Below, we provide a different proof than Dumer~\etal\ for the construction of \LSDC. Our bound is more related to the \emph{Sphere Packing} (aka Hamming) bound for codes. In particular, we show below that BCH codes, which ``near'' the Sphere Packing bound gives us \LSDC\ with certain parameters. It should be noted however that the probability guarantee $\delta$ that we have is quite poor, i.e. $\delta \geqs d^{-\Theta(d)}$, but this works for us since $d$ is bounded by a function of the parameter of our problem. On the other hand, this would not work in NP-hardness reductions of~\cite{DMS03} (and, on top of this, our codes may not satisfy other additional properties required in \LDC).

%We would like to note that \SCC are closely related to Locally Dense Codes (see \cite{DMS03} or \cite{Mic14} for the definition). A key difference in our definition is that the code is tailored towards covering sparse vectors which allows us greater flexibility in the parameters. This is in direct contrast to previous constructions which were designed to cover entire subspaces, with parameters which end up having strong dependencies on the ambient dimension. Consequently, they are not directly applicable to the parameterized setting. 

%Next, we show the existence of \SCC  for a certain range of parameters. Informally, the existence of \SCC follows from the existence of codes near the sphere-packing bound, such as BCH codes.% for every message length and distance of the code. 

\begin{lemma} \label{lem:ldc}
	For any $q, d \in \mathbbm{N}$ such that $d$ is an odd number larger than one, there exist $h,m \in \mathbb{N}$ and $\bL \in \F^{h \times m}$ which is a \LSDC with parameters $\left(m,q,d,\frac{d - 1}{2},\frac{1}{d^{d/2}}\right)$. Additionally, the following holds:
	\begin{itemize}
		\item $h, m \leqs \poly(q, d)$ and $m \geqs q$,
		\item $\bL$ can be computed in $\poly(q, d)$ time.
	\end{itemize}
\end{lemma}

\begin{proof}
	Let $h$ be the smallest integer such that $h + 1$ is a power of two and that $h \geqs \max\{2q, 10d\log d\}$, and let $m = h - \left(\frac{d - 1}{2}\right)\log(h + 1)$. Clearly, $h$ and $m$ satisfy the first condition.

	Let $\mathbf{L}$ be the generator matrix of the $[h,m,d]_2$ linear code as given by Theorem~\ref{thm:bch}. Without loss of generality, we assume that the code is systematic on the first $m$ coordinates. From Theorem~\ref{thm:bch}, $\mathbf L$ can be computed in $\poly(h) = \poly(q, d)$ time.
	
	It remains to show that for our choice of $\mathbf{L}$,~\eqref{eq:dense_prob} holds for any fixed choice of $\bx \in \F^q$. Fix a vector $\bx \in \F^q$ and define the set $\mathcal{C}= \Big\{ {\bf z} \in \F^{h - q} \Big| {\bf x} \circ {\bf z} \in \mathbf{L}(\F^m) \Big\}$. Since the code generated by $\mathbf{L}$ is systematic on the first $m \geqs q$ coordinates, we have that $|\mathcal{C}| \geqs 2^{m - q}$.

Moreover, since the code generated by $\bL$ has distance $d$, every distinct $\mathbf{z}_1,\mathbf{z}_2 \in \mathcal{C}$ are at least $d$-far from each other (i.e. $\|\bz_1 - \bz_2\|_0 \geqs d$). Therefore, for any distinct pair of vectors $\mathbf{z}_1,\mathbf{z}_2 \in \mathcal{C}$, the sets $\cB_{h - q}(\mathbf{z}_1,\frac{d - 1}{2})$ and  $\cB_{h - q}(\mathbf{z}_2,\frac{d - 1}{2})$ are disjoint. Hence the number of vectors in the union of $\left(\frac{d-1}{2}\right)$-radius Hamming balls around every $\mathbf{z} \in \mathcal{C}$ is at least
	\begin{align*}
	2^{m - q}\left\lvert\mathcal{B}_{h - q}\left(\mathbf{0}, \frac{d-1}{2} \right)\right\rvert
	\geqs 2^{m - q}{h - q \choose \frac{d-1}{2} }
	\geqs 2^{m - q}{h/2 \choose \frac{d-1}{2} }
	\geqs 2^{m - q}\Big(\frac{h}{d - 1}\Big)^{\frac{d-1}{2}} 
	\end{align*}	
	On the other hand, $|\F^{h - q}| = 2^{h-q} = 2^{m - q}(h + 1)^{\frac{d-1}{2}}$. Hence, with probability at least $\left(\frac{h}{(d - 1)(h + 1)}\right)^{\frac{d-1}{2}} \geqs \frac{1}{d^{d/2}}$, a vector $\bs$ sampled uniformly from $\F^{h - q}$ lies in $\cB_{h - q}\left(\bz, \frac{d - 1}{2}\right)$ for some vector $\bz \in\mathcal{C}$. This is indeed the desired condition in~\eqref{eq:dense_prob}, which completes our proof.
\end{proof}

\subsubsection{The Reduction} \label{sec:main-red}

In this subsection, we state and prove the FPT reduction from the $\sncp$ problem to the $\mdp$ problem. It is inspired by the reduction from \cite{DMS03}, which is then modified (and simplified) to work in combination with \LSDC\ instead of \LDC.

\begin{lemma}							\label{lem:sncp-to-mdp}
There is a randomized FPT reduction from $\sncp_{2.5}$ to $\mdp_{1.01}$.
\end{lemma}

\begin{proof}
Let $(\bB, \by, t)$ be the input for $\sncp_{\gamma'}$ where $\bB \in \F^{n \times q}$, $\by\in\F^n$, and $t$ is the parameter. We may assume without loss of generality that $t \ge 1000$. Let $d$ be the smallest odd integer greater than $2.5 t$. Let $h, m \in \N, \bL \in \F^{h \times m}$ be as in Lemma~\ref{lem:ldc}.

We produce an instance $(\bA, k)$ for $\mdp_{\gamma}$ by first sampling a random $\bs \sim \F^{h - q}$. Then, we set $k = t + (d - 1)/2$, $\bs' = \bzero_q \circ -\bs$ and
\begin{align*}
\bA =
\begin{bmatrix}
\bB & \bzero_{n \times (m - q)} & \by \\
\bL & & \bs'
\end{bmatrix}
\in \F^{(n + h) \times (m + 1)}.
\end{align*}
Notice that the zeros are padded onto the right of $\bB$ so that the number of rows is the same as that of $\bL$.

Since $k = t + (d - 1)/2 = O_{\gamma'}(t)$ and the reduction clearly runs in polynomial time, we are only left to argue that it appropriately maps YES and NO cases from $\sncp_{\gamma'}$ to those in $\mdp_\gamma$.

{\bf (YES Case)} Suppose that $(\bB, \by, t)$ is a YES instance of $\sncp_{\gamma'}$, i.e., there exists $\bx \in \F^q$ such that $\|\bB\bx - \by\|_0 + \|\bx\|_0 \leqs t$. From Lemma~\ref{lem:ldc}, with probability at least $1/d^{d/2}$, there exists $\bu \in \cB_{h - q}\left(\bs, \frac{d - 1}{2}\right)$ such that $\bx \circ \bu \in \bL(\F^m)$. From this and from systematicity of $\bL$, there exists $\bz' \in \F^{m - q}$ such that $\bL(\bx \circ \bz') = \bx \circ \bu$. Conditioned on this, we can pick $\bz = \bx \circ \bz' \circ 1 \in \F^{m + 1}$, which yields
\begin{align*}
\|\bA \bz\|_0 = \|\bB\bx - \by\|_0 + \|\bx\|_0 + \|\bu - \bs\|_0 \leqs t + \frac{d - 1}{2} = k.
\end{align*}

In other words, with probability at least $1/d^{d/2}$, $(\bA, k)$ is a YES instance of $\mdp_{\gamma}$ as desired.

{\bf (NO Case)} Suppose that $(\bB, \by, t)$ is a NO instance of $\sncp_{\gamma'}$. We will show that, for all non-zero $\bz \in \F^{m + 1}$, $\|\bA\bz\|_0 > 2.5 t$; with our choice of parameters and our assumption on $t$, it is simple to check that $2.5 t > 1.01 k$. Hence, this implies that $(\bA, k)$ is a NO instance of $\mdp_{\gamma}$.

To show that $\|\bA\bz\|_0 > \gamma' t$ for all $\bz \in \F^{m + 1} \setminus \{\bzero\}$, let us consider two cases, based on the last coordinate $\bz[m + 1]$ of $\bz$. For convenience, we write $\bz$ as $\bx \circ \bz' \circ \bz[m + 1]$, where $\bx \in \F^q$ and $\bz' \in \F^{m - q}$.

If $\bz[m + 1] = 0$, then $\|\bA\bz\|_0 = \|\bB\bx\|_0 + \|\bL(\bx \circ \bz')\|_0 \geqs  \|\bL(\bx \circ \bz')\|_0 \geqs d$, where the last inequality comes from the fact that $\bL$ is a generator matrix of a code of distance $d$ (and that $\bz\neq\mathbf 0$). Finally, recall that we select $d > 2.5 t$, which yields the desired result for this case.

On the other hand, if $\bz_{m + 1} = 1$, then $\|\bA\bz\|_0 \geqs \|\bB\bx - \by\|_0 + \|\bx\|_0 \geq 2.5 t$, where the second inequality comes from the assumption that $(\bB, \by, t)$ is a NO instance of $\sncp_{2.5}$.

In conclusion, $\|\bA\bz\|_0 > 2.5 t$ in all cases considered, which completes our proof.
\end{proof}

\smallskip

\noindent\textbf{Gap Amplification.} Finally, the above gap hardness result can be boosted to any constant gap using the now standard technique of tensoring the code (c.f. \cite{DMS03},\cite{AK14}) which is stated formally in the following proposition:

\begin{proposition}[E.g. \cite{DMS03}] \label{prop:gap-amplification}
Given two linear codes $C_1 \subseteq \F^m$ and ${C}_2 \subseteq \F^n$, let ${C}_1 \otimes {C}_2 \subseteq \mathbb{F}^{m \times n}$ be the tensor product of ${C}_1$ and ${C}_2$. Then $d({C}_1 \otimes {C}_2) = d({C}_1)d({C}_2)$.
\end{proposition}

We briefly show how the above proposition can be used to amplify the gap. Consider a $\mdp_\gamma$ instance $(\bA,k)$ where $\bA \in \F^{m \times n}$. Let $C \subseteq \F^m$ be the linear code generated by it. Let $C^{\otimes 2} = C \otimes C$ be the tensor product of the code with itself, and let $\bA^{\otimes 2}$ be its generator matrix. By the above proposition, if $(\bA,k)$ is a YES instance, then $d(C^{\otimes 2}) \leqs k^2$. Conversely, if $(\bA,k)$ is a NO instance, then $d(C^{\otimes 2}) \geqs \gamma^2k^2$. Therefore $(\bA^{\otimes 2},k^2)$ is a $\mdp_{\gamma^2}$ instance. Hence, for any $\alpha \in \mathbb{R}_+$, repeating this argument $\lceil\log_\gamma \alpha \rceil$-number of times gives us an FPT reduction from $k$-$\mdp_\gamma$ to \text{$k^{2\lceil\log_\gamma \alpha\rceil}$-$\mdp_\alpha$}. We have thereby completed our proof of Theorem~\ref{thm:MDPmain}.

%%% Local Variables:
%%% mode: latex
%%% TeX-master: "JACM submission"
%%% End:

\newcommand{\bU}{\mathbf{U}}
\newcommand{\bW}{\mathbf{W}}
\newcommand{\cL}{\mathcal{L}}

\section{Parameterized Intractability of Shortest Vector Problem}\label{sec:svp}

The main result of this section is the parameterized inapproximability of $\svp$, as stated below.

\begin{theorem}[FPT Inapproximability of $\svp$]\label{thm:SVPmain}
For any $p > 1$, there exists constant $\gamma_p > 1$ (where $\gamma_p$ depends on $p$), such that there $\svp_{p,\gamma_p}$ is \W[1]-hard (under randomized reductions).
\end{theorem}

Similar to the Minimum Distance Problem, the proof of Theorem~\ref{thm:SVPmain} goes through two steps. First, we show that the non-homogeneous variant, the Nearest Vector Problem. Then, in the second step, we reduce it to the Shortest Vector Problem.

\subsection{FPT Inapproximability of Nearest Vector Problem}			\label{sec:csp-to-snvp}

In this section, we prove the inapproximability of Nearest Vector Problem, as stated more formally below. The proof is via a simple reduction from Maximum Likelihood Decoding over a large field.

\begin{theorem}[FPT Inapproximability of $\snvp$]\label{thm:CVPmain}
For any $\eta, p \geqs 1$, $\snvp_{\eta, p}$ is \W[1]-hard.
\end{theorem}

\begin{proof}
Let $q$ be the smallest prime number such that $q > 2\eta$. We will reduce from $\gapvec_{2\eta, q}$, which is $\W[1]$-hard from Theorem~\ref{thm:MLDmain}. Let $(\bA \in \mathbb{F}_q^{n \times m}, \by \in \mathbb{F}_q^n, k)$ be an instance of $\gapvec_{2\eta, q}$. We create an instance of $(\bA', \by', k')$ of $\snvp_{\eta, p}$ as follows. First, we set $k' = 2k$ and let
\begin{align*}
\bA' = 
\begin{bmatrix}
\vone_a \otimes \bA & \vone_a \otimes (q \cdot \Id_n) \\
\Id_n & \bzero_{n \times n} \\
\bzero_{k \times k} & \bzero_{k \times k}
\end{bmatrix}
\in 
\mathbb{Z}^{n' \times m'}
\text{, and, }
\by' =
\begin{bmatrix}
\vone_a \otimes \by \\
\bzero_n \\
\vone_k
\end{bmatrix}
\in
\mathbb{Z}^{n'},
\end{align*}
where $a = \lceil 2\eta k + 2 \rceil$, $n' = a n + k$ and $k' = m + n$.
Clearly, the reduction runs in polynomial time. We next argue its correctness.

\textbf{(YES Case)} Suppose that $(\bA, \by, k)$ is a YES instance of $\gapvec_{2\eta, q}$, i.e., that there exists $\bx \in \{0, 1\}^n$ with $\|\bx\|_0 \leqs k$ such that $\bA \bx = \by$ when operations are over $\mathbb{F}_q$. This means that, when view operations over $\Z$, we have $\bA\bx = \by + q \cdot \bz$ for some $\bz \in \mathbb{Z}^n$. Let $\bx' = \bx \circ (-\bz) \in \mathbb{Z}^{m'}$. Then, we have (over $\Z$)
\begin{align*}
\|\bA' \bx'\|_p^p = \|\bzero_a \circ \bx \circ \vone_k\|_p^p \leqs 2k = k'.
\end{align*}
In other words, $(\bA', \by', k')$ is a YES instance of $\snvp_{\eta, p}$ as desired.

\textbf{(NO Case)} Suppose that $(\bA, \by, k)$ is a NO instance of $\gapvec_{2\eta, q}$. Consider any $\bx' \in \mathbb{Z}^{m'}$ and any $w \in \mathbb{Z} \setminus \{0\}$. We would like to show that $\|\bA'\bx' - w \cdot \by'\|_p^p > \eta \cdot k' = 2\eta k$. To do so, let us write $\bx'$ as $\bx \circ \bz$ where $\bx \in \Z^m$ and $\bz \in \Z^n$. We can now rearrange $\|\bA'\bx' - w \cdot \by'\|_p^p$ as
\begin{align*}
\|\bA'\bx' - w \cdot \by'\|_p^p &= a \|\bA\bx + q \bz - \by\|_p^p + \|\bx\|_p^p + |w|^p k.
\end{align*}
As a result, if $\bA\bx + q\bz \ne \by$, then $\|\bA'\bx' - w \cdot \by'\|_p^p \geqs a > 2\eta k$. Furthermore, if $|w| \geqs q$, then we also have $\|\bA'\bx' - w \cdot \by'\|_p^p \geqs |w|^p k \geqs q k > 2\eta k$. Hence, we may henceforth assume that $|w| < q$ and $\bA'\bx' + q\bz = w \cdot \by'$. Since $|w| < q$, it has an inverse modulo $q$, i.e., there exists $u \in [q - 1]$ such that $uw \equiv 1 \Mod{q}$. Now, let us consider $\tilde{\bx} \in \mathbb{F}_q^m$ where $\tilde{\bx}[i]$ is defined as the remainder of $u \cdot \bx[i]$ modulo $q$. From $\bA\bx + q\bz \ne \by$, we have (over $\mathbb{F}_q$)
\begin{align*}
\bA\tilde{\bx} = (uw) \cdot \by = \by.
\end{align*}
Since $(\bA, \by, k)$ is a NO instance of $\gapvec_{2\eta, q}$, we must have $\|\tilde{\bx}\|_0 > 2\eta k$. Observe that $\|\bx\|_0 \geqs \|\tilde{\bx}\|_0$. Thus, we have $\|\bA'\bx' - w \cdot \by'\|_p^p \geqs \|\bx\|_p^p \geqs \|\bx\|_0 > 2\eta k$. In other words, we can conclude that $(\bA', \by', k')$ is a NO instance of $\snvp_{\eta, p}$.
\end{proof}

%%% Local Variables:
%%% mode: latex
%%% TeX-master: "JACM submission"
%%% End:

\subsection{Following Khot's Reduction from NVP to SVP}

We will now reduce from $\snvp$ to $\svp$. This step is  almost the same as that of Khot~\cite{Khot05}, with small changes in parameter selection. Despite this, we repeat the whole argument here (with appropriate adjustments) for completeness.

The main properties of the (randomized) FPT reduction from $\snvp_{p,\eta}$ to $\svp_{p,\gamma}$ are summarized below. For succinctness, we define a couple of additional notation: let $\cL(\bA)$ denote the lattice generated by the matrix $\bA \in \Z^{n \times m}$, i.e., $\cL(\bA) = \{\bA\bx \mid \bx \in \Z^m\}$, and let $\lambda_p(\cL)$ denote the length (in the $\ell_p$ norm) of the shortest vector of the lattice $\cL$, i.e., $\lambda_p(\cL) = \underset{\bzero \ne \by \in \cL}{\min} \|\by\|_p$.

\begin{lemma}				\label{lem:snvp-to-svp}
	Fix $p > 1$, and let $\eta \geqs 1$ be such that $\frac12 + \frac{1}{2^p} + \frac{(2^p + 1)}{\eta} < 1$. Let $(\bB,\by,t)$ be a $\snvp_{p,\eta}$ instance, as given by Theorem~\ref{thm:CVPmain}. Then, there is a randomized FPT reduction from $\snvp_{p,\eta}$ instance $(\bB,\by,t)$ to $\svp_{p,\gamma}$ instance $(\bB_{\rm svp},\gamma^{-1}_pl)$ with $l = \eta \cdot t$ such that
	\begin{itemize}
		%\item $\bB_{\rm svp} \in \Z^{(n + h + g + 1) \times (q + h + g + 2)}$ where $h, g = {\rm poly}(l,n,q)$.
		\item (YES) If $(\bB,\by,t)$ is a YES instance, then with probability at least $0.8$, $\lambda_p(\cL(\bB_{\rm svp}))^p \leqs \gamma^{-1}_p l$.
		\item (NO) If $(\bB,\by,t)$ is a NO instance, then with probability at least $0.9$, $\lambda_p(\cL(\bB_{\rm svp}))^p > l$.
	\end{itemize}
	Here $\gamma_p := \frac{1}{\frac12 + (2^p+1)/\eta + 1/2^p}$ is strictly greater than $1$ by our choice of $\eta$. 
\end{lemma}

Combining the above lemma with Theorem~\ref{thm:CVPmain} gives us Theorem \ref{thm:SVPmain}. 

We devote the rest of this subsection to describing the reduction (which is similar to that from~\cite{Khot05}) and proving Lemma \ref{lem:snvp-to-svp}. In Section \ref{sec:bch-lattic}, we define the BCH lattice, which is the key gadget used in the reduction. Using the BCH lattice and the $\snvp_{p, \eta}$ instance, we construct the intermediate lattice $\bB_{\rm int}$ in Section \ref{sec:int-lattice}. The intermediate lattice serves to blow up the number of ``good vectors'' for the YES case, while controlling the number of ``bad vectors'' for the NO case. In particular, this step ensures that the number of good vectors in the YES case (Lemma \ref{lem:good-vectors}) far outnumber the number of bad vectors in the NO case (Lemma \ref{lem:annoying-vectors}). Finally, in Section \ref{sec:final-lattice} we compose the intermediate lattice with a random homogeneous constraint (sampled from an appropriate distribution), to give the final $\svp_{p,\gamma}$ instance. The additional random constraint is used to annihilate all bad vectors in the NO case, while retaining at least one good vector in the YES case.

For the rest of the section, we fix $(\bB,\by,t)$ to be a $\snvp_{p,\eta}$ instance (as given by Theorem \ref{thm:CVPmain}), and set $l := \eta \cdot t$ and $r := \left(\frac{1}{2} + \frac{1}{2^p} + \frac{1}{\eta}\right) l$. For simplicity of calculations, we will assume that both $l$ and $r$ are integers, and that $l$ is even. Furthermore, we say that a vector $\bu$ is \emph{good} (for the YES case) if $\|\bu\|_p^p \leqs \gamma_p^{-1} l$, and we say that $\bu$ is \emph{bad} (for the NO case) if $\|\bu\|_p^p \leqs l$.

\subsubsection{The BCH Lattice gadget}				\label{sec:bch-lattic}

We begin by defining the BCH lattices which is the key gadget used in the reduction. Given parameters $l,h \in \mathbb{N}$ where $h + 1$ is a power of $2$ and $l < h$. Let $g = (l/2) \cdot \log(h + 1)$. Theorem~\ref{thm:bch} guarantees that there exists a BCH code with block length $h$, message length $h - g$ and distance $l + 1$. Let $\mathbf{P}_{\rm BCH} \in \{0,1\}^{g \times h}$ be the parity check matrix of such code. The BCH lattice is defined by
\begin{align*}
	\bB_{\rm BCH} =
	\begin{bmatrix}
		{\rm Id}_{h} &  \mathbf{0}_{h \times g} \\
		l\cdot\mathbf{P}_{\rm BCH} & 2l\cdot{\rm Id}_{g}
	\end{bmatrix}
	\in \mathbb{Z}^{(h + g) \times (h + g)}.
\end{align*}
\noindent  The following lemma, which is simply a restatement\footnote{In fact, Lemma~\ref{lem:bch-lattice} is even weaker than Khot's lemma, since we do not impose a bound on $\|\mathbf{z}\|_p$.} of Lemma 4.3 in \cite{Khot05}, summarizes the key properties of BCH lattices, as defined above.

\begin{lemma}[\cite{Khot05}]				\label{lem:bch-lattice}
	Let $\mathbf{B}_{\rm BCH} \in \mathbb{Z}^{(h + g) \times ( h + g)}$ be as above. There exists a randomized polynomial time algorithm that, with probability at least $0.99$, returns a vector $\mathbf{s} \in \mathbb{Z}^{h+g}$ such that the following holds: there are at least $\frac{1}{100}2^{-g}{h \choose r}$ distinct vectors $\mathbf{z} \in \mathbb{Z}^{h + g}$ such that $\|\mathbf{B}_{\rm BCH} \mathbf{z} - \mathbf{s}\big\|_p^p = r$.
\end{lemma}

\subsubsection{The Intermediate Lattice}			\label{sec:int-lattice}

We now define the intermediate lattice. Let $(\bB,\by,t)$ be an instance of $\snvp_{p,\eta}$, where $\bB \in \Z^{n \times q}$. The intermediate lattice $\bB_{\rm int}$ is constructed as follows. Let $l = \eta t$. Let $h$ be the smallest power of 2 such that $h \geqs \max\{2n, (10^{10}l)^{2\eta}\}$, and let $\bB_{\rm BCH}$ be constructed as above. Then 
\begin{align*}
	\bB_{\rm int} =
	\begin{bmatrix}
		2\bB &  \mathbf{0}_{ n \times  (h+g)} &  2\by  \\
		\mathbf{0}_{(h + g) \times q} &  \bB_{\rm BCH} &  \bs
	\end{bmatrix}
	\in \Z^{(n + h + g) \times (q + h + g + 1)}.
\end{align*}
where $\mathbf{s} \in \Z^{h + g}$ is the vector given by Lemma \ref{lem:bch-lattice}. %The following lemmas bound the ``good'' and ``bad'' vectors in the YES and NO cases respectively.

\paragraph{Bounding Good Vectors in YES Case.}
We now prove a lower bound on the number of good vectors in the YES case.

\begin{lemma}				\label{lem:good-vectors}
   Let $(\mathbf{B},\mathbf{y},t)$ be a YES instance, and let $\mathbf{B}_{\rm int}$ be the corresponding intermediate lattice. With probability at least $0.99$, there are at least ${h^{r}}\Big({200 h^{l/2}l^{l}}\Big)^{-1}$ good non-zero vectors in $\cL(\mathbf{B}_{\rm int})$.
\end{lemma}

\begin{proof}
	Since $(\mathbf{B},\mathbf{y},t)$ is a YES instance, there exists $\tbx \in \mathbb{Z}^q$ such that $\|{\bf B}\tbx - {\bf y}\|_p^p \leqs t$. From Lemma \ref{lem:bch-lattice}, with probability at least 0.99, there exist at least $2^{-g}{h \choose r}/100$ distinct vectors ${\bf z} \in \mathbb{Z}^{h + g}$ such that $\|{\bf B}_{\rm BCH}{\bf z} - {\bf s}\|^p = r$. For each such $\bz$, consider the vector $\bx = \tbx \circ \bz \circ -1$. It follows that ${\bf B}_{\rm int}\bx = (2{\bf B}\tbx - 2{\bf y}) \circ({\bf B}_{\rm BCH}{\bf z} - {\bf s})$ is a non-zero vector and $\|{\bf B}_{\rm int}\bx\|_p^p = 2^p\|{\bf B}\tbx - {\bf y}\|_p^p + \|{\bf B}_{\rm BCH}{\bf z} - {\bf s}\|_p^p\leqs 2^pt + r = \gamma_p^{-1} l$. Since the number of such vectors $\bx$ is at least the number of distinct coefficient vectors $\mathbf{z}$, it can be lower bounded by
	\begin{equation*}
	\frac{1}{100} \cdot 2^{-g}{h \choose r} \geqs \frac{1}{100} \cdot 2^{-\frac{l}{2}\log(h+1)} {h \choose r} \geqs \frac{1}{100} \cdot \frac{h^r}{r^r (h+1)^{l/2}} \geqs \frac{1}{200} \cdot \frac{h^r}{l^l h^{l/2}},
	\end{equation*}
	\noindent where the last inequality follows from $r \leqs l$ and $l < h$. Finally, observe that each $\bz$ produces different $\bB_{\rm BCH}\bz$ and hence all $\bB_{\rm int}\bx$'s are distinct.
\end{proof}	
	
\paragraph{Bounding Bad Vectors in NO Case.}
We next bound the number of bad vectors in the NO case:

\begin{lemma}				\label{lem:annoying-vectors}
	Let $(\mathbf{B},\mathbf{y},t)$ be a NO instance, and let $\mathbf{B}_{\rm int}$ be the corresponding intermediate lattice. Then the number of bad vectors in $\cL(\bB_{\rm int})$ is at most $10^{-5}{h^{r}}\Big({200 h^{l/2}{l}^{l}}\Big)^{-1}$.
\end{lemma}

At the heart of the proof is the claim that every bad vector must have even coordinates:

\begin{claim} \label{claim:annoying-vectors}
Let $(\mathbf{B},\mathbf{y},t)$ be a NO instance, and let $\mathbf{B}_{\rm int}$ be the corresponding intermediate lattice. Then, for every bad $\bu \in \cL(\bB_{\rm int})$, all coordinates of $\bu$ must be even.
\end{claim}

\begin{proof}
Let $\bu$ be any bad vector in $\cL(\bB_{\rm int})$ and let $\bx \in \mathbb{Z}^{q + h + g + 1}$ be such that $\bB_{\rm int}\mathbf{x} = \bu$. We write $\bx$ as $\bx_1 \circ \bx_2 \circ x$ where $\bx_1 \in \Z^q$, $\bx_2 \in \Z^{m + h}$ and $x \in \Z$. Using this, we can express $\bu$ as $\bB_{\rm int}\mathbf{x} = (2\bB\bx_1 - 2x\cdot\by)\circ(\bB_{\rm BCH}\bx_2 - x\cdot \mathbf{s})$. Recall that $\bu$ is bad means that $\|\bu\|_p^p \leqs l$, which implies that $\|\bB\bx_1 - x\cdot\by\| \leqs l = \eta \cdot t$. Since $(\mathbf{B},\mathbf{y},t)$ is a NO instance, it must be that $x = 0$.

Note that we now have $\bu = (2\bB\bx_1) \circ (\bB_{\rm BCH}\bx_2)$. Let us assume for the sake of contradiction that $\bu$ has at least one odd coordinate; it must be that $(\bB_{\rm BCH}\bx_2)$ has at least one odd coordinate. Let us further write $\bx_2$ as $\bx_2 = \bw_1 \circ \bw_2$ where $\bw_1 \in \Z^m$ and $\bw_2 \in \Z^h$. Notice that $\bB_{\rm BCH} \bx_2 = \bw_1 \circ (l(\bP_{\rm BCH}\bw_1 - 2\bw_2))$. Since every coordinate of $\bB_{\rm BCH} \bx_2$ must be less than $l$ in magnitude, it must be the case that $\bP_{\rm BCH}\bw_1 - 2\bw_2 = \vzero$. In other words, $(\bw_1 \mod 2)$ is a codeword of the BCH code. However, since the code has distance $l+1$, this means that, if $\bw_1$ has at least one odd coordinate, it must have at least $l+1$ odd (non-zero) coordinates, which contradicts $\|\bu\|_p^p \leqs l$.
\end{proof}

Having proved Claim~\ref{claim:annoying-vectors}, we can now prove Lemma~\ref{lem:annoying-vectors} by a simple counting argument.

\begin{proof}[Proof of Lemma~\ref{lem:annoying-vectors}]
From Claim~\ref{claim:annoying-vectors}, all coordinates of $\bu$ must be even. Therefore, $\bu$ must have at most $l/2^p$ non-zero coordinates, all of which have magnitude at most $\lfloor l^{1/p} \rfloor \leqs l - 1$. Hence, we can upper bound the total number of such vectors by
\begin{eqnarray*}
\big(2(l - 1) + 1\big)^{l/2^p}{{n + h + g}\choose {\lfloor \frac{l}{2^p} \rfloor}} \leqs (2l)^l (n + h + g)^{l/2^p}
\leqs (2l)^{l} (2lh)^{l/2^p}
\leqs (2l)^{2l} h^{l/2^p}
\end{eqnarray*}
where the second-to-last step holds since $g \leqs \frac{l}{2}\log (h+1) \leqs lh/2$ and $n \leqs h/2$. On the other hand,
\begin{equation*}
\frac{h^{r}}{h^{l/2}{l}^{l}} = \frac{h^{\big(\frac12 + \frac1\eta + \frac{1}{2^p}\big)l}}{h^{l/2}{l}^{l}} = h^{l/2^p} (h/l^\eta)^{l/\eta} \geqs 10^8 \left((2l)^{2l} h^{l/2^p}\right),
\end{equation*}
\noindent which follows from $h \geqs (10^{10}l)^{2\eta}$. Combining the two bounds completes the proof.
\end{proof}

\subsubsection{The $\svp_{p,\gamma}$ Instance and Proof of The Main Lemma}					\label{sec:final-lattice}

Finally, we construct $\bB_{\rm svp}$ from $\bB_{\rm int}$ by adding a random homogeneous constraint similar to \cite{Khot05}. For ease of notation, let $N_g$ denote the lower bound on the number of distinct coefficient vectors guaranteed by Lemma \ref{lem:good-vectors} in the YES case. Similarly, let $N_a$ denote the upper bound on the number of annoying vectors as given in Lemma \ref{lem:annoying-vectors}. Combining the two Lemmas we have $N_g \geqs 10^5N_a$, which will be used crucially in the construction and analysis of the final lattice.

\bigskip

\noindent\textbf{Construction of the Final Lattice.} Let $\rho$ be any prime number in{\footnote{Note that the density of primes in this range is at least $1/\log N_g = 1/r\log h$. Therefore, a random sample of size $O(r\log h)$ in this range contains a prime with high probability. Since we can test primality for any $\rho \in \Big[10^{-4}N_g, 10^{-2}N_g\Big]$ in FPT time, this gives us an FPT algorithm to sample such a prime number efficiently .}} $\Big[10^{-4}N_g, 10^{-2}N_g\Big]$. Furthermore, let ${\bf r} \overset{\rm unif}{\sim} [0,\rho-1]^{n + h + g}$ be a uniformly sampled lattice point. We construct ${\bf B}_{\rm svp}$ as 
\begin{align*}
	\bB_{\rm svp} =
	\begin{bmatrix}
		\bB_{\rm int} &  {0}  \\
		l \cdot{\bf r}^T \bB_{\rm int} &  l \cdot \rho 
	\end{bmatrix}
	\in \Z^{(n + h + g + 1) \times (q + h + g + 2)}.
\end{align*}

\noindent This can be thought of as adding a random linear constraint to the intermediate lattice. The choice of parameters ensures that with good probability, in the YES case, at least one of the good vectors ${\bf x} \in\mathbb{Z}^{q + h + g + 1}$ evaluates to $0$ modulo $\rho$ on the random constraint, and therefore we can pick $u \in \Z$ such that ${\bf B}_{\rm svp}({{\bf x} \circ u}) = (\bB_{\rm int}\bx) \circ 0$ still has small $\ell_p$ norm. On the other hand, since $N_a \ll N _g$, with good probability, all of bad vectors evaluate to non-zeros, and hence will contribute a coordinate of magnitude $l$. This intuition is formalized below.

\begin{proof}[Proof of Lemma~\ref{lem:snvp-to-svp}]
	Let $\bB_{\rm svp}$ be the corresponding final lattice of $(\bB,\by,t)$ as described above. Observe that given the $\snvp_{p,\eta}$-instance $(\bB,\by,t)$, we can construct $\bB_{\rm svp}$ in ${\rm poly}(n,q,t)$-time.

	Moreover, observe that $\cL(\bB_{\rm svp})$ is exactly equal to $\{\bu \circ (l \cdot w) \mid \bu \in \cL(\bB_{\rm int}), w \equiv \br^T \bu\Mod{\rho}\}$.

	Suppose that $(\bB,\by,t)$ is a NO instance. Consider any $\bu \circ (l \cdot w) \in \cL(\bB_{\rm svp})$. If $\|\bu \circ (l \cdot w)\|_p^p \leqs l$, it must be that $\|\bu\|_p^p \leqs l$ and $w = 0$; the latter is equivalent to $\br^T \bu \equiv 0 \Mod{\rho}$. However, from Lemma~\ref{lem:annoying-vectors}, there are only $N_a$ bad vectors $\bu$ in $\cL(\bB_{\rm int})$. For each such non-zero $\bu$, the probability that $\br^T \bu \equiv 0 \Mod{\rho}$ is exactly $1/\rho$. As a result, by taking union bound over all such $\bu \ne \bzero$, we can conclude that, with probability at least $1 - N_a/\rho \geqs 0.9$, we have $\lambda_p(\cL(\bB_{\rm svp}))^p > l$.
	
	Next, suppose that $(\bB,\by,t)$ is a YES instance. We will show that, with probability at least 0.8, $\lambda_p(\cL(\bB_{\rm svp}))^p \leqs \gamma_p^{-1} l$. To do this, we first  condition on the event that there exists at least $N_g$ good vectors as guaranteed by Lemma \ref{lem:good-vectors}. Consider any two good vectors $\bu_1 \ne \bu_2$. Since each entry of $\bu_1$ and $\bu_2$ is of magnitude at most $(\gamma_p^{-1} l)^{1/p}$, they are pairwise independent modulo $\rho > 2l$. Therefore, instantiating Lemma 5.8 from \cite{Khot05} with the lower bound on the number of good vectors $N_g$, and our choice of $\rho$, it follows that with probability at least $0.9$, there exists a good vector $\bu$ such that $\br^T\bu \equiv 0 \Mod \rho$, i.e., $\bu \circ 0$ belongs to $\cL(\bB_{\rm svp})$. Therefore, by union bound, with probability at least $0.8$ (over the randomness of Lemma \ref{lem:good-vectors} and the choice of $\mathbf{r}$), there exists a good $\bu \in  \cL(\bB_{\rm int})$ such that $\bu \circ 0$ remains in $\cL(\bB_{\rm svp})$, which concludes the proof.
\end{proof}

%%% Local Variables:
%%% mode: latex
%%% TeX-master: "JACM submission"
%%% End:

\section{Conclusion and Open Questions}
\label{sec:open}
In this work, we have shown the parameterized inapproximability of $k$-Minimum Distance Problem ($k$-\MDP) and $k$-Shortest Vector Problem ($k$-\SVP) in the $\ell_p$ norm for every $p > 1$ assuming $\W[1]\neq\FPT$ (and under randomized reductions).  

An immediate open question stemming from our work is whether $k$-\SVP\ in the $\ell_1$ norm is in FPT. Khot's reduction unfortunately does not work for $\ell_1$; indeed, in the work of Haviv and Regev~\cite{HR07}, they arrive at the hardness of approximating \SVP\ in the $\ell_1$ norm by embedding \SVP\ instances in $\ell_2$ to instances in $\ell_1$ using an earlier result of Regev and Rosen~\cite{RR06}. The Regev-Rosen embedding inherently does not work in the FPT regime either, as it produces non-integral lattices. Similar issue applies to an earlier hardness result for \SVP\ on $\ell_1$ of~\cite{Mic00}, whose reduction produces irrational bases.

An additional question regarding $k$-\SVP\ is whether we can prove hardness of approximation for \emph{every} constant factor for $p \ne 2$. We note here that for $p=2$,  
we can use the tensor product of lattices to amplify the gap,  as Khot's construction is tailored so that the resulting lattice is
``well-behaved'' under tensoring, and gap amplification is indeed possible for such
instances. However, if $p\neq 2$ then the gap amplification techniques of~\cite{Khot05,HR07} require the distance $k$ to be dependent on the input size $nm$, and hence are not applicable for us. To the best of our knowledge, it is unknown whether this dependency is necessary. If they are indeed required, it would also be interesting to see whether other different techniques that work for our settings can be utilized for gap amplification instead of those from~\cite{Khot05,HR07}.

Furthermore, the Minimum Distance Problem can be defined for linear codes in $\mathbb{F}_p$ for any larger field of size $p > 2$ as well. It turns out that our result does not rule out FPT algorithms for $k$-\MDP\ over $\mathbb{F}_p$ with $p > 2$, when $p$ is fixed and is not part on the input. The issue here is that, in our proof of existence of Locally Suffix Dense Codes (Lemma~\ref{lem:ldc}), we need the co-dimension of the code to be small compared to its distance. In particular, the co-dimension $h - m$ has to be at most $(d/2 + O(1))\log_p h$ where $d$ is the distance. While the BCH code over binary alphabet satisfies this property, we are not aware of any linear codes that satisfy this for larger fields. It is an intriguing open question to determine whether such codes exist, or whether the reduction can be made to work without existence of such codes.

Since the current reductions for both $k$-\MDP\ and $k$-\SVP\ are randomized, it is still an intriguing open question whether we can find deterministic reductions for these problems. As stated in the introduction, even in the non-parameterized setting, \NP-hardness of \SVP\ through deterministic reductions is not known. On the other hand, \MDP\ is known to be \NP-hard even to approximate under deterministic reductions; in fact, even the Dumer~\etal's reduction~\cite{DMS03} that we employ can be derandomized, as long as one has a deterministic construction for Locally Dense Codes~\cite{CW12,Mic14}. In our settings, if one can deterministically construct Locally Suffix Dense Codes (i.e. derandomize Lemma~\ref{lem:ldc}), then we would also get a deterministic reduction for $k$-\MDP.

\subsection*{Acknowledgements}

We are grateful to Ishay Haviv for providing insights on how the gap amplification for $p \ne 2$ from~\cite{HR07} works. Pasin would like to thank Danupon Nanongkai for introducing him to the $k$-Even Set problem and for subsequent useful discussions.

Arnab Bhattacharyya contributed to this work while he was at the Indian Institute of Science and was supported by Ramanujan Fellowship DSTO 1358 and the Indo-US Joint Center for Pseudorandomness in Computer Science. \'Edouard Bonnet, L\'{a}szl\'{o} Egri, Bingkai Lin, and D\'{a}niel Marx are supported by the European Research Council (ERC) consolidator grant No. 725978 SYSTEMATICGRAPH.  L\'{a}szl\'{o} Egri is also supported by NSERC. Karthik C.\ S.\ is supported by Irit Dinur's ERC-CoG grant 772839. Bingkai Lin is also supported by JSPS KAKENHI Grant (JP16H07409) and the JST ERATO Grant (JPMJER1201) of Japan. Pasin Manurangsi is supported by the Indo-US Joint Center for Pseudorandomness in Computer Science. 

\bibliographystyle{alpha}
\bibliography{main}

\newcommand{\etalchar}[1]{$^{#1}$}
\begin{thebibliography}{DFVW99}

\bibitem[ABSS97]{ABSS97}
Sanjeev Arora, L{\'{a}}szl{\'{o}} Babai, Jacques Stern, and Z.~Sweedyk.
\newblock The hardness of approximate optima in lattices, codes, and systems of
  linear equations.
\newblock {\em J. Comput. Syst. Sci.}, 54(2):317--331, 1997.

\bibitem[AD97]{AD97}
Mikl{\'{o}}s Ajtai and Cynthia Dwork.
\newblock A public-key cryptosystem with worst-case/average-case equivalence.
\newblock In {\em STOC}, pages 284--293, 1997.

\bibitem[Ajt96]{Ajt96}
Mikl{\'{o}}s Ajtai.
\newblock Generating hard instances of lattice problems (extended abstract).
\newblock In {\em STOC}, pages 99--108, 1996.

\bibitem[Ajt98]{Ajt98}
Mikl{\'{o}}s Ajtai.
\newblock The shortest vector problem in $\ell_2$ is {NP}-hard for randomized
  reductions (extended abstract).
\newblock In {\em STOC}, pages 10--19, 1998.

\bibitem[AK14]{AK14}
Per Austrin and Subhash Khot.
\newblock A simple deterministic reduction for the gap minimum distance of code
  problem.
\newblock {\em {IEEE} Trans. Information Theory}, 60(10):6636--6645, 2014.

\bibitem[AS18]{AD17}
Divesh Aggarwal and Noah Stephens{-}Davidowitz.
\newblock (gap/s)eth hardness of {SVP}.
\newblock In {\em STOC}, pages 228--238, 2018.

\bibitem[BGGS16]{BGGS16}
Arnab Bhattacharyya, Ameet Gadekar, Suprovat Ghoshal, and Rishi Saket.
\newblock On the hardness of learning sparse parities.
\newblock In {\em ESA}, pages 11:1--11:17, 2016.

\bibitem[BGS17]{BGS17}
Huck Bennett, Alexander Golovnev, and Noah Stephens{-}Davidowitz.
\newblock On the quantitative hardness of {CVP}.
\newblock In {\em FOCS}, pages 13--24, 2017.

\bibitem[BMvT78]{BMT78}
Elwyn~R. Berlekamp, Robert~J. McEliece, and Henk C.~A. van Tilborg.
\newblock On the inherent intractability of certain coding problems (corresp.).
\newblock {\em {IEEE} Trans. Information Theory}, 24(3):384--386, 1978.

\bibitem[BR60]{BR60}
R.~C. Bose and Dwijendra~K. Ray{-}Chaudhuri.
\newblock On a class of error correcting binary group codes.
\newblock {\em Information and Control}, 3(1):68--79, 1960.

\bibitem[CFHW17]{CFHW17}
Marek Cygan, Fedor~V. Fomin, Danny Hermelin, and Magnus Wahlstr{\"{o}}m.
\newblock Randomization in parameterized complexity (dagstuhl seminar 17041).
\newblock {\em Dagstuhl Reports}, 7(1):103--128, 2017.

\bibitem[CFJ{\etalchar{+}}14]{CFJKLMPS14}
Marek Cygan, Fedor Fomin, Bart~MP Jansen, Lukasz Kowalik, Daniel Lokshtanov,
  D{\'a}niel Marx, Marcin Pilipczuk, and Saket Saurabh.
\newblock Open problems for fpt school 2014.
\newblock 2014.

\bibitem[CFK{\etalchar{+}}15]{CFKLMPPS15}
Marek Cygan, Fedor~V. Fomin, Lukasz Kowalik, Daniel Lokshtanov, D{\'{a}}niel
  Marx, Marcin Pilipczuk, Michal Pilipczuk, and Saket Saurabh.
\newblock {\em Parameterized Algorithms}.
\newblock Springer, 2015.

\bibitem[CN99]{CN99}
Jin{-}yi Cai and Ajay Nerurkar.
\newblock Approximating the {SVP} to within a factor $(1+1/\text{dim}^{\xi})$
  is {NP}-hard under randomized reductions.
\newblock {\em J. Comput. Syst. Sci.}, 59(2):221--239, 1999.

\bibitem[CW12]{CW12}
Qi~Cheng and Daqing Wan.
\newblock A deterministic reduction for the gap minimum distance problem.
\newblock {\em {IEEE} Trans. Information Theory}, 58(11):6935--6941, 2012.

\bibitem[DF99]{DF99}
Rodney~G. Downey and Michael~R. Fellows.
\newblock {\em Parameterized Complexity}.
\newblock Monographs in Computer Science. Springer, 1999.

\bibitem[DF13]{DF13}
Rodney~G. Downey and Michael~R. Fellows.
\newblock {\em Fundamentals of Parameterized Complexity}.
\newblock Texts in Computer Science. Springer, 2013.

\bibitem[DFVW99]{DFVW99}
Rodney~G. Downey, Michael~R. Fellows, Alexander Vardy, and Geoff Whittle.
\newblock The parametrized complexity of some fundamental problems in coding
  theory.
\newblock {\em {SIAM} J. Comput.}, 29(2):545--570, 1999.

\bibitem[DGMS07]{DGMS07}
Erik~D. Demaine, Gregory Gutin, D{\'{a}}niel Marx, and Ulrike Stege.
\newblock 07281 open problems -- structure theory and {FPT} algorithmcs for
  graphs, digraphs and hypergraphs.
\newblock In {\em Structure Theory and {FPT} Algorithmics for Graphs, Digraphs
  and Hypergraphs, 08.07. - 13.07.2007}, 2007.

\bibitem[Din02]{Din02}
Irit Dinur.
\newblock Approximating $\text{SVP}_{\infty}$ to within almost-polynomial
  factors is {NP}-hard.
\newblock {\em Theor. Comput. Sci.}, 285(1):55--71, 2002.

\bibitem[Din16]{D16}
Irit Dinur.
\newblock Mildly exponential reduction from gap {3SAT} to polynomial-gap
  label-cover.
\newblock {\em ECCC}, 23:128, 2016.

\bibitem[DKRS03]{DKRS03}
Irit Dinur, Guy Kindler, Ran Raz, and Shmuel Safra.
\newblock Approximating {CVP} to within almost-polynomial factors is {NP}-hard.
\newblock {\em Combinatorica}, 23(2):205--243, 2003.

\bibitem[DMS03]{DMS03}
Ilya Dumer, Daniele Micciancio, and Madhu Sudan.
\newblock Hardness of approximating the minimum distance of a linear code.
\newblock {\em {IEEE} Trans. Information Theory}, 49(1):22--37, 2003.

\bibitem[FGMS12]{FGMS12}
Michael~R. Fellows, Jiong Guo, D{\'{a}}niel Marx, and Saket Saurabh.
\newblock Data reduction and problem kernels (dagstuhl seminar 12241).
\newblock {\em Dagstuhl Reports}, 2(6):26--50, 2012.

\bibitem[FM12]{DBLP:conf/birthday/FominM12}
Fedor~V. Fomin and D{\'{a}}niel Marx.
\newblock {FPT} suspects and tough customers: Open problems of downey and
  fellows.
\newblock In Hans~L. Bodlaender, Rod Downey, Fedor~V. Fomin, and D{\'{a}}niel
  Marx, editors, {\em The Multivariate Algorithmic Revolution and Beyond -
  Essays Dedicated to Michael R. Fellows on the Occasion of His 60th Birthday},
  volume 7370 of {\em Lecture Notes in Computer Science}, pages 457--468.
  Springer, 2012.

\bibitem[GKS12]{GKS12}
Petr~A. Golovach, Jan Kratochv{\'{\i}}l, and Ondrej Such{\'{y}}.
\newblock Parameterized complexity of generalized domination problems.
\newblock {\em Discrete Applied Mathematics}, 160(6):780--792, 2012.

\bibitem[GMSS99]{GMSS99}
Oded Goldreich, Daniele Micciancio, Shmuel Safra, and Jean{-}Pierre Seifert.
\newblock Approximating shortest lattice vectors is not harder than
  approximating closest lattice vectors.
\newblock {\em Inf. Process. Lett.}, 71(2):55--61, 1999.

\bibitem[Gol06]{Gol06}
Oded Goldreich.
\newblock On promise problems: {A} survey.
\newblock In {\em Theoretical Computer Science, Essays in Memory of Shimon
  Even}, pages 254--290, 2006.

\bibitem[Hoc59]{H59}
Alexis Hocquenghem.
\newblock Codes correcteurs d’erreurs.
\newblock {\em Chiffres}, 2:147–156, September 1959.

\bibitem[HR07]{HR07}
Ishay Haviv and Oded Regev.
\newblock Tensor-based hardness of the shortest vector problem to within almost
  polynomial factors.
\newblock In {\em STOC}, pages 469--477, 2007.

\bibitem[Joh90]{J90}
D.~S. Johnson.
\newblock Handbook of theoretical computer science.
\newblock volume A (Algorithms and Complexity), chapter 2, A catalog of
  complexity classes, pages 67--161. Elseveir, 1990.

\bibitem[Kho05]{Khot05}
Subhash Khot.
\newblock Hardness of approximating the shortest vector problem in lattices.
\newblock {\em J. {ACM}}, 52(5):789--808, 2005.

\bibitem[KLM19]{KLM18}
{Karthik {C. S.}}, Bundit Laekhanukit, and Pasin Manurangsi.
\newblock On the parameterized complexity of approximating dominating set.
\newblock {\em J. {ACM}}, 66(5):33:1--33:38, 2019.

\bibitem[Len83]{Len83}
Hendrik~Willem Lenstra.
\newblock Integer programming with a fixed number of variables.
\newblock {\em Math. Oper. Res.}, 8(4):538--548, 1983.

\bibitem[Lin18]{Lin15}
Bingkai Lin.
\newblock The parameterized complexity of the \emph{k}-biclique problem.
\newblock {\em J. {ACM}}, 65(5):34:1--34:23, 2018.

\bibitem[LLL82]{LLL}
Arjen~Klaas Lenstra, Hendrik~Willem Lenstra, and L{\'a}szl{\'o} Lov{\'a}sz.
\newblock Factoring polynomials with rational coefficients.
\newblock {\em Mathematische Annalen}, 261(4):515--534, 1982.

\bibitem[Maj17]{Maj17}
Ruhollah Majdoddin.
\newblock Parameterized complexity of {CSP} for infinite constraint languages.
\newblock {\em CoRR}, abs/1706.10153, 2017.

\bibitem[MG12]{Micc12}
Daniele Micciancio and Shafi Goldwasser.
\newblock {\em Complexity of lattice problems: a cryptographic perspective},
  volume 671.
\newblock Springer Science \& Business Media, 2012.

\bibitem[Mic00]{Mic00}
Daniele Micciancio.
\newblock The shortest vector in a lattice is hard to approximate to within
  some constant.
\newblock {\em {SIAM} J. Comput.}, 30(6):2008--2035, 2000.

\bibitem[Mic01]{Mic01}
Daniele Micciancio.
\newblock The hardness of the closest vector problem with preprocessing.
\newblock {\em {IEEE} Trans. Information Theory}, 47(3):1212--1215, 2001.

\bibitem[Mic12]{Mic12}
Daniele Micciancio.
\newblock Inapproximability of the shortest vector problem: Toward a
  deterministic reduction.
\newblock {\em Theory of Computing}, 8(1):487--512, 2012.

\bibitem[Mic14]{Mic14}
Daniele Micciancio.
\newblock Locally dense codes.
\newblock In {\em CCC}, pages 90--97, 2014.

\bibitem[MR09]{MR-survey}
Daniele Micciancio and Oded Regev.
\newblock Lattice-based cryptography.
\newblock In {\em Post-quantum cryptography}, pages 147--191. Springer, 2009.

\bibitem[MR16]{MR16}
Pasin Manurangsi and Prasad Raghavendra.
\newblock A birthday repetition theorem and complexity of approximating dense
  {CSP}s.
\newblock {\em CoRR}, abs/1607.02986, 2016.

\bibitem[NSS95]{NSS95}
Moni Naor, Leonard~J. Schulman, and Aravind Srinivasan.
\newblock Splitters and near-optimal derandomization.
\newblock In {\em 36th Annual Symposium on Foundations of Computer Science,
  Milwaukee, Wisconsin, USA, 23-25 October 1995}, pages 182--191, 1995.

\bibitem[NV10]{NV10}
Phong~Q. Nguyen and Brigitte Vall{\'{e}}e, editors.
\newblock {\em The {LLL} Algorithm - Survey and Applications}.
\newblock Information Security and Cryptography. Springer, 2010.

\bibitem[Reg03]{Reg03}
Oded Regev.
\newblock New lattice based cryptographic constructions.
\newblock In {\em STOC}, pages 407--416, 2003.

\bibitem[Reg05]{Reg05}
Oded Regev.
\newblock On lattices, learning with errors, random linear codes, and
  cryptography.
\newblock In {\em STOC}, pages 84--93, 2005.

\bibitem[Reg06]{Reg06}
Oded Regev.
\newblock Lattice-based cryptography.
\newblock In {\em CRYPTO}, pages 131--141, 2006.

\bibitem[Reg10]{Reg10}
Oded Regev.
\newblock The learning with errors problem (invited survey).
\newblock In {\em CCC}, pages 191--204, 2010.

\bibitem[RR06]{RR06}
Oded Regev and Ricky Rosen.
\newblock Lattice problems and norm embeddings.
\newblock In {\em STOC}, pages 447--456, 2006.

\bibitem[Ste93]{Ste93}
Jacques Stern.
\newblock Approximating the number of error locations within a constant ratio
  is {NP}-complete.
\newblock In {\em AAECC}, pages 325--331, 1993.

\bibitem[Var97a]{Var97a}
Alexander Vardy.
\newblock Algorithmic complexity in coding theory and the minimum distance
  problem.
\newblock In {\em STOC}, pages 92--109, 1997.

\bibitem[Var97b]{Var97b}
Alexander Vardy.
\newblock The intractability of computing the minimum distance of a code.
\newblock {\em {IEEE} Trans. Information Theory}, 43(6):1757--1766, 1997.

\bibitem[vEB81]{VEB}
Peter van Emde-Boas.
\newblock {\em Another {NP}-complete partition problem and the complexity of
  computing short vectors in a lattice}.
\newblock Report. Department of Mathematics. University of Amsterdam.
  Department, Univ., 1981.

\end{thebibliography}

\appendix
\section{Inapproximability of Odd Set}\label{app:oddset}

In this section, we show how the hardness for the more general \gapvec, also implies hardness for the Odd Set problem, which can be defined in a similar manner as \gapvec\  except that $\by$ is always fixed as the all-ones vector instead of being part of the input. More formally, we define the gap version of Odd Set below.

\begin{framed}
$\gamma$-Gap Odd Set Problem ($\gapodds_{\gamma}$)

{\bf Input: } A matrix $\bA \in \mathbb{F}_2^{n \times m}$ and a positive integer $k \in \mathbb{N}$

{\bf Parameter: } $k$

{\bf Question: } Distinguish between the following two cases:
\begin{itemize}
\item (YES) there exists $\bx \in \cB_{m}(\mathbf{0}, k)$ such that $\bA\bx = \vone$
\item (NO) for all $\bx \in \cB_{m}(\mathbf{0}, \gamma k)$, $\bA\bx \ne \vone$\end{itemize}
\end{framed}

It is obvious that hardness for \gapodds gives the hardness for \gapvec, by simply setting $\by = \vone$. Below we show that the opposite implication is also true; note that, together with Theorem~\ref{thm:MLDmain}, it implies that $\gapodds_\gamma$ is \W[1]-hard for every $\gamma \geqs 1$.

\begin{proposition}\label{prop:mld-to-odd}
For every $\gamma' > \gamma \geqs 1$, there is an FPT reduction from $\gapvec_{\gamma'}$ to $\gapodds_\gamma$ 
\end{proposition}

\begin{proof}
Let $(\bA, \by, k)$ be an instance of $\gapvec_{\gamma'}$ where $\bA \in \F^{n \times m}$ and $\by \in \F^n$. We may assume without loss of generality that $k > \frac{\gamma}{\gamma' - \gamma}$. The instance $(\bA' \in \F^{(n + 1) \times (m + 1)}, k')$ of $\gapodds_{\gamma}$ is defined as follows. First, we let $k' = k + 1$. Then, for $i \in [m]$, we let the $i$-th column of $\bA'$ be the $i$-th column concatenated with zero, and we let the $(m + 1)$-th column be $(\vone_n + \by) \circ 1$. That is,
\begin{align*}
\bA' = 
\begin{bmatrix}
\bA & \vone_n + \by \\
\vzero_{1 \times m} & 1
\end{bmatrix}.
\end{align*}
Clearly, the reduction runs in polynomial time. We next argue its correctness.

\textbf{(YES Case)} Suppose that $(\bA, \by, k)$ is a YES instance of $\gapvec_{\gamma'}$, i.e., there exists $\bx \in \F^n$ with $\|\bx\|_0 \leqs k$ such that $\bA\bx = \by$. Let $\bx' = \bx \circ 1$; it is simple to see that $\bA'\bx' = \vone_{m+1}$ and that $\|\bx'\|_0 \leqs k + 1 = k'$. Hence, $(\bA', k')$ is a YES instance of $\gapodds_{\gamma}$.

\textbf{(NO Case)} Suppose that $(\bA, \by, k)$ is NO instance of $\gapvec_{\gamma'}$. Now, let us consider any $\bx' \in \F^{n + 1}$ such that $\bA' \bx' = \vone$. Notice that $(\bA' \bx')[m + 1] = \bx'[n + 1]$, which implies that $\bx'[n + 1] = 1$. 

Now, consider the vector $\bx = (\bx[1], \dots, \bx[n])$; it is easy to verify that $\bA\bx = \by$. Since $(\bA, \by, k)$ is NO instance of $\gapvec_{\gamma'}$, we have $\|\bx\|_0 > \gamma' k$. As a result, $\|\bx'\|_0 = 1 + \|\bx\|_0 > 1 + \gamma' k > \gamma \cdot k'$, where the last inequality follows from $k \geqs \frac{\gamma}{\gamma' - \gamma}$. Thus, $(\bA', k')$ is indeed a NO instance of $\gapodds_{\gamma}$.
\end{proof}

%%% Local Variables:
%%% mode: latex
%%% TeX-master: "JACM submission"
%%% End:

\end{document}